\documentclass[11pt]{amsart}

\pdfoutput=1

\usepackage[text={420pt,660pt},centering]{geometry}

\usepackage{cancel}
\usepackage{esint,amssymb}
\usepackage{graphicx}
\usepackage{mathtools} 
\usepackage[colorlinks=true, pdfstartview=FitV, linkcolor=blue, citecolor=blue, urlcolor=blue,pagebackref=false]{hyperref}
\usepackage{microtype}

\usepackage{bm}
\usepackage{dsfont}
\usepackage{mathrsfs}
\usepackage{xcolor}
\usepackage{accents} \usepackage{enumitem} \parskip= 2pt
\usepackage[normalem]{ulem}

\newtheorem{proposition}{Proposition}
\newtheorem{theorem}[proposition]{Theorem}
\newtheorem{lemma}[proposition]{Lemma}
\newtheorem{corollary}[proposition]{Corollary}
\newtheorem{claim}[proposition]{Claim}

\theoremstyle{remark}
\newtheorem{remark}[proposition]{Remark}

\newtheorem*{example*}{Example}
\theoremstyle{definition}

\numberwithin{equation}{section}
\numberwithin{proposition}{section}
\numberwithin{figure}{section}
\numberwithin{table}{section}

\newcommand{\Z}{\mathbb{Z}}
\newcommand{\N}{\mathbb{N}}
\newcommand{\Q}{\mathbb{Q}}
\newcommand{\R}{\mathbb{R}}

\newcommand{\E}{\mathbb{E}}

\renewcommand{\S}{{S}}

\newcommand{\eps}{\varepsilon}

\renewcommand{\le}{\leqslant}
\renewcommand{\ge}{\geqslant}
\renewcommand{\leq}{\leqslant}
\renewcommand{\geq}{\geqslant}

\renewcommand{\subset}{\subseteq}
\renewcommand{\bar}{\overline}
\renewcommand{\tilde}{\widetilde}

\renewcommand{\hat}{\widehat}

\newcommand{\Ll}{\left}
\newcommand{\Rr}{\right}
\renewcommand{\d}{\mathrm{d}}
\newcommand{\dr}{\partial}

\newcommand{\D}{D}

\newcommand{\mcl}{\mathcal}

\newcommand{\si}{\sigma}

\DeclareMathOperator{\supp}{supp}
\DeclareMathOperator*{\esssup}{ess\,sup}

\newcommand{\la}{\left\langle}
\newcommand{\ra}{\right\rangle}

\newcommand{\cH}{{L^2}}

\newcommand{\pert}{{\mathsf{pert}}}

\newcommand{\identity}{{\mathbf{Id}}}
\newcommand{\one}{{\boldsymbol{1}}}

\newcommand{\fR}{{\mathfrak{R}}}
\newcommand{\bsigma}{{\boldsymbol{\sigma}}}

\newcommand{\sP}{\mathscr{P}}

\renewcommand*{\dot}[1]{\accentset{\mbox{\large\bfseries .}}{#1}}
\newcommand{\ellipt}{\mathsf{Ellipt}}
\newcommand{\sS}{\mathscr{S}}
\renewcommand{\vec}{{\mathrm{vec}}}
\newcommand{\bxi}{\boldsymbol{\xi}}
\newcommand{\bq}{{\mathbf{q}}}
\newcommand{\bp}{{\mathbf{p}}}
\newcommand{\ba}{{\mathbf{a}}}
\newcommand{\bw}{\mathbf{w}}
\newcommand{\M}{{M}} \newcommand{\sfM}{\mathsf{M}}

\newcommand{\btau}{{\boldsymbol{\tau}}}
\newcommand{\seq}{{\mathsf{seq}}}

\begin{document}

\author{Hong-Bin Chen}
\address{Institut des Hautes \'Etudes Scientifiques, France}
\email{hbchen@ihes.fr}

\title{On free energy of non-convex multi-species spin glasses}

\begin{abstract}
In~\cite{HJ_critical_pts}, it was shown that if the limit of the free energy in a non-convex vector spin glass model exists, it must be a critical value of a certain functional. In this work, we extend this result to multi-species spin glass models with non-convex interactions, where spins from different species may lie in distinct vector spaces. Since the species proportions may be irrational and the existence of the limit of the free energy is not generally known, non-convex multi-species models cannot be approximated by vector spin models in a straightforward manner, necessitating more careful treatment.
\end{abstract}

\maketitle

\section{Introduction}

In mean-field vector spin glass models with general non-convex interactions, the main result in~\cite{HJ_critical_pts} shows that the limit of free energy, if exists, must be a critical value of some functional.
Moreover, up to a small perturbation, the limit distribution of the spin overlap is determined by the critical point.
In this work, we extend these results to the setting of multi-species spin glasses, where spins are grouped into different species and interactions are between species.
We consider the general setting where spins from different species can take values in different vector spaces. 
Hence, the model considered here can be described as a multi-species vector spin glass model with possibly non-convex interactions. Main results corresponding to those in~\cite{HJ_critical_pts} are stated as Theorems in Section~\ref{s.main_results} and some less important results from~\cite{HJ_critical_pts} are adapted and collected in Section~\ref{s.crti_pt_and_rel_results}.

We clarify the necessity for this separate work.
If species proportions are rational, then the multi-species model can be reduced to a vector spin glass model, to which results in~\cite{HJ_critical_pts} are directly applicable. If not all proportions are rational, it is tempting to simply take a sequence of models with rational proportions to approximate the non-rational one. This is indeed the argument used by Panchenko in~\cite{pan.multi} to treat the multi-species Sherrington--Kirkpatrick model (see the second paragraph in~\cite[Section~5]{pan.multi}). However, this argument relies on that the limit of free energy exists for each approximation model, which is not known if the interaction is non-convex. Therefore, results for non-convex multi-species models are not direct corollaries from those~\cite{HJ_critical_pts} for vector spin glasses. This justifies our pursuits here. 

Moreover, a direct motivation is from \cite{simul_rsb} which studies the simultaneous replica symmetry breaking (RSB) in vector and multi-species spin glasses. The relation satisfied by critical points is the key structure for the simultaneous RSB. This work supplies the relation in the multi-species setting and thus together with~\cite{HJ_critical_pts} forms the groundwork for~\cite{simul_rsb}.

Aside from results in the general non-convex setting, we include results in the convex case in Section~\ref{s.convex}.
In particular, we prove the Parisi formula (Proposition~\ref{p.mp-parisi}) for convex multi-species vector spin glass models; and we adapt the main results for non-convex models to the convex setting and summarize them in Corollary~\ref{c.convex}.

\subsection{Setting}

Before describing the multi-species model, we start with some basic notation used throughout this work. 

\subsubsection{Matrices, inner products, and paths} \label{s.matrices,inner_products,paths}

Throughout, we write $R_+=[0,+\infty)$.

For $m,n\in\N$, we denote by $\R^{m\times n}$ the space of all $m\times n$ real matrices.
For any $a=(a_{ij})_{1\leq i\leq m,\, 1\leq j\leq n}\in \R^{m\times n}$, we denote its $j$-th column vector as $a_{\bullet j}=(a_{ij})_{1\leq i\leq m}$ and its $i$-th row vector as $a_{i\bullet}=(a_{ij})_{1\leq j\leq n}$. If not specified, a vector is understood to be a column vector. For a matrix or vector $a$, we denote by $a^\intercal$ its transpose.
For $a,b\in \R^{m\times n}$, we write $a\cdot b= \sum_{ij}a_{ij}b_{ij}$, $|a|=\sqrt{a\cdot a}$, and similarly for vectors. More generally, for any finite set $I$ and $a,b\in \R^I$, we write
\begin{align}\label{e.dot_product}
    a\cdot b=\sum_{i\in I}a_ib_i; \qquad |a|=\sqrt{a\cdot a}.
\end{align}

For each $\D\in \N$, let $\S^\D$ be the linear space of $\D\times\D$ real symmetric matrices. For $a,b\in\S^\D$, viewing them as elements in $\R^{\D\times \D}$, we write $a\cdot b= \sum_{ij}a_{ij}b_{ij}$ and $|a|=\sqrt{a\cdot a}$. Let $\S^\D_+$ (resp.\ $\S^\D_{++}$) be the subset consisting of positive semi-definite (resp.\ definite) matrices. For $a,b \in \S^\D_+$, we write $a\geq b$ provided $a-b\in\S^\D_+$, which gives a natural partial order on $\S^\D_+$.
In particular, when $\D=1$, we have $\S^\D_+=\R_+$.

Let $\mcl Q(\D)$ be the collection of right-continuous increasing paths $q:[0,1)\to\S^\D_+$. Here, $q$ is said to be increasing provided that $q(r')\geq q(r)$ in $\S^\D_+$ for every $0\leq r\leq r'<1$. For $p\in [1,\infty)$, we write $\mcl Q_p(\D) = \mcl Q(\D) \cap L^p([0,1);\S^\D)$. For $q\in \mcl Q_\infty(\D)$, we set $q(1) = \lim_{s\nearrow 1} q(s)$ which exists by monotonicity.

\subsubsection{Species and proportions}

Fix a finite set $\sS$, the elements of which are interpreted as symbols for different species.
For each $N$, let $I_{N,s}$, for $s\in\sS$, be a partition of $\{1,\dots,N\}$. We interpret each $I_{N,s}$ as the set of indices for the $s$-species. 

For each $N$, we set
\begin{align}\label{e.lambda_N,s=}
    \lambda_{N,s} = |I_{N,s}|/N,\quad\forall s\in\sS;\qquad \lambda_N=\Ll(\lambda_{N,s}\Rr)_{s\in\sS}.
\end{align}
We interpret $\lambda_{N,s}$ as the proportion of the $s$-species at size $N$.
Introducing the notation for the discrete simplex
\begin{align}\label{e.simplex}
    \blacktriangle_N = \Big\{(\lambda_s)_{s\in\sS}\ \big|\ \lambda_s \in [0,1]\cap (\Z/N),\,\forall s\in\sS;\; \sum_{s\in\sS}\lambda_s =1\Big\}
\end{align}
we have $\lambda_N\in \blacktriangle_N$. We often assume that $(\lambda_N)_{N\in\N}$ converges to some $\lambda_\infty$ belonging to the continuous simplex
\begin{align}\label{e.cts_simplex}
    \blacktriangle_\infty = \Big\{(\lambda_s)_{s\in\sS}\ \big|\ \lambda_s \in [0,1],\,\forall s\in\sS;\; \sum_{s\in\sS}\lambda_s =1\Big\}.
\end{align}

\subsubsection{Spins, overlaps, and interaction}\label{s.spin}
We fix some $\kappa_s\in\N$ for each $s\in\sS$ we write $\kappa=(\kappa_s)_{s\in\sS}$. We assume that spins in the $s$-species are vectors in $\R^{\kappa_s}$. 
Denote by 
\begin{align}\label{e.Sigma}
    \Sigma=\sqcup_{s\in\sS} [-1,+1]^{\kappa_s}
\end{align}
the state space for a single spin (here $\sqcup$ stands for the disjoint union). 
For $s\in\sS$, let $\mu_s$ be a finite nonnegative measure supported on $[-1,+1]^{\kappa_s}$. We extend $\mu_s$ trivially to $\Sigma$. 
For each $N\in\N$, a spin configuration consisting of $N$ spins is denoted by $\sigma = (\sigma_{\bullet 1},\dots , \sigma_{\bullet N}) \in \Sigma^N$ where each $\sigma_{\bullet n} = (\sigma_{kn})_{1\leq k\leq \kappa_s}$ is a (column) vector in $\R^{\kappa_s}$ if $n\in I_{N,s}$.
We sample a configuration $\sigma = (\sigma_{\bullet 1},\dots \sigma_{\bullet N})$ by independently drawing each $\sigma_{\bullet n}$ according to $\mu_s$ if $n\in I_{N,s}$.
In other words, denoting $P_{N,\lambda_N}$ as the distribution of $\sigma$, we have
\begin{align}\label{e.P_N=}
    \d P_{N,\lambda_N}(\sigma) = \otimes_{s\in\sS}\otimes_{n \in I_{N,s}} \d\mu_s(\sigma_{\bullet n}).
\end{align}

For two spin configurations $\sigma,\sigma'$ of size $N$ and $s\in\sS$, we consider the $\R^{\kappa_s\times\kappa_s}$-valued overlap of the $s$-species:
\begin{align}\label{e.R_N,s=}
    R_{N,\lambda_N,s}(\sigma,\sigma') = \frac{1}{N}\sigma_{\bullet I_{N,s}} (\sigma'_{\bullet I_{N,s}})^\intercal \in [-1,+1]^{\kappa_s\times\kappa_s}
\end{align}
where
\begin{align}\label{e.sigma_bullet,I}
    \sigma_{\bullet I_{N,s}} = (\sigma_{kn})_{1\leq k\leq \kappa_s,\, n\in I_{N,s}}
\end{align}
is a $\kappa_s\times |I_{N,s}|$-matrix and similarly for $\sigma'_{\bullet I_{N,s}}$. Putting them together, we consider the $\prod_{s\in\sS}\R^{\kappa_s\times \kappa_s}$-valued overlap:
\begin{align*}
    R_{N,\lambda_N}(\sigma,\sigma') = \Ll(R_{N,\lambda_N,s}(\sigma,\sigma')\Rr)_{s\in\sS}.
\end{align*}
Notice that the overlap depends on the partition $(I_{N,s})_{s\in\sS}$. Here, we choose to only display the dependence on $\lambda_N$ because the distribution of the overlap under Gibbs measures to be introduced only depends on the proportion.

Let $\xi:\prod_{s\in\sS}\R^{\kappa_s\times \kappa_s}\to\R$ be a smooth function and assume the existence of a centered Gaussian process $(H_N(\sigma))_{\sigma\in \Sigma^N}$ with covariance
\begin{align}\label{e.EH_N(sigma)H_N(sigma')=}
    \E\Ll[ H_N(\sigma)H_N(\sigma')\Rr] = N\xi\Ll(R_{N,\lambda_N}(\sigma,\sigma')\Rr).
\end{align}
The form of $\xi$ and the construction of $H_N(\sigma)$ are similar to those in the vector spin case. We refer to~\cite[Section~1.5]{HJ_critical_pts} for the detail.

\begin{example*}[Multi-species Sherrington--Kirkpatrick model \cite{sherrington1975solvable}]
The model considered in~\cite{pan.multi} corresponds to the case where $\kappa_s=1$, $\mu_s = \delta_{-1}+\delta_{+1}$ for every $s\in\sS$, and $\xi$ is given by $\xi(a) =\sum_{s,s'\in\sS}\Delta^2_{ss'}a_sa_{s'}$ where $\Delta^2 = (\Delta^2_{ss'})_{s,s'\in\sS}$ is a matrix in $\S^{|\sS|}_+$.
\end{example*}

\subsubsection{Cascade}\label{s.cascade}
To describe the external fields in the model, we need first to recall the Ruelle probability cascade (RPC).
Let $\fR$ denote the RPC with overlap uniformly distributed over $[0,1]$ (see \cite[Theorem~2.17]{pan}). Precisely, $\fR$ is a random probability measure on the unit sphere of a separable Hilbert space, with the inner product denoted by $\alpha\wedge\alpha'$. Let $\alpha$ and $\alpha'$ be independent samples from $\fR$. Then, law of $\alpha\wedge\alpha'$ under $\E \fR^{\otimes 2}$ is the uniform distribution over $[0,1]$, where $\E$ integrates the randomness of $\fR$.
This overlap distribution uniquely determines $\fR$ (see~\cite[Theorem~2.13]{pan}).
Almost surely, the support of $\fR$ is ultrametric in the induced topology. For rigorous definitions and basic properties, we refer to \cite[Chapter2]{pan} (also see \cite[Chapter5]{HJbook}).
We also refer to~\cite[Section~4]{HJ_critical_pts} for the construction and properties of $\fR$ useful in this work.
Throughout, we denote by $\la\cdot\ra_\fR = \fR^{\otimes \N}$ the tensorized version of $\fR$.

\subsubsection{External fields}

Recall the space $\mcl Q(\D)$ of right-continuous increasing paths from Section~\ref{s.matrices,inner_products,paths}.
We use the following shorthand notation:
\begin{align}\label{e.Q^S(kappa)=}
    \mcl Q^\sS_\square (\kappa) = \prod_{s\in\sS} \mcl Q_\square(\kappa_s)
\end{align}
where $\square$ is a placeholder for subscripts.
We explicitly construct the external field parametrized by any $q =(q_s)_{s\in\sS} \in \mcl Q^\sS_\infty(\kappa)$.
For almost every realization of $\fR$, every $s\in\sS$, and every $n\in I_{N,s}$, let $(w^{q_s}_n(\alpha))_{\alpha\in\supp\fR}$ be the $\R^{\kappa_s}$-valued centered Gaussian process with covariance
\begin{align}\label{e.Ew^q_s_iw^q_s_i=}
    \E\Ll[ w^{q_s}_n (\alpha)w^{q_s}_n(\alpha')^\intercal\Rr]  = q_s(\alpha\wedge\alpha').
\end{align}
The existence of such a process and its properties are given in~\cite[Section~4]{HJ_critical_pts}.
Conditioned on $\fR$, we assume that all these processes, indexed by $s$ and $n$, are independent. For each $s$, we write $w^{q_s}_{I_{N_s}} = \Ll(w^{q_s}_n\Rr)_{n\in I_{N,s}}$. Recall the notation in~\eqref{e.sigma_bullet,I}. For each $N\in\N$ and $q\in \mcl Q^\sS_\infty(\kappa)$, we define
\begin{align}\label{e.W^q_N(sigma,alpha)=}
    W^q_N(\sigma,\alpha) = \sum_{s\in \sS}w^{q_s}_{I_{N,s}}(\alpha)\cdot \sigma_{\bullet I_{N,s}}
\end{align}
which, conditioned on $\fR$, is a centered Gaussian process with covariance
\begin{align}\label{e.cov_external_multi-sp}
    \E \Ll[ W^q_N(\sigma,\alpha) W^q_N(\sigma',\alpha')\Rr]\stackrel{\eqref{e.R_N,s=},\eqref{e.Ew^q_s_iw^q_s_i=}}{=} N q(\alpha\wedge\alpha')\cdot R_{N,\lambda_N}(\sigma,\sigma')
\end{align}
where the dot product follows the rule as in~\eqref{e.dot_product}.

\subsubsection{Hamiltonian and free energy}
Now, for $N\in \N$, $\lambda_N\in \blacktriangle_N$, $t\in\R_+$, and $q \in \mcl Q_\infty^\sS(\kappa)$, we consider the Hamiltonian
\begin{align}\label{e.H^t,q_N=}
    H^{t,q}_N(\sigma,\alpha)= \sqrt{2t}H_N(\sigma) - t N \xi \Ll(R_{N,\lambda_N}(\sigma,\sigma)\Rr)  
    + \sqrt{2}W^q_N(\sigma,\alpha) - Nq(1)\cdot R_{N,\lambda_N}(\sigma,\sigma)
\end{align}
where $q(1)= (q_s(1))_{s\in\sS}\in \prod_{s\in\sS}\S^{\kappa_s}_+$ and we understand $q(1)\cdot R_{N,\lambda_N}(\sigma,\sigma) = \sum_{s\in\mathscr{\sS}}q_s(1)\cdot R_{N,\lambda_N,s}(\sigma,\sigma)$ still in the sense of entry-wise product. 
Here, $R_{N,\lambda_N}(\sigma,\sigma)$ is the overlap between the same sample and is thus often called the self-overlap. The two terms in~\eqref{e.H^t,q_N=} involving the self-overlap are respectively the variances of $\sqrt{t}H_N(\sigma)$ and $W^q_N(\sigma,\alpha)$. These two terms are often called the self-overlap correction, which resembles the drift term in an exponential martingale.

We define the associated free energy and Gibbs measure
\begin{gather}
    \bar F_{N,\lambda_N}(t,q) = - \frac{1}{N}\E\log \iint  \exp\Ll( H^{t,q}_N(\sigma,\alpha)\Rr)\d P_{N,\lambda_N}(\sigma)\d \fR(\alpha),  \label{e.F_N(t,q)=}
    \\
    \la\cdot\ra_{N,\lambda_N} \propto \exp\Ll( H^{t,q}_N(\sigma,\alpha)\Rr)\d P_{N,\lambda_N}(\sigma)\d \fR(\alpha) \label{e.<>_N=}.
\end{gather}
Here, $\E$ first averages over all the Gaussian randomness in $H_N(\sigma)$ and $W^q_N(\sigma,\alpha)$ and then the randomness in $\fR$. This particular order of integration is needed to ensure that there is no measurability issues (see~\cite[Lemma~4.5]{HJ_critical_pts}).
Notice the additional minus sign on the right-hand side of~\eqref{e.F_N(t,q)=}.
We have omitted the dependence on $t$ and $q$ from the notation of $\la\cdot\ra_{N,\lambda_N}$, which should be clear from the context.
The dependence of $\bar F_{N,\lambda_N}(t,q)$ on the partition $(I_{N,s})_{s\in\sS}$ is only through $\lambda_N$, which is the reason for us to include it in the notation.
We omit the dependence on $\lambda_N$ in the notation of $H^{t,q}_N$, $H_N$, and $W^q_N$, which we prefer to keep implicit.

We can view $\bar F_{N,\lambda_N}$ as a function of $(t,q)\in\R_+\times\mcl Q^\sS_\infty(\kappa)$. By continuity in Proposition~\ref{p.F_N_smooth}, we can extend $\bar F_{N,\lambda_N}$ to the domain $\R_+\times \mcl Q^\sS_1(\kappa)$.

\subsubsection{Initial condition, functional, and critical points}
For each $\lambda\in\R_+^\sS$, there is a function $\psi_\lambda:\mcl Q^\sS_1(\kappa)\to\R$ (lemma~\ref{l.reg_initial}) such that $\bar F_{N,\lambda_N}(0,\cdot) =\psi_{\lambda_N}$ (see Lemma~\ref{l.initial_condition}).
For any $\lambda_\infty\in \blacktriangle_\infty$, $(t,q)\in \R_+\times \mcl Q^\sS_2(\kappa)$, we consider the functional
\begin{align}\label{e.mcJ=}
    \mcl J_{\lambda_\infty, t, q}(q',p) = \psi_{\lambda_\infty}(q') + \la p, q-q'\ra_{L^2}+t\int_0^1\xi(p)
\end{align}
defined for $q'\in \mcl Q^\sS_2(\kappa)$, $p\in L^2([0,1],\prod_{s\in\sS}\S^{\kappa_s})$. 
Here, $\la\cdot,\cdot\ra_{L^2}$ is the inner product in $L^2([0,1],\prod_{s\in\sS}\S^{\kappa_s})$.
Lastly integral is $\int_0^1\xi(p(s))\d s$.
We call~\eqref{e.mcJ=} the Hamilton--Jacobi functional due to its close connection to the Hamilton--Jacobi equation~\eqref{e.hj}. For instance, in Theorem~\ref{t.convex_2} to be stated, the variational formula representation for the limit of free energy is exactly the Hopf--Lax formula. 
We refer to~\cite[Section~1]{HJ_critical_pts} for more detail on the Hamilton--Jacobi equation approach to spin glass.

We say that a pair $(q',p) \in \mcl Q^\sS_2(\kappa)\times L^2([0,1],\prod_{s\in\sS}\S^{\kappa_s})$ is a \textbf{critical point} of the functional $\mcl J_{\lambda_\infty, t, q}$ if
\begin{align}\label{e.critical_rel}
    q=q'-t\nabla\xi(p) \qquad\text{and}\qquad p=\partial_q \psi_{\lambda_\infty}(q').
\end{align}
Here, the derivative $\partial_q \psi_{\lambda_\infty}$ is understood in the sense of Fr\'echet, which is rigorously defined in~\eqref{e.def_Frechet}. The differentiability of $\psi_{\lambda_\infty}$ is ensured by Lemma~\ref{l.reg_initial}.

Heuristically, at any critical point $(q',p)$, the derivatives of $\mcl J_{\lambda_\infty, t, q}$ in $q'$ and $p$ are both zero.
Critical points and the value of the functional at these points are important to our main results to be stated.

\subsubsection{Regular paths}We introduce subsets of $\mcl Q^\sS(\kappa)$ consisting of regular paths that play important roles in the differentiability of limits of free energy on $\mcl Q^\sS_2(\kappa)$.
For any $\D\in\N$ and any matrix $a\in\S^\D_+$, we define
\begin{align}\label{e.ellipt(a)=}
    \ellipt(a) = \frac{\max_{u\in\R^\D:\:|u|=1}u^\intercal a u}{\min_{u\in\R^\D:\:|u|=1}u^\intercal a u}.
\end{align}
In other word, $\ellipt(a)$ is the ratio of the largest eigenvalue over the smallest eigenvalue of $a$. We have $a\in\S^\D_{++}$ if and only if $\ellipt(a)<+\infty$.
For $\D\in\N$ and $c>0$, we define 
\begin{gather}
\begin{split}\label{e.Q_uparrow,c(D)=}
    \mcl Q_{\uparrow,c}(\D)=\big\{ q\in \mcl Q(\D)\ \big|\ q(0)=0\  \text{and} \  \forall r\leq r'\in[0,1),\  q(r')-q(r)\geq c(r'-r)\identity_\D
   \\
   \text{and} \  \ellipt(q(r')-q(r))\leq c^{-1}\big\};
\end{split}
    \\
    \mcl Q_{\uparrow}(\D) = \mcl \cup_{c>0}\mcl Q_{\uparrow,c}(\D);\qquad \mcl Q_{\infty,\uparrow}(\D) =  Q_{\infty}(\D)\cap  Q_{\uparrow}(\D).\label{e.Q_uparrow(D)=}
\end{gather}
Here, $\identity_\D$ is the $\D\times\D$ identity matrix.
Let $\mcl Q_{\uparrow,c}(\kappa)$, $\mcl Q_{\uparrow}(\kappa)$, and $\mcl Q_{\infty,\uparrow}(\kappa)$ be given as in~\eqref{e.Q^S(kappa)=}.

\subsection{Main results}\label{s.main_results}

In the results to be stated below, we say that $\Ll(\bar F_{N,\lambda_N}\Rr)_{N\in\N}$ converges to some $f$ if $(\bar F_{N,\lambda_N}(t,q))_{N\in\N}$ converges to $f(t,q)$ at every $(t,q)\in\R_+\times \mcl Q^\sS_1(\kappa)$. Namely, the convergence is pointwise.

The result below adapts \cite[Theorem~1.1]{HJ_critical_pts} partially. The missing part to recover the full version (Theorem~\ref{t.convex_2}) is formulated into Claim~\ref{c.free_energy_upper_bd}. We discuss this issue after the statements of four main theorems.

\begin{theorem}\label{t.convex}
If $\xi$ is a convex function on $\prod_{s\in\sS}\S^{\kappa_s}_+$ and $(\lambda_N)_{N\in\N}$ converges to some $\lambda_\infty$, then for every $(t,q)\in \R_+\times \mcl Q^\sS_2(\kappa)$, we have
\begin{align}\label{e.t.convex}
    \lim_{N\to\infty} \bar F_{N,\lambda_N}(t,q) = \inf_{p\in \mcl Q^\sS_\infty(\kappa)}\mcl J_{\lambda_\infty,t,q}\Ll(q+t\nabla\xi(p),p\Rr).
\end{align}
\end{theorem}

This theorem is in fact simply the Parisi formula, which can be rewritten in a more familiar form in Proposition~\ref{p.mp-parisi}. 

The next two results adapt \cite[Theorems~1.2 and 1.3]{HJ_critical_pts} respectively to the setting here.
In the following, the Gateaux differentiability is defined in Section~\ref{s.diff_monot_precompt}.

\begin{theorem}[Critical point representation]\label{t.crit_pt_rep}
Assume that $(\lambda_N)_{N\in\N}$ converges to some $\lambda_\infty$ and that $(\bar F_{N,\lambda_N})_{N\in\N}$ converges to some $f$. Then, for every $(t,q) \in \R_+\times \mcl Q^\sS_2(\kappa)$, there exists $(q',p)\in \mcl Q^\sS_2(\kappa)\times \mcl Q^\sS_2(\kappa)$ that is a critical point of $\mcl J_{\lambda_\infty,t,q}$ and is such that
\begin{align}\label{e.t.crit_pt_rep}
    \lim_{N\to\infty} \bar F_{N,\lambda_N}(t,q) = \mcl J_{\lambda_\infty,t,q}(q',p).
\end{align}
\end{theorem}

\begin{theorem}[Critical point identification]\label{t.crit_pt_id}
Assume that $(\lambda_N)_{N\in\N}$ converges to some $\lambda_\infty$ and that $(\bar F_{N,\lambda_N})_{N\in\N}$ converges to some $f$. 
For every $(t,q)\in\R_+\times \mcl Q^\sS_{\infty,\uparrow}(\kappa)$, if $f(t,\cdot)$ is Gateaux differentiable at $q$, then letting $p=\partial_q f(t,q)$ and $q'=q+t\nabla \xi (p)$, we have that $(q',p)\in \mcl Q^\sS_2(\kappa)\times \mcl Q^\sS_2(\kappa)$ is a critical point of $\mcl J_{\lambda_\infty,t,q}$ and is such that~\eqref{e.t.crit_pt_rep} holds.
\end{theorem}

Theorem~\ref{t.crit_pt_rep} stats that, if the free energy converges, then the limit must be a critical value of Hamilton--Jacobi functional. 
For $t> 0$ sufficiently small and $q = 0$, namely, in the replica symmetry regime, Theorem~\ref{t.crit_pt_rep} was shown in \cite{dey2021fluctuation} for the multi-species Sherrington--Kirkpatrick model. 
Theorem~\ref{t.crit_pt_id} shows that, at every differentiable point in $\mcl Q^\sS_{\infty,\uparrow}(\kappa)$ of the limit, the critical point is uniquely determined. By Proposition~\ref{p.reg_limit}, such differentiability points are dense in $\mcl Q^\sS_2(\kappa)$.

Next, we state an adapted version of~\cite[Theorem~1.4]{HJ_critical_pts}.
By adding a small perturbation of quadratic interaction, we can show that the overlap $R_{N,\lambda_N}(\sigma,\sigma')$ under $\E\la\cdot\ra_{N,\lambda_N}$ converges in law to $p$ from the critical point $(q',p)$. Here, $\sigma'$ is an independent copy of $\sigma$ sampled from $\la\cdot\ra_{N,\lambda}$ introduced in~\eqref{e.<>_N=}.
The perturbation is given by an additional Gaussian Hamiltonian $(\hat H_N(\sigma))_{\sigma \in \Sigma^N}$ characterized by the covariance
\begin{equation}\label{e.hatH_N}  \E \Ll[ \hat H_N(\sigma) \hat H_N(\sigma') \Rr] = N \Ll| R_{N,\lambda_N}(\sigma,\sigma') \Rr| ^2.
\end{equation}
Letting $(g_{s,n,n',k,k'})$ be a family of i.i.d.\ standard Gaussian variables independent of all other randomness, we can explicitly construct
\begin{equation*}  \hat H_N(\sigma) := \frac{1}{\sqrt{N}} \sum_{s\in \sS}\sum_{n,n' \in I_{N,s}}\sum_{k,k'=1}^{\kappa_s} g_{s,n,n',k,k'}\sigma_{kn}\sigma_{k'n'}.
\end{equation*}
For every $\hat t \ge 0$,  to $H^{t,q}_N(\sigma,\alpha)$ in~\eqref{e.H^t,q_N=}, we add the quantity
\begin{equation*}  
\sqrt{2\hat t} \, \hat H_N(\si) - \hat t N \Ll| R_{N,\lambda_N}(\sigma,\sigma) \Rr| ^2
\end{equation*}
and let $\hat F_{N,\lambda_N}(t,\hat t,q)$ be the corresponding free energy, which is precisely expressed in \eqref{e.def.FN.t.hat}. The corresponding functional (see~\eqref{e.mcJ=}) becomes
\begin{equation}  \label{e.hatJ=}
\hat {\mcl J}_{\lambda_\infty,t,\hat t,q}(q',p) := \psi_{\lambda_\infty}(q') +  \la p , q-q'\ra_\cH + t \int_0^1 \xi(p) + \hat t \int_0^1 |p|^2.
\end{equation}
\begin{theorem}[Identification of overlap distribution]
\label{t.main3}
Assume that $(\lambda_N)_{N\in\N}$ converges to some $\lambda_\infty$ and that $(\hat F_{N,\lambda_N})_{N\in\N}$ converges to some $\hat f$ along some subsequence.
Suppose also that, for some $(t,q) \in \R_+\times \mcl Q^\sS_2(\kappa)$ and $\hat t > 0$, we have that $\hat f(t,\hat t, \cdot)$ is Gateaux differentiable at $q$, and that $\hat f(t,\cdot,q)$ is differentiable at $\hat t$. Then letting $p = \dr_q \hat f(t,\hat t, q)$ and $q' = q + t \nabla \xi(p) + 2 \hat t p$, we have that $(q',p)\in \mcl Q^\sS_2(\kappa)\times \mcl Q^\sS_2(\kappa)$ is a critical point of $\hat {\mcl J}_{\lambda_\infty,t,\hat t, q}$. 
Moreover, as $N$ tends to infinity along the said subsequence, the overlap $R_{N,\lambda_N}(\sigma,\sigma')$ under $\E\la\cdot\ra_{N,\lambda_N}$ converges in law to $p(U)$, where $U$ is the uniform random variable over $[0,1]$. 
\end{theorem}
In the above theorem, we do not require that the free energy converges as $N$ tends to infinity. Any subsequential convergence is enough. We can always extract such a convergent subsequence through a precompactness result such as Proposition~\ref{p.precompactness}.
If $\hat F_{N,\lambda_N}$ does converge along the full sequence, then an application of Theorem~\ref{t.crit_pt_rep} yields a representation of the free energy itself as
\begin{equation*}  \lim_{N \to \infty} \hat F_{N,\lambda_N}(t,\hat t,q) = \hat{\mcl J}_{\lambda_\infty,t,\hat t,q}(q',p),
\end{equation*}
for the same critical point $(q',p)$ of $\hat{\mcl J}_{\lambda_\infty,t,\hat t,q}$. Here and in the above theorem, that $(q',p)$ is a critical point of $\hat{\mcl J}_{\lambda_\infty,t,\hat t,q}$ means that
\begin{equation*}  q = q'-t \nabla \xi(p) - 2\hat t  p \quad \text{ and } \quad p = \dr_q \psi_{\lambda_\infty}(q').
\end{equation*}
As in Theorem~\ref{t.crit_pt_id}, for each $t \ge 0$, almost every $(\hat t, q) \in \R_+\times \mcl Q^\sS_{\infty,\uparrow}(\kappa)$ is a point that satisfies the assumptions of Theorem~\ref{t.main3}. In particular, the set of such points is dense in $\R_+ \times \mcl Q^\sS_2(\kappa)$.

Adapted versions of \cite[Propositions~1.5 and 1.6]{HJ_critical_pts} in the introduction of \cite{HJ_critical_pts} are adapted into Propositions~\ref{p.uniq_crit_pit_high_temp} and~\ref{p.crit_pt_stable}.

Lastly, we come back to the issue with adapting~\cite[Theorem~1.1]{HJ_critical_pts}. The missing piece is summarized as follows.

\begin{claim}[Free energy upper bound]\label{c.free_energy_upper_bd}
Assume that $(\lambda_N)_{N\in\N}$ converges to some $\lambda_\infty$. For every $(t,q)\in \R_+\times \mcl Q^\sS_2(\kappa)$, we have
\begin{align}\label{e.liminfF_N>}
    \liminf_{N\to\infty} \bar F_{N,\lambda_N}(t,q)\geq f(t,q)
\end{align}
where $f: \R_+\times \mcl Q^\sS_2(\kappa) \to\R$ is the solution to the equation
\begin{align}\label{e.hj}
    \begin{cases}
        \partial_t f- \int_0^1 \xi(\partial_q f)=0,\quad &\text{on $\R_+\times \mcl Q^\sS_2(\kappa)$},
        \\
        f(0,\cdot) = \psi_{\lambda_\infty}, &\text{on $ \mcl Q^\sS_2(\kappa)$}.
    \end{cases}
\end{align}
Moreover, if $\xi$ is a convex function on $\prod_{s\in\sS}\S^{\kappa_s}_+$, then $f$ admits the Hopf--Lax representation at every $(t,q)\in \R_+\times \mcl Q^\sS_2(\kappa)$:
\begin{align}\label{e.f=hopf-lax}
    f(t,q) = \sup_{q'\in q+ \mcl Q^\sS_\infty(\kappa)}\inf_{p\in \mcl Q^\sS_\infty(\kappa)}\mcl J_{\lambda_\infty,t,q}\Ll(q',p\Rr).
\end{align}
\end{claim}

Parts~\eqref{e.liminfF_N>} and~\eqref{e.f=hopf-lax} of this claim respectively adapt~\cite[Theorem~3.4]{mourrat2023free} and~\cite[Corollary~4.14]{chen2022hamilton} (for vector spins) to the multi-species setting. The modification to the proofs should be straightforward but tedious. Hence, we only formulate this claim here without proving it.
Given this claim, we can recover the full version of \cite[Theorem~1.1]{HJ_critical_pts}.

\begin{theorem}\label{t.convex_2}
Assume that Claim~\ref{c.free_energy_upper_bd} is valid.
If $\xi$ is a convex function on $\prod_{s\in\sS}\S^{\kappa_s}_+$ and $(\lambda_N)_{N\in\N}$ converges to some $\lambda_\infty$, then for every $(t,q)\in \R_+\times \mcl Q^\sS_2(\kappa)$, we have
\begin{align}\label{e.t.convex_2}
    \lim_{N\to\infty} \bar F_{N,\lambda_N}(t,q) = \sup_{q'\in q+ \mcl Q^\sS_\infty(\kappa)}\inf_{p\in \mcl Q^\sS_\infty(\kappa)}\mcl J_{\lambda_\infty,t,q}\Ll(q',p\Rr) =f(t,q)
\end{align}
where $f$ is the solution to~\eqref{e.hj}.
\end{theorem}

Convexity is needed here to identify the limit of free energy to be the solution of~\eqref{e.hj}. Without convexity and assuming the convergence of free energy, we can only show that the limit satisfies~\eqref{e.hj} at every differentiability point (see~\eqref{e.d_tf=intxi(d_qf)} in Proposition~\ref{p.crit_pt_2}). This is not enough to identify the limit because we do not have the uniqueness of solutions in the sense of satisfying the equation on a dense set.

\subsection{Related works}

This work concerns the limit of free energy. In physics, the replica method has been developed and proven to be powerful \cite{MPV, parisi1979infinite, parisi1980order, parisi1980sequence, parisi1983order}.
As a special case of the multi-species setting, the bipartite spin glass model has been considered in \cite{fyo1, fyo2, hartnett2018replica, korenblit1985spin}, where formulas for the limit free energy in terms of functionals similar to $\mcl J_{\lambda_\infty,t,q}$ in~\eqref{e.mcJ=} has been obtained using the replica method.

Another perspective is based on Hamilton--Jacobi equations, which is explored in \cite{mourrat2021nonconvex, mourrat2023free}; see also \cite{HJbook}. The functional $\mcl J_{\lambda_\infty,t,q}$ is closely related to the characteristic lines of the Hamilton--Jacobi equation~\eqref{e.hj}. Our main result can be restated as that the limit of free energy, if exists, attains its value along a certain characteristic line. The conjecture of \cite{mourrat2021nonconvex, mourrat2023free} (see also \cite[Question~6.11]{HJbook}) is that, in the general non-convex case, the free energy $\bar F_{N,\lambda_N}$ converges to the unique solution of~\eqref{e.hj}

In the convex case, much is known. The limit of free energy in the Sherrington--Kirkpatrick models has been mathematically identified as the Parisi formula in \cite{gue03, Tpaper} (also see~\cite{Tbook1, Tbook2}), which is then extended to more general scalar models \cite{pan.aom, pan, pan14} using an argument based on the ultrametricity of the asymptotic Gibbs measure (see also~\cite{aizenman1998stability, ghirlanda1998general, guerra1996overlap}). Then, this is extended further to multi-species and vector spin glasses in \cite{barcon, pan.multi, pan.potts, pan.vec} under the assumption that $\xi$ is convex on the whole space. Here, Theorem~\ref{t.convex} (see also Proposition~\ref{p.mp-parisi}) encompasses these results with a weaker assumption that $\xi$ is convex on positive semi-definite matrices. For spherical spins, the limit of free energy has been identified for scalar models~\cite{chen2013aizenman, talagrand2006free} and multi-type models~\cite{bates2022crisanti, bates2022free, ko2020free, panchenko2007overlap}.

In the context of spin glasses, the connection to the Hamilton--Jacobi equation was first explored in \cite{abarra,barra1} and further in various settings in~\cite{barra2,abarra,barramulti,barra2014quantum}. 
The Parisi formula in the scalar case was rewritten as the Hopf--Lax formula for the solution of an equation of the form in~\eqref{e.hj} in \cite{mourrat2022parisi} (later extended in~\cite{mourrat2020extending}).
Here, the equation is infinite-dimensional and its well-posedness together with variational representations was proved in~\cite{chen2022hamilton}. For non-convex models, the best result so far is the lower bound as in~\eqref{e.liminfF_N>} established in \cite{mourrat2021nonconvex, mourrat2023free}. In particular, when $\xi$ is not convex, the Hopf--Lax formula as in~\eqref{e.t.convex_2} is false (see~\cite[Section~6]{mourrat2021nonconvex}). Also, since $\psi_{\lambda_\infty}$ does not seem to be convex in general, another possibility, the Hopf formula
\begin{align}\label{e.hopf}
    \sup_{p\in \mcl Q^\sS_\infty(\kappa)}\inf_{q'\in \mcl Q^\sS_\infty(\kappa)}\mcl J_{\lambda_\infty,t,q}\Ll(q',p\Rr)
\end{align}
cannot be the limit of free energy.

For some non-convex models, other partial bounds for the limit free energy were obtained in \cite{alberici2020annealing, alberici2021deep}. Some works also focus on the high-temperature regime, including \cite{barramulti, dey2021fluctuation, genovese2023minimax}.
We also mention that additional symmetry in the model can lead to simplification \cite{bates2023parisi, chen2023parisi, issa2024existence}. 

Formulas of the form in~\eqref{e.hopf} appear in some problems from high-dimensional statistical inference~\cite{barbier2016, barbier2019adaptive, barbier2017layered, chen2022statistical, HB1, HBJ, kadmon2018statistical,   lelarge2019fundamental, lesieur2017statistical, luneau2020high, mayya2019mutualIEEE, 
miolane2017fundamental, mourrat2020hamilton, mourrat2021hamilton, reeves2020information, reeves2019geometryIEEE} (see also~\cite[Chapter~4]{HJbook} for a Hamilton--Jacobi approach). In certain problems on sparse graphs \cite{dominguez2022infinite, dominguez2022mutual, kireeva2023breakdown}, an issue similar to the non-convexity of $\psi_{\lambda_\infty}$ occurs and thus invalidates the Hopf formula.
More broadly, connections between mean-field models and Hamilton--Jacobi equations have been noticed in~\cite{bra83, new86}. We refer to a recent survey \cite{bauerschmidt2023stochastic} for related topics.

In certain spherical models with multiple types and non-convex $\xi$, the limit of free energy, if exists, was explicitly identified in \cite{subag2021multi2, subag2021free, subag2021multi1, subag2021multi3}. Also see the related work \cite{baik2020free} on the bipartite spherical Sherrington–Kirkpatrick model. A more geometric analysis of the energy landscape can be found in \cite{benarous2022exponential, huang2023strong, kivimae2023ground, mckenna2021complexity}. For scalar models, also see \cite{auffinger2013complexity, auffinger2013random, fyodorov2004complexity, subag2017complexity}.

\subsection{Outline of the paper}

In Section~\ref{s.prelim}, we recall the Gaussian interpolation technique and use it to show two results: the continuity of $\bar F_{N,\lambda_N}$ in $\lambda_N$ and the estimate on the discrepancy between $\bar F_{N,\lambda_N}$ and $\bar F_{N+\M, \lambda_{N+\M}}$. 

In Section~\ref{s.rel_to_vector_spin}, we relate the multi-species spin glass model with rational species proportions to some vector spin glass model. In particular, we show that, if $(\lambda_{N,s})_{N\in\N}$ converges to a rational number for every $s\in\sS$, then the free energy of the multi-species model is asymptotically equal to that of a vector spin model. In this case, results from~\cite{HJ_critical_pts} for vector spin glasses are directly applicable. 
It remains to consider the case where $(\lambda_{N,s})_{N\in\N}$ converges to an irrational number for some $s\in\sS$, which we refer to as the irrational case. 

In preparation for treating the irrational case, in Section~\ref{s.analytic_properties}, we collect analytic properties of the free energy $\bar F_{N,\lambda_N}$, its limit (if exists), and the initial condition $\bar F_{N,\lambda_N}(0,\cdot)$. These are results adapted from those in~\cite[Sections~3 and 5]{HJ_critical_pts}.

To handle the general case including the irrational case, we need to redo some of the cavity computation in Section~\ref{s.cavity_and_proofs}. After importing cavity computation results for the rational case from~\cite[Section~6]{HJ_critical_pts}, we need an additional approximation procedure to obtain the results in the general case. This allows us to prove corresponding results from~\cite[Section~7]{HJ_critical_pts} adapted to the current setting. 
In particular, we prove Theorems~\ref{t.crit_pt_rep},~\ref{t.crit_pt_id}, and~\ref{t.main3} in Section~\ref{s.crti_pt_and_rel_results}. 

In Section~\ref{s.convex}, we apply those results to the convex setting, namely, the one where $\xi$ is convex. We recover corresponding results from~\cite[Section~8]{HJ_critical_pts} and prove Theorems~\ref{t.convex} and~\ref{t.convex_2},

\section{Gaussian interpolation and two estimates}\label{s.prelim}

In this section, we record two estimates on the continuity of the free energy in terms of $\lambda_N$ in Lemma~\ref{l.lambda_continuity} and the discrepancy of free energy with different configuration sizes in Lemma~\ref{l.NF_N-(N+M)F_(N+M)}. 
To prove them, we first recall the Gaussian interpolation technique stated as in Lemma~\ref{l.Gaussian interpolation} and Corollary~\ref{c.Gaussian_interpolation}, which will also be needed later.

\begin{lemma}[Gaussian interpolation technique]\label{l.Gaussian interpolation}
Let $\mathfrak{P}$ be a random probability measure on some Polish space $\mathcal{X}$. Suppose that, conditioned on $\mathfrak{P}$, there are two independent centered Gaussian process $(\mathbf{G}_i(\mathbf{x}))_{\mathbf{x}\in \supp \mathfrak{P}}$ and two bounded deterministic function $\mathbf{D}_i:\mathcal{X}\to\R$ for $i\in\{0,1\}$. 
Also, assume that there are deterministic functions $\mathbf{V}_i:\mathcal{X}\times \mathcal{X}\to \R$ for $i\in\{0,1\}$ such that, conditioned on $\mathfrak{P}$,
\begin{align*}
    \mathbf{V}_i(\mathbf{x},\mathbf{x}') = \E \mathbf{G}(\mathbf{x})\mathbf{G}(\mathbf{x}'),\quad\forall \mathbf{x}, \mathbf{x}'\in \supp\mathfrak{P}.
\end{align*}
For $r\in[0,1]$, define $\mathbf{G}_r=\sqrt{1-r}\mathbf{G}_0 + \sqrt{r}\mathbf{G}_1$, $\mathbf{D}_r=(1-r)\mathbf{D}_0 + r\mathbf{D}_1$, and
\begin{align*}
    \varphi(r) = -\E \log \int_\mathcal{X} \exp\Ll(\sqrt{2}\mathbf{G}_r(\mathbf{x}) + \mathbf{D}_r(\mathbf{x})\Rr)\d\mathfrak{P}(\mathbf{x})
\end{align*}
where $\E$ first averages over all Gaussian randomness (conditioned on $\mathfrak{P}$) and then that of $\mathfrak{P}$. Also, take the Gibbs measure $\la\cdot\ra_r \propto \exp\Ll(\sqrt{2}\mathbf{G}_r(\mathbf{x}) + \mathbf{D}_r(\mathbf{x})\Rr)\d\mathfrak{P}$. Then, we have
\begin{align*}
    \varphi(1)-\varphi(0)=\int_0^1\E\la \mathbf{V}_1(\mathbf{x},\mathbf{x}')-\mathbf{V}_0(\mathbf{x},\mathbf{x}')-\mathbf{V}_1(\mathbf{x},\mathbf{x})+\mathbf{V}_0(\mathbf{x},\mathbf{x})-\mathbf{D}_1(\mathbf{x})+\mathbf{D}_0(\mathbf{x})\ra_r\d r
\end{align*}
where $\mathbf{x}'$ is an independent copy of $\mathbf{x}$ under $\la\cdot\ra_r$.
\end{lemma}

We have an immediate corollary.
\begin{corollary}[Interpolation with self-overlap correction]\label{c.Gaussian_interpolation}
Under the same setting of Lemma~\ref{l.Gaussian interpolation}, we further assume that $\mathbf{D}_i(\mathbf{x})=-\mathbf{V}_i(\mathbf{x},\mathbf{x})$ for every $\mathbf{x}\in\mathcal{X}$ and $i\in\{0,1\}$. Then, we have
\begin{align*}
    \varphi(1)-\varphi(0)=\int_0^1\E\la \mathbf{V}_1(\mathbf{x},\mathbf{x}')-\mathbf{V}_0(\mathbf{x},\mathbf{x}')\ra_r\d r
\end{align*}
\end{corollary}

\begin{proof}[Proof of Lemma~\ref{l.Gaussian interpolation}]
We can compute
\begin{align*}
    \frac{\d}{\d r}\varphi(r) = -\E \la (2-2r)^{-1/2} \mathbf{G}_1(\mathbf{x}) - (2r)^{-1/2} \mathbf{G}_0(\mathbf{x}) + \mathbf{D}_1(\mathbf{x})-\mathbf{D}_0(\mathbf{x})\ra_r.
\end{align*}
Apply the Gaussian integration by parts (e.g.\ see~\cite[Lemma~1.1]{pan}), we get
\begin{align*}
    \frac{\d}{\d r}\varphi(r) =\E\la \mathbf{V}_1(\mathbf{x},\mathbf{x}')-\mathbf{V}_0(\mathbf{x},\mathbf{x}')-\mathbf{V}_1(\mathbf{x},\mathbf{x})+\mathbf{V}_0(\mathbf{x},\mathbf{x})-\mathbf{D}_1(\mathbf{x})+\mathbf{D}_0(\mathbf{x})\ra_r
\end{align*}
which gives the desired result.
\end{proof}

Now, with the Gaussian interpolation technique, we can start to prove the first estimate.

\begin{lemma}\label{l.lambda_continuity}
Let $C_\xi = \max_{ x \in \prod_{s\in\sS}[-1,+1]^{\kappa_s\times \kappa_s}}\Ll|\nabla \xi(x)\Rr|$, $C_\mu = \sqrt{\sum_{s\in\sS}\Ll|\log \mu_s(\R^{\kappa_s})\Rr|^2}$, and $|\kappa|_\infty=\max_{s\in\sS}\kappa_s$.
Then, for every $N\in \N$, $t\in\R_+$, $q\in \mcl Q_\infty^\sS(\kappa)$, and $\lambda_N,\,\lambda'_N\in \blacktriangle_N$, we have
\begin{align}\label{e.|F-F|<C|lambda-lambda|}
    \Ll|\bar F_{N,\lambda_N}(t,q) - \bar F_{N,\lambda'_N}(t,q)\Rr|\leq \Ll( |\kappa|_\infty C_\xi t  + |\kappa|_\infty |q|_{L^1} +C_\mu\Rr)\Ll|\lambda_N-\lambda'_N\Rr|.
\end{align}
\end{lemma}

Here and henceforth, the $L^p$-norm of $q$ for $p\in[0,\infty]$ is given as
\begin{align*}
    \Ll|q\Rr|_{L^p} =\Big(\int_0^1 |q(r)|^p\d r\Big)^\frac{1}{p},\qquad \Ll|q\Rr|_{L^\infty} =\esssup_{r\in[0,1]}|q(r)|
\end{align*}
where $|q(r)|$ is the norm in $\prod_{s\in\sS}\S^{\kappa_s}$ as given in~\eqref{e.dot_product}.

\begin{proof}
Throughout this proof, we fix $t$ and $q$ and thus omit them from the notation of different versions of free energy. We write $R_N = R_{N,\lambda_N}$ and $P_N = P_{N,\lambda_N}$ for brevity.

Let $I_N=(I_{N,s})_{s\in\N}$ be the partition associated with $\lambda_N$. Let $J=(J_s)_{s\in\N}$ be a collections of subsets $J_s\subset I_{N,s}$. For each configuration $\sigma\in \Sigma^N$, we write $\sigma^J = \Ll(\sigma_{\bullet n} \one_{n \in \cup_{s\in \sS} J_s}\Rr)_{1\leq n\leq N}$. 
We define $\bar F_{N,I_N}^J$ in the same way as~\eqref{e.F_N(t,q)=} but with every instance of $\sigma$ in $\exp(\cdots)$ replaced by $\sigma^J$. Notice that in this definition, we keep $P_N(\sigma)$ in $\bar F_{N,I_N}^J$ intact, which is source of the dependence on $I_N$.
Define $\bar F^J_N$ by further replacing $\d P_n(\sigma)$ in $\bar F^J_{N,I_N}$ by
\begin{align}\label{e.P^J_N=}
    \d P^J_N\Ll(\Ll(\sigma_{\bullet n}\Rr)_{n\in \cup_s J_s}\Rr)=\otimes_{s\in\sS} \otimes_{n\in J_s} \d \mu_s(\sigma_{\bullet n}).
\end{align}
It is straightforward to see that if $J'=(J'_s)$ is a partition of a subset of $\{1,\dots ,N\}$ and $|J'_s|=|J_s|$ for every $s$, we have
\begin{align}\label{e.F^J_N=F^J'_N}
    \bar F^J_N = \bar F^{J'}_N.
\end{align}
Writing $|I_{N}\setminus J| =\Ll(\sum_{s\in\sS}|I_{N,s}\setminus J_s|^2 \Rr)^\frac{1}{2}$, we claim
\begin{align}
    \Ll|\bar F_{N,\lambda_N}-\bar F_{N,I_N}^J\Rr|&\leq \Ll(tC_\xi + |q|_{L^1}\Rr) |\kappa|_\infty N^{-1}|I_{N}\setminus J|,\label{e.F-F^J_I<}
    \\
     \Ll|\bar F_{N,I_N}^J-\bar F_N^J\Rr|&\leq C_\mu N^{-1}|I_{N}\setminus J|.\label{e.F^J_I-F^J<}
\end{align}
Before proving them, we first use them to deduce the desired estimate.

Let $(I'_{N,s})_{s\in\N}$ be a partition with proportions given by $\lambda'_N$. We can find a bijection $\iota$ from $\{1,\dots,N\}$ to itself such that, setting $J_s = I_{N,s}\cap \iota^{-1} (I'_{N,s})$ for each $s$, we have $J_s=I_{N,s}$ or $\iota (J_s)=I'_{N,s}$ for every $s\in\sS$.
This property implies 
\begin{align}\label{e.|I-J|+|I'-iotaJ|=||I|-|I'||}
     \Ll|I_{N,s}\setminus J_s\Rr|+ \Ll|I'_{N,s}\setminus \iota(J_s)\Rr| = \Ll||I_{N,s}|-|I'_{N,s}|\Rr|=N\Ll|\lambda_{N,s}-\lambda'_{N,s}\Rr|,\quad\forall s\in\sS.
\end{align}
Write $J= (J_s)_{s\in \sS}$ and $\iota(J)= (\iota(J_s))_{s\in\sS}$. We can apply~\eqref{e.F-F^J_I<} and ~\eqref{e.F^J_I-F^J<} to the pair $J$ and $I_N$ and the pair $\iota(J)$ and $I'_N$. These along with the triangle inequality give
\begin{align*}
    \Ll|\bar F_{N,\lambda_N} - \bar F_{N,\lambda'_N}\Rr|\leq \Ll|\bar F^J_N - \bar F^{\iota(J)}_N\Rr|+ CN^{-1} \Ll|I_N\setminus J\Rr|+ CN^{-1}\Ll|I'_N\setminus \iota(J_s)\Rr|
    \\
    \stackrel{\eqref{e.F^J_N=F^J'_N},\,\eqref{e.|I-J|+|I'-iotaJ|=||I|-|I'||}}{=}0+ C|\lambda_{N,s}-\lambda'_{N,s}|
\end{align*}
for $C=\Ll(tC_\xi + |q|_{L^1}\Rr)|\kappa|_\infty +C_\mu$ as announced.

It remains to verify~\eqref{e.F-F^J_I<} and~\eqref{e.F^J_I-F^J<}. We first show the latter. Recall that the reference measure in $\bar F^J_{N,I_N}$ is $\d P_N(\sigma)$ as in~\eqref{e.P_N=} and that in $\bar F^J_N$ is $\d P^J_N(\cdots)$ as in~\eqref{e.P^J_N=}. Also notice that the term $\exp(\cdots)$ in $\bar F^J_{N,I_N}$ does not depend on spins $\sigma_{\bullet n}$ for $n\not\in \cup_{s\in\sS}J_s$. Hence, we can compute
\begin{align*}
    \bar F^J_{N,I_N} = \bar F^J_N - N^{-1}\sum_{s\in\sS}|I_s\setminus J_s| \log \mu_s(\R^{\kappa_s})
\end{align*}
which gives~\eqref{e.F^J_I-F^J<}.

To show~\eqref{e.F-F^J_I<}, we need an interpolation argument. 
Let us rewrite the terms inside $\exp(\cdots)$ in $\bar F_{N,\lambda_N}(t,q)$ (see~\eqref{e.F_N(t,q)=}) as $\sqrt{2}\mathbf{G}_1(\sigma,\alpha)+\mathbf{D}_1(\sigma)$, where $\mathbf{G}(\sigma,\alpha)$ collects the Gaussian terms (conditioned on $\fR$) and $\mathbf{D}(\sigma)$ collects the deterministic terms. 
Set $\mathbf{G}_0(\sigma,\alpha)= \mathbf{G}_1(\sigma^J,\alpha)$ and $\mathbf{D}_0(\sigma)=\mathbf{D}_1(\sigma^J)$. 
Notice that 
\begin{align*}
    \E \mathbf{G}_1(\sigma,\alpha)\mathbf{G}_1(\sigma',\alpha')\stackrel{\eqref{e.EH_N(sigma)H_N(sigma')=},\,\eqref{e.Ew^q_s_iw^q_s_i=}}{=} Nt\xi(R_N(\sigma,\sigma')) + Nq(\alpha \wedge\alpha')\cdot R_N(\sigma,\sigma')  
\end{align*}
and also the covariance of $\mathbf{G}_0$ is the same with $\sigma$ above replaced by $\sigma^J$.
We can apply Corollary~\ref{c.Gaussian_interpolation} in the setting (with $\varphi(1) = N\bar F_{N,\lambda_N}$ and $\varphi(0) = N\bar F^J_{N,I_N}$). 
Thus, we get
\begin{align*}
    N\bar F_{N,\lambda_N} - N\bar F^J_{N,I_N} = \int_0^1 \E \la t(\xi(R_N) -\xi(R^J_N))+ q(\alpha \wedge\alpha')\cdot ( R_N-R^J_N)\ra_r \d r
\end{align*}
where we used the shorthand $R_N=R_N(\sigma,\sigma')$ and $R_N^J=R_N(\sigma^J,\sigma'^J)$. 
We denote their $s$-coordinate as $R_{N,s}$ and $R^J_{N,s}$.
We recall an important property of invariance of cascade: for every $r\in[0,1]$,
\begin{align}\label{e.E<>_r=int}
    \E \la q(\alpha\wedge\alpha')\ra_r = \E \la q(\alpha\wedge\alpha')\ra_\fR =\int_0^1q(u) \d r.
\end{align}
This property is restated as in Lemma~\ref{l.invar} later.
By the definition of $\sigma^J$, almost surely, we have $|R_{N,s}-R^J_{N,s}|\leq \kappa_s N^{-1}|I_{N,s}\setminus J_s|$. Hence, we have $|R_N-R^J_N|\leq |\kappa|_\infty N^{-1}|I_N\setminus J|$. Using this and~\eqref{e.E<>_r=int}, we get
\begin{align*}
    \Ll|N\bar F_{N,\lambda_N} - N\bar F^J_{N,I_N}\Rr|\leq  t C_\xi|\kappa|_\infty  |I_N\setminus J| +  |q|_{L^1}|\kappa|_\infty  |I_N\setminus J|
\end{align*}
which implies~\eqref{e.F-F^J_I<}.
\end{proof}

Now, we turn to the second estimate.
We again take $|\kappa|_\infty = \max_{s\in\sS}\kappa_s$.

\begin{lemma}\label{l.NF_N-(N+M)F_(N+M)}
There are constants $C_\xi$ depending only on $\xi$ (over $\prod_{s\in\sS}[-1,+1]^{\kappa_s\times \kappa_s}$) and $C_\mu$ depending only on $\mu$ such that the following holds.
For every $N,\M\in\N$ with $N\geq M$, $t\in\R_+$, $q\in \mcl Q_\infty^\sS(\kappa)$, $\lambda_N\in \blacktriangle_N$, and $\lambda_{N+M} \in \blacktriangle_{N+M}$, we have
\begin{align*}
    &\Ll|N\bar F_{N,\lambda_N}(t,q) - (N+M)\bar F_{N+M,\lambda_{N+M}}(t,q)\Rr|
    \\
    &\leq 4\Ll(|\kappa|_\infty C_\xi t + |\kappa|_\infty|q|_{L^1}+C_\mu\Rr)\Ll((N+M)\Ll|\lambda_N-\lambda_{N+M}\Rr|+ M|\sS|\Rr).
\end{align*}
\end{lemma}

\begin{proof}
Fixing any $t$ and $q$, we omit them from the notation of free energy. Let $I_N=(I_{N,s})_{s\in\sS}$ and $I_{N+M}=(I_{N+M,s})_{s\in\sS}$ be partitions associated with $\lambda_N$ and $\lambda_{N+M}$. Let us fix any partition $I'_{N+M}=(I'_{N+M,s})_{s\in\sS}$ of $\{1,\dots,M\}$ satisfying
\begin{align}\label{e.I_NsubsetI'_N+M}
    I_{N,s}\subset I'_{N+M,s},\quad\forall s\in\sS.
\end{align}
Let $\lambda'_{N+M}$ be the associated proportions. Before proceeding, we derive some simple estimates.
By~\eqref{e.I_NsubsetI'_N+M}, we get
\begin{align}\label{e.|I'_N+M,s|-|I_N,s|<M}
    |I'_{N+M,s}|-|I_{N,s}|\leq M,\quad\forall s\in\sS.
\end{align}
Using this, we also have
\begin{align}\label{e.|lambda'_N+M,s-lambda_N,s|<M/N+M}
    |\lambda'_{N+M,s}-\lambda_{N,s}|\stackrel{\eqref{e.lambda_N,s=}}{\leq } \frac{|I'_{N+M,s}|-|I_{N,s}|}{N+M} \stackrel{\eqref{e.I_NsubsetI'_N+M}}{\leq} \frac{M}{N+M}
\end{align}

We first compare $\bar F_{N,\lambda_N}$ with $\bar F_{N+M,\lambda'_{N+M}}$. Denote by $\sigma\in \Sigma^{N+M}$ the configuration appearing in $\bar F_{N+M,\lambda'_{N+M}}$. For every such $\sigma=(\sigma_{\bullet n})_{1\leq n\leq N+M}$, write $\bar\sigma = (\sigma_{\bullet n})_{1\leq n\leq N}$. By~\eqref{e.I_NsubsetI'_N+M} and the fact that $I_N$ is a partition of $\{1,\dots,N\}$, we have $\bar\sigma =(\sigma_{\bullet n})_{n\in \cup_{s\in\sS}I_{N,s}}$ and $\bar\sigma$ is the configuration appearing in $\bar F_{N,\lambda_N}$. 

Let $\sqrt{2}\mathbf{G}_1(\sigma,\alpha)$ and $\mathbf{D}_1(\sigma)$ (resp.\ $\sqrt{2}\mathbf{G}_0(\sigma,\alpha)$ and $\mathbf{D}_0(\sigma)$) collect respectively Gaussian terms and deterministic terms inside $\exp(\cdots)$ in $\bar F_{N,\lambda_N}$ (resp.\ $\bar F_{N+M,\lambda'_{N+M}}$). Notice that $\mathbf{G}_1$ and $\mathbf{D}_1$ depends on $\sigma$ only through $\bar \sigma$.
With this setup, we apply Corollary~\ref{c.Gaussian_interpolation} (with $\varphi(1)= N \bar F_{N,\lambda_N}$ and $\varphi(0)=(N+M)\bar F_{N+M,\lambda'_{N+M}}$). 
We can use~\eqref{e.EH_N(sigma)H_N(sigma')=} and~\eqref{e.cov_external_multi-sp} to compute the covariances of $\mathbf{G}_0$ and $\mathbf{G}_1$. Then, Corollary~\ref{c.Gaussian_interpolation} implies
\begin{align}
    &\Ll|N \bar F_{N,\lambda_N}- (N+M)\bar F_{N+M,\lambda'_{N+M}}\Rr|\notag
    \\
    &=\int_0^1\E\la Nt\xi(\bar R)-(N+M)t\xi(R) + q(\alpha\wedge\alpha')\cdot(N\bar R-(N+M)R) \ra_r\d r\label{e.int_0^1E<N-N+M>}
\end{align}
where we used the short hand $\bar R= R_{N,\lambda_N}(\bar \sigma,\bar\sigma')$ and $R = R_{N+\M ,\lambda'_{N+\M}}(\sigma,\sigma')$. Notice that $\bar R$ and $R$ depend on the partition $I_N$ and $I'_{N+M}$ respectively (see~\eqref{e.R_N,s=}).
By the definition of $\bar\sigma$, we have the following a.s.\ entry-wise bound 
\begin{align*}
    \Ll|N\bar R_s - (N+M)R_s\Rr|\leq \kappa_s \Ll(|I'_{N+M,s}|-|I_{N,s}|\Rr)\stackrel{\eqref{e.|I'_N+M,s|-|I_N,s|<M}}{\leq} \kappa_s M,\quad\forall s\in\sS.
\end{align*}
Using this and $|\bar R|\leq |\kappa|_\infty|\sS|$, we can get
\begin{align*}
    \Ll|\bar R - R\Rr|&\leq \frac{\Ll|N\bar R - (N+M)R\Rr|}{N+M}+ \frac{M|\bar R|}{N+M}
    \leq \frac{2|\kappa|_\infty M|\sS|}{N}.
\end{align*}
Using this and~\eqref{e.E<>_r=int}, we can bound the absolute value of the integrand in~\eqref{e.int_0^1E<N-N+M>} by 
\begin{align*}
    &\E \la Nt\Ll|\xi(\bar R)-\xi(R)\Rr| + Mt \Ll|\xi (R) \Rr|+ \Ll|q(\alpha\wedge\alpha')\Rr|\Ll(N\Ll|\bar R-R\Rr|+ M \Ll|R\Rr|\Rr)\ra_r 
    \\
    &\leq 3\Ll(tC_\xi+|q|_{L^1}\Rr)|\kappa|_\infty M|\sS|
\end{align*}
where $C_\xi$ only depends on $\xi$ over $\prod_{s\in\sS}[-1,+1]^{\kappa_s\times \kappa_s}$. Inserting this back to~\eqref{e.int_0^1E<N-N+M>}, we get
\begin{align*}
    \Ll|N \bar F_{N,\lambda_N}- (N+M)\bar F_{N+M,\lambda'_{N+M}}\Rr|\leq 3\Ll(tC_\xi+|q|_{L^1}\Rr)|\kappa|_\infty M|\sS|.
\end{align*}
On the other hand, Lemma~\ref{l.lambda_continuity} yields
\begin{align*}
    \Ll|\bar F_{N+M,\lambda'_{N+M}} - \bar F_{N+M,\lambda_{N+M}}\Rr|\leq \Ll(tC_\xi |\kappa|_\infty   +  |q|_{L^1}|\kappa|_\infty +C_\mu\Rr)\Ll|\lambda'_{N+M}-\lambda_{N+M}\Rr|.
\end{align*}
Notice that
\begin{align*}
    \Ll|\lambda'_{N+M}-\lambda_{N+M}\Rr| \stackrel{\eqref{e.|lambda'_N+M,s-lambda_N,s|<M/N+M}}{\leq } \Ll|\lambda_N-\lambda_{N+M}\Rr| + \frac{M|\sS|}{N+M}
\end{align*}
Combining the above two displays, we get the desired result.
\end{proof}

\section{Relation to vector spin glasses}\label{s.rel_to_vector_spin}

In this section, we show that if $\lambda_{N,s}$ is rational for every $s\in\sS$, we can equate $\bar F_{N,\lambda_N}$ to the free energy of a vector spin glass. We start by introducing the setting of vector spin glass models in Section~\ref{s.vector_spin_glass} and then prove Lemma~\ref{l.match_mp_with_vector} and Corollary~\ref{c.equiv_rational}.

\subsection{Vector spin glasses}\label{s.vector_spin_glass}
Let $\D\in\N$ be the dimension for vector-valued spins and let $P^\vec_1$ be a finite nonnegative measure supported on a compact set in $\R^\D$. We view $P^\vec_1$ as the distribution for a single spin. For each $N\in\N$, a spin configuration with size $N$ is denoted by $\bsigma = (\bsigma_{dn})_{1\leq d\leq \D,\, 1\leq n\leq N}$. Column vectors $\bsigma_{\bullet n}$ are individual spins in $\bsigma$. We sample $\bsigma$ by independently drawing each $\bsigma_{\bullet n}$ according to $P_1$. More precisely, denoting the distribution of $\bsigma$ as $P^\vec_N$, we have
\begin{align}\label{e.P^vec_N=}
    \d P^\vec_N(\bsigma) = \otimes_{n=1}^N \d P_1^\vec(\bsigma_{\bullet n}).
\end{align}
Given a smooth function $\bxi :\R^{\D\times\D}\to \R$, for each $N\in\N$, we assume the existence of a centered Gaussian process $\Ll(H^\vec_N(\bsigma)\Rr)_{\bsigma\in \R^{\D\times N}}$ with covariance
\begin{align}\label{e.EH^vec_N(sigma)H^vec_N(sigma')=}
    \E H^\vec_N(\bsigma)H^\vec_N(\bsigma') = N \bxi\Ll(\frac{\bsigma\bsigma'^\intercal}{N}\Rr).
\end{align}
We interpret $\frac{\bsigma\bsigma'^\intercal}{N}$ as the $\R^{\D\times\D}$-valued overlap between configurations $\bsigma$ and $\bsigma'$.
For $\bq\in \mcl Q_\infty(\D)$ (see Section~\ref{s.matrices,inner_products,paths}), conditioned on $\fR$, let $(\bw^\bq(\alpha))_{\alpha\in\supp\fR}$ be the $\R^\D$-valued centered Gaussian process with covariance
\begin{align}\label{e.E[bwbw]=}
    \E \bw^\bq(\alpha) \bw^\bq(\alpha') = \bq(\alpha\wedge\alpha').
\end{align}
Conditioned on $\fR$, for each $i\in\{1,\dots,N\}$, let $\bw^\bq_i$ be independent copies of $\bw^\bq$. Then, we set
\begin{align}\label{e.W^q_N(alpha)=}
    W^\bq_N(\alpha) = \Ll(\bw^\bq_1(\alpha),\dots, \bw^\bq_N(\alpha)\Rr),\quad\forall \alpha\in\supp\fR.
\end{align}
We view $W^\bq_N(\alpha)$ as an $\R^{\D\times N}$-valued process with column vectors $\bw^\bq_i$.
For each $N\in\N$, $t\in\R_+$, and $\bq\in\mcl Q_\infty(\D)$, we consider the Hamiltonian and free energy:
\begin{gather}
    H^{\vec,t,\bq}_N(\bsigma,\alpha) = \sqrt{2t}H^\vec_N(\bsigma) - N \bxi\Ll(\bsigma\bsigma^\intercal/N\Rr)
     + \sqrt{2}W^\bq_N(\alpha)\cdot \bsigma - \bq(1)\cdot \bsigma\bsigma^\intercal, \label{e.H^vec,t,q_N=}
    \\
    \bar F_N^\vec(t,\bq) = - \frac{1}{N}\E \log \iint \exp\bigg(  H^{\vec,t,\bq}_N(\bsigma,\alpha)\bigg) \d P^\vec_N(\bsigma)\d\fR(\alpha),\label{e.F^vec(t,q)=}
\end{gather}
where the expectation is first taken over Gaussian randomness in $H^\vec_N$ and $W^\bq_N$ and then over the randomness in $\fR$.

\subsection{Reduction}

For $\M\in \N$, we call a collection $(\sfM_s)_{s\in\sS}$ of subsets a \textbf{weak partition} of $\{1,\dots,\M\}$ if $\cup_{s\in\sS}\sfM_s = \{1,\dots,\M\}$ and $ \sfM_s\cap \sfM_{s'}=\emptyset$ whenever $s\neq s'$. This differs from the standard notion in that we allow $\sfM_s$ to be empty.
We work with the multi-species spin glass with system size $\M N$ for $N\in\N$ and with species proportion satisfy
\begin{align}\label{e.lambda^MN_s=|M_s|/M}
    \lambda_{\M N,s} = |\sfM_s|/\M,\quad\forall s \in\sS
\end{align}
for some weak partition $(\sfM_s)_{s\in\sS}$ of $\{1,\dots ,\M\}$.
Under this assumption, among $\M N$ spins of the multi-species configuration $\sigma$, there are exactly $|\sfM_s|N$ spins belonging to the $s$-species for each $s\in\sS$.

We want to map this model to a vector spin model with spins in $\R^{\Delta}$ and size $N$, where
\begin{align}\label{e.Delta=}
    \Delta = \Delta(\M,\lambda_{\M N}) = \sum_{s\in\sS}\lambda_{\M N,s}\M \kappa_s = \sum_{s\in\sS}|\sfM_s|\kappa_s.
\end{align}
Notice that $\Delta$ does not depend on $N$ due to the form of $\lambda_{\M N,s}$ in~\eqref{e.lambda^MN_s=|M_s|/M}.
We reorder $\{1,2,\dots,\Delta\}$ by fixing an arbitrary bijection:
\begin{align}\label{e.lexico}
    \{1,2,\dots, \Delta\} \quad\stackrel{}{\longleftrightarrow}\quad \cup_{s\in\sS}\Ll\{(m,k):\: m\in \sfM_s,\, k\in\{1,\dots,\kappa_s\}\Rr\}.
\end{align}
Under this re-ordering, we view $a\in \R^{\Delta\times \Delta}$ and $b\in \R^{\Delta\times N}$ as $a=(a_{(m,k)(m',k')})_{m,m'; k,k'}$ and $b=(b_{(m,k)n})_{m, k,n}$. For $a\in \R^{\Delta\times \Delta}$, we use the notation, for $m\in \sfM_s$ and $m'\in\sfM_{s'}$
\begin{align}\label{e.a(m,bullet)(m',bullet)=}
    a_{(m,\bullet)(m',\bullet)} = \Ll(a_{(m,k)(m',k')}\Rr)_{1\leq k\leq\kappa_s,\, 1\leq k\leq \kappa_{s'}}\in \R^{\kappa_s\times \kappa_{s'}}.
\end{align}

We can fix a bijection $\sigma\mapsto \bsigma$ from $\Sigma^{\M N}$ (see~\eqref{e.Sigma}) to $[-1,+1]^{\Delta \times N}$ with the property:
\begin{align}\label{e.sigma_to_bsigma_bijection}
    \forall i \in I_{N,s},\ \exists m\in \sfM_s,\ \exists n\in\{1,\dots,N\}:\quad \sigma_{\bullet i}\mapsto \bsigma_{(m,\bullet)n}=(\bsigma_{(m,k)n})_{1\leq k\leq \kappa_s}.
\end{align}
In words, every single spin in the $s$-species is mapped to a sub-column vector crossing $\kappa_s$ rows with indices $(m,k)$, for $k\in\{1,\dots,\kappa_s\}$, for some $m \in\sfM_s$. Under such a bijection, we have the following correspondence of overlaps (recall the overlap~\eqref{e.R_N,s=} in the multi-species setting):
\begin{align}\label{e.R_N,s=(matrix_overlap)}
    R_{\M N,\lambda_{\M N},s}(\sigma,\sigma')= (\M N)^{-1}\sum_{m\in \sfM_s}\Ll(\bsigma\bsigma'^\intercal\Rr)_{(m,\bullet)(m,\bullet)},\quad\forall s\in\sS,
\end{align}
where we have used the notation in~\eqref{e.a(m,bullet)(m',bullet)=} and each entry in $\Ll(\bsigma\bsigma'^\intercal\Rr)_{(m,\bullet)(m,\bullet)}$ is given by
\begin{align}\label{e.bsigmabisgma'=}
    \Ll(\bsigma\bsigma'^\intercal\Rr)_{(m,k)(m,k')}= \sum_{n=1}^N\bsigma_{(m,k)n}\bsigma'_{(m,k')n},\quad \forall k,k'\in\{1,\dots,\kappa_s\}.
\end{align}
This is the usual matrix multiplication in view of the re-ordering~\eqref{e.lexico}.

We describe the corresponding single vector spin distribution.
Let $P^\vec_1$ be the finite nonnegative measure supported on $[-1,1]^{\Delta}$ given by
\begin{align}\label{e.P^vec_1=mp}
    \d P^\vec_1(\bsigma) = \otimes_{s\in\sS}\otimes_{m\in\sfM_s} \d\mu_{s}\Ll(\bsigma_{(m,\bullet)1}\Rr)
\end{align}
Suggested by~\eqref{e.R_N,s=(matrix_overlap)}, for $\xi$ in~\eqref{e.EH_N(sigma)H_N(sigma')=}, we take
\begin{align}\label{e.bxi=mp}
    \bxi(a) = \M\xi \bigg(\Big(M^{-1}\sum_{m\in \sfM_s}a_{(m,\bullet)(m,\bullet)}\Big)_{s\in\sS}\bigg),\quad\forall a \in \R^{\Delta\times\Delta},
\end{align}
Lastly, we describe the corresponding parameter for the external field.
To each $q\in\mcl Q_\infty^\sS(\kappa)$ (see~\eqref{e.Q^S(kappa)=}), we associate $\bq\in \mcl Q_\infty(\Delta)$ given by
\begin{align}\label{e.bq=mp}
     \bq_{(m,\bullet)(m',\bullet)}=0,\quad\forall m\neq m';\qquad \bq_{(m,\bullet)(m,\bullet)} = q_s,\quad\forall m \in \sfM_s.
\end{align}

\begin{lemma}[Relation between multi-species and vector spin glass]\label{l.match_mp_with_vector}
Let $\M\in\N$ and let $(\sfM_s)_{s\in\sS}$ be a weak partition of $\{1,\dots,\M\}$.
Correspondingly, let $P^\vec_1$ and $\bxi$ be given as in \eqref{e.P^vec_1=mp} and~\eqref{e.bxi=mp} and let $\Ll(\bar F^\vec_N\Rr)_{N\in\N}$ be the associated free energy given in~\eqref{e.F^vec(t,q)=}.
In particular, the vector spin is $\Delta$-dimensional for $\Delta$ in~\eqref{e.Delta=}.
Then, for every $N\in\N$, $\lambda_{\M N}$ satisfying~\eqref{e.lambda^MN_s=|M_s|/M}, and $(t,q)\in\R_+\times \mcl Q_\infty^\sS(\kappa)$, we have
\begin{align}\label{e.F_MN=M^-1F_N}
    \bar F_{\M N,\,\lambda_{\M N}}(t,q) =\M^{-1} \bar F^\vec_N(t,\bq)
\end{align}
with $\bq$ given in~\eqref{e.bq=mp}.
\end{lemma}

\begin{proof}
We identify the two types of configurations $\sigma$ and $\bsigma$ via a bijection described as in~\eqref{e.sigma_to_bsigma_bijection}. Under this identification, we can see that $P_{\M N,\lambda_{\M N}}$ in~\eqref{e.P_N=} is the same as~$P^\vec_N$ in~\eqref{e.P^vec_N=} with $P^\vec_1$ in~\eqref{e.P^vec_1=mp}. Next, using~\eqref{e.R_N,s=(matrix_overlap)} and~\eqref{e.bxi=mp}, we have $N\bxi(\bsigma\bsigma'^\intercal/N)=\M N\xi(R_{\M N,\lambda_{\M N}}(\sigma,\sigma'))$. From this and the covariances~\eqref{e.EH_N(sigma)H_N(sigma')=} and~\eqref{e.EH^vec_N(sigma)H^vec_N(sigma')=}, we can deduce that $(H_{\M N}(\sigma))_\sigma \stackrel{\d}{=} (H^\vec_N(\bsigma))_\bsigma $. It remains to verify that the external fields have the same distribution. We can compute 
\begin{align*}
\begin{split}
    &\E W^q_{\M N}(\alpha,\sigma)W^q_{\M N}(\alpha',\sigma') 
    \stackrel{\eqref{e.cov_external_multi-sp}}{=} \M N\sum_{s\in\sS} q_s(\alpha\wedge\alpha')\cdot R_{\M N,\lambda_{\M N},s}(\sigma,\sigma')
    \\
    &\stackrel{\eqref{e.R_N,s=},\,\eqref{e.R_N,s=(matrix_overlap)}}{=} \sum_{s\in\sS} q_s(\alpha\wedge\alpha')\cdot\sum_{m\in \sfM_s} \Ll(\bsigma\bsigma'^\intercal\Rr)_{(d\bullet)(d\bullet)} \stackrel{\eqref{e.bq=mp}}{=} \bq(\alpha\wedge\alpha') \cdot\bsigma\bsigma'^\intercal.
\end{split}
\end{align*}
This implies that, under the identification~\eqref{e.sigma_to_bsigma_bijection} and conditioned on $\fR$, we have
\begin{align}\label{e.external_field_equal_distri}
    \Ll(W^q_{\M N}(\alpha,\sigma)\Rr)_{\sigma} \stackrel{\d}{=}\Ll(W^\bq_N(\alpha)\cdot \bsigma\Rr)_{\bsigma}
\end{align}
where the latter appears in $\bar F^\vec_N(t,\bq)$ (see~\eqref{e.F^vec(t,q)=}). The equivalence of distributions of various Gaussian terms also ensures that the self-overlap correction terms are the same in $\bar F_{\M N,\,\lambda_{\M N}}(t,q)$ and $\bar F^\vec_N(t,\bq)$. 
Hence, under the identification between $\sigma$ and $\bsigma$ as in~\eqref{e.sigma_to_bsigma_bijection}, we have thus verified
\begin{align}\label{e.H^t,q_N=(d)=}
    \Ll(H^{t,q}_{\M N}(\sigma,\alpha)\Rr)_{\sigma\in \Sigma^{\M N}} \stackrel{\d}{= } \Ll(H^{\vec, t,\bq}_{\M N}(\bsigma,\alpha)\Rr)_{\bsigma\in [-1,+1]^{\Delta\times N}}
\end{align}
where the two sides are given in~\eqref{e.H^t,q_N=} and~\eqref{e.H^vec,t,q_N=}.
In particular, this yields~\eqref{e.F_MN=M^-1F_N}. 
\end{proof}

For any $r\in\R$, write $\lceil r \rceil = \min\{n\in\N:n\geq r\}$.

\begin{corollary}[Equivalence in the rational case]\label{c.equiv_rational}
Assume that there are $M\in\N$ and a weak partition $(\sfM_s)_{s\in\sS}$ of $\{1,\dots,\M\}$ such that
\begin{align*}
    \lim_{N\to\infty}\lambda_{N,s} = |\sfM_s|/\M ,\quad\forall s \in \sS.
\end{align*}
Let $\bar F^\vec_N$ be the free energy with $P^\vec_1$ and $\bxi$ specified in~\eqref{e.P^vec_1=mp} and~\eqref{e.bxi=mp}. For every $t\in\R_+$ and $q\in \mcl Q_\infty^\sS(\kappa)$, let $\bq$ be given as in~\eqref{e.bq=mp} and we have
\begin{align*}
    \lim_{N\to\infty}\Ll|  \bar F_{N,\lambda_N}(t,q)- \M^{-1}\bar  F^\vec_{\lceil N/\M\rceil}(t,\bq)\Rr|=0.
\end{align*}
\end{corollary}

\begin{proof}
Fixing $t$ and $q$, we omit them and $\bq$ from the notation. For two sequences $(r_N)_{N\in\N}$ and $(r'_N)_{N\in\N}$ of real numbers, we write $r_N \approx r'_N$ provided $\lim_{N\to\infty}|r_N-r'_N|=0$. We need the following estimates for every $N\in\N$,
\begin{align}\label{e.0<[N/D]D-N<D}
    0\leq\lceil N/\M\rceil\M - N\leq \M,\qquad \Ll|\bar F_{N,\lambda_N}\Rr|\leq C
\end{align}
for some constant $C>0$. The latter estimate follows from Jensen's inequality applied to the expression of $\bar F_{N,\lambda_N}$ (see~\eqref{e.F_N(t,q)=}).
Set $\lambda_\infty = (|\sfM_s|/\M)_{s\in\sS}$.
Now, we can get
\begin{align*}
    \bar F_{N,\lambda_N} \stackrel{\text{L.\ref{l.lambda_continuity}}}{\approx} \bar F_{N,\lambda_\infty} \stackrel{\text{\eqref{e.0<[N/D]D-N<D}, L.\ref{l.NF_N-(N+M)F_(N+M)}}}{\approx} \bar F_{\lceil N/\M\rceil\M,\,\lambda_\infty} \stackrel{\text{L.\ref{l.match_mp_with_vector}}}{=} \M^{-1}\bar  F^\vec_{\lceil N/\M\rceil}
\end{align*}
which gives the desired result.
\end{proof}

\section{Analytic properties of the free energy}\label{s.analytic_properties}

We study analytic properties of the free energy, its limits (if exists), and the initial condition at $t=0$. In Section~\ref{s.diff_monot_precompt}, we show that the free energy is differentiable and monotone, and that any subsequence of $(\bar F_{N,\lambda_N})_{N\in\N}$ is precompact in the local uniform topology.
In Section~\ref{s.semi-concave}, we show that the free energy is locally semi-concave uniformly in $N$ and then deduce properties of any subsequential limit. In Section~\ref{s.initial_cond}, we identify $\bar F_{N,\lambda_N}(0,\cdot)$ and establish its regularity properties. Results in this section adapt those in~\cite[Sections~3 and 5]{HJ_critical_pts}.

\subsection{Differentiability, monotonicity, and precompactness}\label{s.diff_monot_precompt}

Recall the notation of $\mcl Q_\square(\D)$ in Section~\ref{s.matrices,inner_products,paths} and $\mcl Q_\square(\kappa)$ in~\eqref{e.Q^S(kappa)=} for spaces of paths.

Let $G$ be either $\mcl Q^\sS_2(\kappa)$, $\R_+\times \mcl Q^\sS_2(\kappa)$, or $\mcl Q_2(\D)$ for some $\D\in\N$. Let $L^2$ be the ambient Hilbert space for $G$. A function $g:G\to\R$ is said to be \textbf{Fr\'echet differentiable} at some $q\in G$ if there is a unique $y\in L^2$ such that
\begin{align}\label{e.def_Frechet}
    \lim_{r\to0}\sup_{\substack{q'\in G\setminus\{q\}\\ |q'-q|_{L^2}\leq r}}\frac{\Ll|g(q')-g(q)-\la y,q'-q\ra_{L^2}\Rr|}{|q'-q|_{L^2}} = 0.
\end{align}
In this case, we call $y$ the \textbf{Fr\'echet derivative} of $g$ at $q$.

For every $q\in G$, we define
\begin{align*}
    \mathrm{Adm}(G,q) = \Ll\{e\in L^2 \ \big|\  \exists r>0:\: \forall r'\in[0,r],\ q+re\in G \Rr\}
\end{align*}
to be set of directions along which a small line segment starting from $q$ belongs to $G$. We say that $g:G\to\R$ is \textbf{Gateaux differentiable} at $q\in G$ if
\begin{itemize}
    \item $g'(q,e)= \lim_{r\searrow0}\frac{g(q+re)-g(q)}{r}$ exists for every $e\in \mathrm{Adm}(G,q)$;
    \item there is a unique $y\in L^2$ such that $g'(q,e)=\la y, e\ra_{L^2}$ for every $e\in \mathrm{Adm}(G,q)$.
\end{itemize}
In this case, we call $y$ the \textbf{Gateaux derivative} of $g$ at $q$.

If $g$ is differentiable at $q$ in either of the two senses, we denote its derivative at $q$ by $\partial_q g(q)$, which is an element in $L^2$.

For any $\D\in\N$ and $u\in \R_+$, we define
\begin{align*}
    \mcl Q_{\infty,\leq u} (\D) = \Ll\{q\in \mcl Q_{\infty}(\D):\: |q(r)|\leq u,\quad\forall r\in[0,1)\Rr\}.
\end{align*}
Then, for any $\lambda =(\lambda_s)_{s\in\sS}\in \R_+^\sS$, we define
\begin{align}\label{e.Q^s_infty,<lambda(kappa)=}
    \mcl Q^\sS_{\infty,\leq \lambda}(\kappa) = \prod_{s\in\sS} \mcl Q_{\infty, \leq \lambda_s}(\kappa_s).
\end{align}
Recall $\bar F_{N,\lambda_N}$ in~\eqref{e.F_N(t,q)=} and $\la\cdot\ra_{N,\lambda_N}$ in~\eqref{e.<>_N=}.

\begin{proposition}[Differentiability of $\bar F_{N,\lambda_N}$]
\label{p.F_N_smooth}
Let $N \in \N$, $\lambda_N\in \blacktriangle_N$, and let $\bar F_{N,\lambda_N}$ be given as in~\eqref{e.F_N(t,q)=}. We have for every $t,t' \in \R_+$ and $q,q' \in \mcl Q_\infty^\sS(\kappa)$ that 
\begin{equation*}
\Ll|\bar F_{N,\lambda_N}(t,{q}) - \bar F_{N,\lambda_N}(t',{q'})\Rr|\leq \Ll|{q}-{q'}\Rr|_{L^1} + |t-t'| \, \sup_{|a| \le 1} |\xi(a)|.
\end{equation*}
In particular, the free energy in~\eqref{e.F_N(t,q)=} can be extended by continuity to $\R_+\times\mcl Q_1^\sS(\kappa)$.
Moreover, the restriction of the function $\bar F_{N,\lambda_N}$ to $\R_+ \times \mcl Q_2^\sS(\kappa)$ is Fréchet (and Gateaux) differentiable everywhere, jointly in its two variables. We denote its Fréchet (and Gateaux) derivative in $q$ by $\dr_q \bar F_{N,\lambda_N}(t,q) = \dr_q \bar F_{N,\lambda_N}(t,q, \cdot) \in L^2([0,1]; \prod_{s\in\sS}S^{\kappa_s})$. For every $t \ge 0$, we have, for every $q \in \mcl Q_2^\sS(\kappa)$,
\begin{equation}
\label{e.bounds.der.FN}
\partial_q \bar F_{N,\lambda_N}(t,q) \in \mcl Q^\sS_{\infty,\leq \lambda_N}(\kappa), \qquad \Ll|\partial_t  \bar F_{N,\lambda_N}(t,q) \Rr| \le \sup_{|a| \le 1} |\xi(a)|,
\end{equation}
and, for every $q \in \mcl Q_\infty^\sS(\kappa)$  and $\pi \in L^2([0,1]; \prod_{s\in\sS}S^{\kappa_s})$,
\begin{equation}
\label{e.def.der.FN}
\begin{split}
    \la \pi,\dr_q \bar F_{N,\lambda_N}(t,{q})\ra_\cH &= \E \la \pi \Ll(\alpha\wedge\alpha'\Rr)\cdot R_{N,\lambda_N}(\sigma,\sigma')\ra_{N,\lambda_N}
\\
\partial_t \bar F_{N,\lambda_N}(t,{q}) &= \E \la \xi \Ll(R_{N,\lambda_N}(\sigma,\sigma')\Rr)\ra_{N,\lambda_N}.
\end{split}
\end{equation}
Finally, for every $r \in [1,+\infty]$,  $t \ge 0$ and $q, q' \in \mcl Q_2^\sS(\kappa)$ with $q'-q \in L^r$, we have
\begin{equation}
\label{e.continuity.der.FN}
\Ll|\dr_q \bar F_{N,\lambda_N}(t,{q'}) -\dr_q \bar F_{N,\lambda_N}(t,{q})\Rr|_{L^r}\leq 16N\Ll|q'-{q}\Rr|_{L^r}.
\end{equation}
In particular, the mapping $q \mapsto \dr_q \bar F_{N,\lambda_N}(t,q)$ can be extended to $\mcl Q_1^\sS(\kappa)$ by continuity, and the properties in \eqref{e.bounds.der.FN} and \eqref{e.continuity.der.FN} remain valid with $q, q' \in \mcl Q_1^\sS(\kappa)$. 
\end{proposition}

The proof follows the same lines as those for~\cite[Proposition~5.1]{HJ_critical_pts}, which are based on the Gaussian interpolation arguments.
We only need to explain the first inclusion in~\eqref{e.bounds.der.FN}. This follows from the first line in~\eqref{e.def.der.FN}, the fact that $|R_{N,\lambda_N,s}(\sigma,\sigma')|\leq \lambda_{N,s}$ due to~\eqref{e.R_N,s=}, and the invariance of cascade in Lemma~\ref{l.invar}.

\begin{proposition}[Precompactness of $\bar F_{N,\lambda_N}$]\label{p.precompactness}
For every $r\in(1,+\infty]$ and any sequence $(\lambda_N)_{N\in\N}$ with $\lambda_N\in\blacktriangle_N$, any subsequence of $\Ll(\bar F_{N,\lambda_N}\Rr)_{N\in\N}$ has a further subsequence which converges uniformly on every bounded metric ball in $\R_+\times \mcl Q_r^\sS(\kappa)$.
\end{proposition}

This proposition can be proved in the same way as~\cite[Proposition~3.3]{HJ_critical_pts} using the following lemma.

\begin{lemma}[Compact embedding of paths]\label{l.compact_embed_paths}
Let $r\in(1,+\infty]$ and let $(q_n)_{n\in\N}$ be a sequence in $\mcl Q_r^\sS(\kappa)$ such that
\begin{align*}
    \sup_{N\in\N} \Ll|q_n\Rr|_{L^r}<+\infty.
\end{align*}
Then, there exists a subsequence $\Ll(q_{n_k}\Rr)_{k\in\N}$ and some $q\in\mcl Q_r^\sS(\kappa)$ such that, for every $r'\in [1,r)$, this subsequence converges almost everywhere on $[0,1]$ and in $L^{r'}$ to $q$.
\end{lemma}

This lemma is a straightforward adaption of~\cite[Lemma~3.4]{HJ_critical_pts} (a special case with $|\sS|=1$).

Next, we describe a monotonicity property of the free energy in the second variable. 
Given a nonempty closed convex cone $\mcl C$ in a Hilbert space $L^2$, we define the dual cone of $\mcl C^*$ by
\begin{align}\label{e.Q^*_2(kappa)}
    \mcl C^* = \Ll\{p\in L^2\ \big| \ \forall q \in \mcl C,\ \la p,q\ra_{L^2}\geq 0\Rr\}.
\end{align}
For each $s\in\sS$, the dual cone of $\mcl Q_2(\kappa_s)$ has a simple characterization given by~\cite[Lemma~3.5]{HJ_critical_pts} (see also~\cite[Lemma~3.4(2)]{chen2022hamilton}):
\begin{align*}
    \mcl Q_2(\kappa_s)^* = \Ll\{p\in L^2\ \big| \  \forall t\in[0,1), \ \int_t^1 p(r)\d r \in \S^{\kappa_s}_+\Rr\}.
\end{align*}
Since any $p\in \prod_{s\in\sS}\mcl Q_2(\kappa_s)^*$ satisfies exactly $\la p,q\ra_{L^2} = \sum_{s\in\sS}\la p_s,q_s\ra_{L^2}\geq 0$ for every $q \in \mcl Q_2^\sS(\kappa)$, we can identify
\begin{align*}
    \big(\mcl Q^\sS_2(\kappa)\big)^* = \prod_{s\in\sS}\mcl Q_2(\kappa_s)^*.
\end{align*}
For a subset $G$ of $L^2 ([0,1];\prod_{s\in\sS}\S^{\kappa_s})$ and a function $g:G\to\R$, we say that $g$ is \textbf{$\big(\mcl Q^\sS_2(\kappa)\big)^*$-increasing} if for every $q,q'\in G$, we have
\begin{align}\label{e.def_Q^*-increasing}
    q-q'\in \big(\mcl Q^\sS_2(\kappa)\big)^*\quad \Longrightarrow \quad g(q)\geq g(q').
\end{align}

\begin{proposition}[Monotonicity of free energy]\label{p.monotone}
For every $N\in\N$, $\lambda_N\in\blacktriangle_N$, and $t\in\R_+$, the function $\bar F_{N,\lambda_N}(t,\cdot)$ is $\big(\mcl Q^\sS_2(\kappa)\big)^*$-increasing.
\end{proposition}
\begin{proof}
Fix any $N$, $\lambda_N$, and $t$.
We apply Lemma~\ref{l.match_mp_with_vector} with $\M,N$ therein substituted with $N,1$ respectively. Also, we replace $\sfM_s$ in~\eqref{e.lambda^MN_s=|M_s|/M} by $I_{N,s}$ associated with $\lambda_N$. Then, Lemma~\ref{l.match_mp_with_vector} gives $\bar F_{N,\lambda_N}(t,q)= N^{-1} \bar  F^\vec_1(t,\bq)$ for every $q \in \mcl Q_\infty^\sS(\kappa)$ and $\bq \in \mcl Q_\infty(\Delta)$ given in~\eqref{e.bq=mp}. We are allowed by Proposition~\ref{p.F_N_smooth} to take $q \in\mcl Q_2^\sS(\kappa)$ instead of $\mcl Q_\infty^\sS(\kappa)$. Given any $q'$, we define $\bq'$ in the same was as in~\eqref{e.bq=mp}. 

Now, let us assume $q-q'\in \big(\mcl Q^\sS_2(\kappa)\big)^*$ and we want to show $\bq -\bq' \in \mcl Q_2(\Delta)^*$ defined as in~\eqref{e.Q^*_2(kappa)}. Let $\mathbf{p}$ be any element in $\mcl Q_2(\Delta)$. Using the notation in~\eqref{e.a(m,bullet)(m',bullet)=} and the form of $\bq$ and $\bq'$ given in~\eqref{e.bq=mp}, we can compute
\begin{align*}
    \la \mathbf{p},\,\bq -\bq'\ra_{L^2} = \sum_{n=1}^N \la \mathbf{p}_{(n,\bullet)(n,\bullet)} ,\ \bq_{(n,\bullet)(n,\bullet)}-\bq'_{(n,\bullet)(n,\bullet)}\ra_{L^2}.
\end{align*}
Since $\mathbf{p}$ is an increasing path in $\S^{\Delta}_+$ and each $\mathbf{p}_{(n,\bullet)(n',\bullet)}$ is a projection of $\mathbf{p}$ into a minor matrix, we have that $\mathbf{p}_{(n,\bullet)(n',\bullet)}$ is an increasing path in $\S^{\kappa_s}_+$ for $n\in I_{N,s}$ and thus $\mathbf{p}_{(n,\bullet)(n',\bullet)}\in \mcl Q_2(\kappa_s)$. From the second part in~\eqref{e.bq=mp}, we can see $\bq_{(n,\bullet)(n,\bullet)}-\bq'_{(n,\bullet)(n,\bullet)}\in \mcl Q_2(\kappa_s)^*$ (from $q-q'\in \big(\mcl Q^\sS_2(\kappa)\big)^*$). Therefore, by the definition of dual cones in~\eqref{e.Q^*_2(kappa)}, the right-hand side of the above display is nonnegative. Since $\mathbf{p}$ is arbitrary, we conclude that $\bq -\bq' \in \mcl Q_2(\Delta)^*$.

By the result in the setting of vector spin glasses~\cite[Proposition~3.6]{HJ_critical_pts}, we have that $\bar F^\vec_1(t,\cdot)$ is $\mcl Q_2(\Delta)^*$-nondecreasing defined in the same way as in~\eqref{e.def_Q^*-increasing}.
Therefore, $\bar F^\vec_1(t,\bq)\geq \bar F^\vec_1(t,\bq')$ and thus $\bar F_{N,\lambda_N}(t,q)\geq \bar F_{N,\lambda_N}(t,q')$, which gives the desired result.
\end{proof}

\subsection{Semi-concavity and consequences}\label{s.semi-concave}

Recall the definition of $\mcl Q^\sS_{\uparrow,c}(\kappa)$ from~\eqref{e.Q_uparrow,c(D)=} and~\eqref{e.Q^S(kappa)=}. For any increasing path $q$, we denote by $\dot q$ its distributional derivative.

\begin{proposition}[Semi-concavity of the free energy]\label{p.semi-concave}
There exists a constant $C<+\infty$ (depending only on $(\mu_s)_{s\in\sS}$, $(\kappa_s)_{s\in\sS}$, and $\xi$) such that, for every $N\in\N$, $\lambda_N\in\blacktriangle_N$, $c>0$, $t,t'\geq c$, $q,q'\in \mcl Q^\sS_{\uparrow,c}(\kappa)$ with $\dot q-\dot q' \in L^2$, and $r\in[0,1]$,
\begin{align}\label{e.semi-concave_cts_F_N}
\begin{split}
    (1-r)\bar F_{N,\lambda_N}(t,q)+ r \bar F_{N,\lambda_N}(t,'q') - \bar F_{N,\lambda_N}\Ll((1-r)(t,q)+r(t',q')\Rr) 
    \\
    \leq Cr(1-r)c^{-2}\Ll((t-t')^2+\Ll|\dot q-\dot q'\Rr|_{L^2}^2\Rr).
\end{split}
\end{align}
\end{proposition}

\begin{proof}
The argument is the same as that for~\cite[Propositions~3.7 and 3.8]{HJ_critical_pts}. Here, we sketch similar parts and highlight the differences.

By a density argument and Proposition~\ref{p.F_N_smooth}, it suffices to consider smooth $q$ and $q'$. We then approximate them by piece-wise constant paths. For each $K\in\N$ and $k\in\{0,\dots,K\}$, we set
\begin{align*}
    q_k = q\Big(\frac{k}{K+1}\Big),\qquad \forall k\in\{0,\dots,K\};\qquad q^K = \sum_{k=0}^K q_k \one_{\Ll[\frac{k}{K+1},\, \frac{k+1}{K+1}\Rr)} \in \mcl Q^\sS(\kappa).
\end{align*}
Recall $\ellipt$ from~\eqref{e.ellipt(a)=}.
Due to $q\in \mcl Q^\sS_{\uparrow,c}(\kappa)$, we have
\begin{align}\label{e.q_k>}
    \frac{c}{K+1}\identity_{\kappa_s}\leq q_{k,s}-q_{k-1,s}, \quad \ellipt(q_{k,s}-q_{k-1,s})\leq c^{-1},\quad\forall k\in \{0,\dots,K\},\ s\in\sS.
\end{align}
Here, $\identity_{\kappa_s}$ is the $\kappa_s\times\kappa_s$ identity matrix.
Similarly, we construct $\Ll(q'_k\Rr)_{0\leq k\leq K}$ and $q'^K$, which satisfy an analogous version of~\eqref{e.q_k>}. Also set $q_{-1}=q'_{-1}=0$. We claim that for a constant $C$ as announced in the statement, we have
\begin{align}\label{e.semi-concavity_discrete_F_N}
\begin{split}
    (1-r)\bar F_{N,\lambda_N}\Ll(t,q^K\Rr)+ r \bar F_{N,\lambda_N}(t,'q'^K) - \bar F_{N,\lambda_N}\Ll((1-r)\Ll(t,q^K\Rr)+r\Ll(t',q'^K\Rr)\Rr) 
    \\
    \leq Cr(1-r)c^{-2}\Big((t-t')^2+(K+1)\sum_{k=0}^K \Ll|(q_k-q_{k-1})-(q'_k-q'_{k-1})\Rr|^2\Big).
\end{split}
\end{align}
Then, we can use approximation arguments by sending $K\to\infty$ as in the proof of~\cite[Proposition~3.8]{HJ_critical_pts} (below (3.34)) to get~\eqref{e.semi-concave_cts_F_N}.

In the remainder, we explain the proof of~\eqref{e.semi-concavity_discrete_F_N}.
For preparation, we need some properties on the matrix square root. For every $h\in \S^\D_{++}$ (see Section~\ref{s.matrices,inner_products,paths}) and $a\in\S^\D$ for any $\D\in\N$, we define the first and second order derivatives of matrix square root:
\begin{align*}
    \mathcal{D}_{\sqrt{h}}(a) = \lim_{\eps\to 0}\eps^{-1}\Ll(\sqrt{h+\eps a}-\sqrt{h}\Rr),\qquad \mathcal{D}^2_{\sqrt{h}}(a) = \lim_{\eps\to0}\eps^{-1}\Ll(\mathcal{D}_{\sqrt{h+\eps}}(a)-\mathcal{D}_{\sqrt{h+\eps}}(a)\Rr).
\end{align*}
Denoting by $h_\mathrm{min}$ the smallest eigenvalue of $h$. Then,~\cite[(3.23) and (3.24)]{HJ_critical_pts} give
\begin{align}\label{e.der_square_root_matrix}
    \Ll|\mathcal{D}_{\sqrt{h}}(a)\Rr|\leq |a|h_\mathrm{min}^{-\frac{1}{2}}/2,\qquad \Ll|\mathcal{D}^2_{\sqrt{h}}(a)\Rr|\leq |a|^2h_\mathrm{min}^{-\frac{3}{2}}/4.
\end{align}
To prove~\eqref{e.semi-concavity_discrete_F_N}, we need to bound the Hessian of $\bar F_{N,\lambda_N}$ after a change of variables. As in the proof of~\cite[Proposition~3.7]{HJ_critical_pts}, we focus on the semi-concavity in the second variable for the brevity of presentation. Henceforth, we omit the first variable from the notation by fixing some $t$ and writing $\bar F_{N,\lambda_N}(\cdot) = \bar F_{N,\lambda_N}(t,\cdot)$. 
For simplicity, we write $\mathbf{S}= \prod_{s\in\sS} S^{\kappa_s}$ and $\mathbf{S}_+ = \prod_{s\in\sS} S^{\kappa_s}_+$.
We view $\mathbf{S}^{K+1}$ as a subset of the linear space $\Ll(\prod_{s\in\sS}\R^{\kappa_s\times\kappa_s}\Rr)^{K+1}$ with inner product given by the entry-wise dot product as in~\eqref{e.dot_product}.

When considering $\bar F_{N,\lambda_N}$ over paths of the form $q^K$, we can think of $\bar F_{N,\lambda_N}$ as a function of $(q_k)_{0\leq k\leq K}\in \mathbf{S}^{K+1}_+$. 
For each $x =(x_k)_{0\leq k\leq K}\in\mathbf{S}^{K+1}_+$, we write $\sqrt{x}=\Ll(\sqrt{x_k}\Rr)_{0\leq k\leq K}$ and
$\sqrt{x_k} = \Ll(\sqrt{x_{k,s}}\Rr)_{s\in\sS}$ with each $x_k= (x_{k,s})_{s\in\sS} \in \mathbf{S}_+$.
Then, we define a function $G_N: \mathbf{S}^{K+1}\to \R$ through the relation
\begin{align}\label{e.G_N=}
    \bar F_{N,\lambda_N}\Ll(q^K\Rr) =\bar F_{N,\lambda_N}\Ll((q_k)_{0\leq k\leq K}\Rr) = G_N\Ll(\Ll(\sqrt{q_k-q_{k-1}}\Rr)_{0\leq k\leq K}\Rr).
\end{align}
We consider the function $\mathbf{S}^{K+1}\ni y\mapsto G_N(y)$. For each $k\in \{0,\dots,K\}$ and $s\in\sS$, we define $\partial_{y_{k,s}} G_N(y) \in \S^{\kappa_s}$ via the relation: for every $a\in\S^{\kappa_s}$,
\begin{align*}
    a\cdot \partial_{y_{k,s}} G_N(y) = \frac{\d}{\d\eps}G_N(y_1,\dots,y_{k-1},y_k+\eps \bar a,y_{k+1},\dots,y_K)\Big|_{\eps =0}
\end{align*}
where $\bar a = (\bar a_{s'})_{s'\in\sS}$ satisfies $\bar a_{s'}=0$ when $s'\neq s$, and $\bar a_{s}=a$.
The Hessian $\nabla^2G_N(y)$ viewed as a linear map from $\mathbf{S}^{K+1}$ to $\mathbf{S}^{K+1}$ is defined through: for every $\mathbf{a}\in \mathbf{S}^{K+1}$,
\begin{align*}
    \mathbf a\cdot \nabla^2G_N(y)\mathbf a=\frac{\d^2}{\d \eps^2} G_N(y+\mathbf{a})\Big|_{\eps=0}.
\end{align*}
By the same computation for~\cite[(3.28) and (3.29)]{HJ_critical_pts}, we can find a constant $C$ as in the statement of the proposition such that
\begin{align}\label{e.bound_on_der_G_N}
    \Ll|\partial_{y_{k,s}}G_N(y)\Rr|\leq C|y_{k,s}|,\qquad \mathbf{a}\cdot \nabla^2 G_N(y)\mathbf{a}\leq C|\mathbf{a}|^2
\end{align}
for every $k$ and every $\mathbf{a}$. The computation is basic and involves Gaussian interpolation technique as in Corollary~\ref{c.Gaussian_interpolation} and Jensen's inequality.

Next, we introduce the change of variable: for every $x\in\S^{K+1}_+$, we set $\tilde G_N(x) = G_N(\sqrt{x})$.
We can make sense of the Hessian of $\tilde G_N$ in the same way as above.
As in the second step of the proof of~\cite[Proposition~3.7]{HJ_critical_pts}, we can compute, for every $\mathbf{a}\in\mathbf{S}^{K+1}$ 
\begin{align*}
    \mathbf{a}\cdot \tilde \nabla^2G_N(x)\mathbf{a}
    = \mathbf{D}\cdot \nabla^2 G_N(\sqrt{x})\cdot \mathbf{D} + \sum_{k,\, s} \mathbf{D}_{k,s}^2
    \cdot \partial_{y_{k,s}}G_N(\sqrt{x})
\end{align*}
where
\begin{align*}
    \mathbf{D}= \Ll(\mathcal{D}_{\sqrt{x_{k,s}}}(a_{k,s})\Rr)_{0\leq k\leq K,\, s\in\sS}\in\mathbf{S}^K,\qquad \mathbf{D}_{k,s}^2 =\mathcal{D}^2_{\sqrt{x_{k,s}}}(a_{k,s}).
\end{align*}
Then, we get
\begin{align*}
    \mathbf{a}\cdot \tilde \nabla^2G_N(x)\mathbf{a}&\stackrel{\eqref{e.bound_on_der_G_N}}{\leq} C\Ll|\mathbf{D}\Rr|^2 + C\sum_{k=0}^K \Ll|\mathbf{D}^2_k\Rr|\Ll|\sqrt{x_k}\Rr|
    \\
    &\stackrel{\eqref{e.der_square_root_matrix}}{\leq} \frac{C}{4}\sum_{k,\,s}|a_{k,s}|^2\Ll(\Ll(x_{k,s}\Rr)_\mathrm{min}\Rr)^{-1} + \frac{C}{4}\sum_{k,\,s}|a_{k,s}|^2\Ll(\Ll(x_{k,s}\Rr)_\mathrm{min}\Rr)^{-3/2}\Ll|\sqrt{x_{k,s}}\Rr|
\end{align*}
It is easy to see that
\begin{align*}
    \Ll|\sqrt{x_{k,s}}\Rr|^2\leq \kappa_s \ellipt(x_{k,s})(x_{k,s})_\mathrm{min}.
\end{align*}
Combining the above two displays, we have that for every $x\in\mathbf{S}^{K+1}$ satisfying
\begin{align}\label{e.x_k>}
    \frac{c}{K+1}\identity_{\kappa_s}\leq x_{k,s}, \quad \ellipt(x_{k,s})\leq c^{-1},\quad\forall k\in \{0,\dots,K\},\ s\in\sS,
\end{align}
we have
\begin{align*}
     \mathbf{a}\cdot \tilde \nabla^2G_N(x)\mathbf{a} \leq C(K+1)c^{-2}|a|^2
\end{align*}
where we have absorbed $\kappa_s$ into $C$. Hence, for every $x$ satisfying~\eqref{e.x_k>} and $x'$ satisfying an analogous version of~\eqref{e.x_k>}, we have
\begin{align*}
    (1-r)\tilde G_N(x) + r \tilde G_N(x') \leq \tilde G_N\Ll((1-r)x+rx'\Rr)+ Cr(1-r)c^{-1}(K+1)\Ll|x-x'\Rr|^2.
\end{align*}
In view of~\eqref{e.G_N=}, we have $\bar F_{N,\lambda_N}\Ll(q^K\Rr)=\tilde G_N\Ll((q_k-q_{k-1})_{0\leq k\leq K})\Rr)$. Recall that in the relation~\eqref{e.semi-concavity_discrete_F_N} to prove, both $q$ and $q'$ satisfy~\eqref{e.q_k>}. Therefore, we are allowed to use the above display to deduce~\eqref{e.semi-concavity_discrete_F_N}.
\end{proof}

As a consequence of the semi-concavity, we have the following result.
Recall the definition of $\mcl Q^\sS_\uparrow(\kappa)$ from~\eqref{e.Q_uparrow(D)=} and~\eqref{e.Q^S(kappa)=}. 

\begin{proposition}[Convergence of derivatives]\label{p.cvg_der_semi_concave}
Suppose that $\bar F_{N,\lambda_N}$ converges pointwise to some limit $f$ along a subsequence $(N_n)_{n\in\N}$.
\begin{enumerate}
    \item For every $t\in\R_+$, if $f(t,\cdot)$ is Gateaux differentiable at $q\in\mcl Q^\sS_\uparrow(\kappa)$, then $\partial_q \bar F_{N_n,\lambda_{N_n}}(t,q,\cdot)$ converges in $L^r$ to $\partial_q f(t,q,\cdot)$ for every $r\in[1,+\infty)$.
    \item For every $q\in\mcl Q^\sS_1(\kappa)$, if $f(\cdot,q)$ is differentiable at $t>0$, then $\partial_t \bar F_{N_n,\lambda_{N_n}}(t,q)$ converges to $\partial_t f(t,q)$.
\end{enumerate}
\end{proposition}

The proof is the same as that for \cite[Proposition~5.4]{HJ_critical_pts} based on the semi-concavity proved in Proposition~\ref{p.semi-concave}.

\begin{remark}[Convergence of derivatives for perturbed free energy]\label{r.cvg_der_semi_concave}
Let $(N_n)_{N\in\N}$ be a strictly increasing sequence in $\N$. For each $n\in\N$, let $\M_n$ satisfy $N_n \in \M_n \N$, let $\lambda_{N_n} \in \blacktriangle_{\M_n}$ (see~\eqref{e.simplex}), let $x_n$ be any perturbation parameter from~\eqref{e.x_perturb_parameter} (the space of which depends on $\M_n$). Then, we consider the perturbed free energy $\bar F^{\M_n,x_n}_{N_n,\lambda_{N_n}}$ in~\eqref{e.F^x_N=} or $\tilde F^{\M_n,x_n}_{N_n,\lambda_{N_n}}$ in~\eqref{e.tilde_F^x_N=}. For them, we can also prove Propositions~\ref{p.semi-concave} and Proposition~\ref{p.cvg_der_semi_concave} (along any further subsequence of $(N_n)_{n\in\N}$). The proofs are the same.
\end{remark}

Next, we state the result on the regularity of limits of the free energy. In the statement below, we use the notion of ``Gaussian null sets'' which is a natural generalization of ``Lebesgue null sets'' to infinite dimensions. We refer to~\cite[Definition~4.2]{HJ_critical_pts} for the exact definition. The only property important to us here is that the complement of a Gaussian null set is dense, which allows us to say that any limit of free energy is differentiable on a dense set.

\begin{proposition}[Regularity of the limit]\label{p.reg_limit}
Suppose that $(\bar F_{N,\lambda_N})_{N\in\N}$ converges pointwise to some $f$ along some subsequence. Then, for every $r\in(1,\infty]$, the function $\bar F_{N,\lambda_N}$ converges locally uniformly along this subsequence. 
The limit $f$ satisfies the same Lipschitz, monotone, and local semi-concave properties of $\bar F_{N,\lambda_N}$ in Propositions~\ref{p.monotone}, \ref{p.F_N_smooth}, and~\ref{p.semi-concave}. Moreover,
\begin{itemize}
    \item for each $t\geq0$, there is a Gaussian null set $\mcl N_t$ of $L^2([0,1],\prod_{s\in\sS}\S^{\kappa_s})$ such that $f(t,\cdot) :\mcl Q^\sS_2(\kappa)\to\R$ is Gateaux differentiable at every point in $\mcl Q^\sS_2(\kappa)\setminus \mcl N_t$ and $\mcl Q^\sS_{\infty,\uparrow}\setminus\mcl N_t$ is dense in $\mcl Q^\sS_2(\kappa)$;
    \item there is a Gaussian null set $\mcl N$ of $\R\times L^2([0,1],\prod_{s\in\sS}\S^{\kappa_s})$ such that $f:\R_+\times \mcl Q^\sS_2(\kappa)\to\R$ is Gateaux differentiable on $(\R_+\times \mcl Q^\sS_2(\kappa))\setminus \mcl N$ and $(\R_+\times \mcl Q^\sS_{\infty,\uparrow}(\kappa))\setminus \mcl N$ is dense in $\R_+\times \mcl Q^\sS_2(\kappa)$.
\end{itemize}
\end{proposition}

This corresponds to~\cite[Proposition~5.3]{HJ_critical_pts} and can be deduced by the same argument. Hence, we omit the proof here.

\subsection{Initial condition}\label{s.initial_cond}

We want to describe the initial condition $\bar F_{N,\lambda_N}(0,\cdot)$.
For any finite positive measure $\nu$ on $[-1,1]^{\D}$ for some $\D\in\N$ and for $\bq\in \mcl Q_\infty(\D)$, define
\begin{align}\label{e.psi^vec=}
    \psi^\vec_\nu(\bq) = - \E \log\iint \exp\Ll(\sqrt{2}\bw^\bq(\alpha)\cdot \bsigma-\bq(1)\cdot \bsigma\bsigma^\intercal\Rr)\d \nu(\bsigma)\d \fR(\alpha)
\end{align}
where $\bw^\bq$ is the Gaussian process with covariance given in~\eqref{e.E[bwbw]=}.
This is related to the initial condition in the vector spin glass. Indeed, comparing with the expression of the free energy in~\eqref{e.F^vec(t,q)=}, we have
\begin{align*}
    \bar F^\vec_N(0,\bq) = \psi^\vec_{P^\vec_1}(\bq).
\end{align*}
This identity holds clearly for $N=1$ and the general case follows from a standard property of the cascade measure (see~\cite[Proposition~3.2]{HJ_critical_pts}).

In the multi-species setting, given $\lambda=(\lambda_s)_{s\in\sS} \in \R_+^\sS$, we define, for $q\in \mcl Q^\sS_\infty(\kappa)$,
\begin{align}\label{e.psi=}
    \psi_{\lambda}(q) =\sum_{s\in\sS}\lambda_s \psi^\vec_{\mu_s}(q_s)
\end{align}
where $(\mu_s)_{s\in\sS}$ are the fixed distributions of spins of different species (see Section~\ref{s.spin}).

We recall the result~\cite[Corollary~5.2]{HJ_critical_pts} on the regularity of the initial condition~\eqref{e.psi^vec=} in the vector spin model.
For $\bq \in \mcl Q_\infty(\D)$ and a positive measure $\nu$ on $\R^\D$, we introduce the following Gibbs measure:
\begin{align*}
    \la\cdot\ra^\vec_{\nu,\bq} = \exp\Ll(\bw^\bq(\alpha)\cdot\sigma-\bq(1)\cdot\sigma\sigma^\intercal\Rr) \d \nu(\sigma)\d \fR(\alpha)
\end{align*}
where $\bw^\bq(\alpha)$ is given as in~\eqref{e.E[bwbw]=}.

\begin{lemma}[Regularity of vector-spin initial condition]\label{l.psi^vec_smooth}
For any $\D\in\N$ and any positive measure $\nu$ supported on $[-1,+1]^\D$, the function $\psi^\vec_\nu$ given in~\eqref{e.psi^vec=} can be extended to $\mcl Q_1(\D)$ and satisfies
\begin{align*}
    \Ll|\psi^\vec_\nu(\bq)-\psi^\vec_\nu(\bq')\Rr|\leq |\bq-\bq'|_{L^1},\quad\forall \bq,\bq'\in\mcl Q_1(\D).
\end{align*}
The restriction $\psi^\vec_\nu:\mcl Q_2(\D)\to\R$ is Fréchet (and Gateaux) differentiable everywhere; we denote its Fréchet (and Gateaux) derivative by $\dr_\bq \psi^\vec_\nu(\bq) = \dr_\bq \psi^\vec_\nu(\bq, \cdot) \in L^2([0,1]; S^D)$. We have, for every $\bq \in \mcl Q_2(\D)$,
\begin{equation}
\label{e.bound.der.psi}
\partial_\bq\psi^\vec_\nu(\bq) \in \mcl Q_{\infty,\leq 1}(\D),
\end{equation}
and, for every $\bq \in \mcl Q_\infty(\D)$ and $\pi \in L^2([0,1]; S^D)$,
\begin{equation*}
\la \pi,\dr_\bq \psi^\vec_\nu(\bq)\ra_\cH = \E \la \pi \Ll(\alpha\wedge\alpha'\Rr)\cdot \sigma\sigma'^\intercal\ra^\vec_{\nu,\bq}.
\end{equation*}
Moreover, for every $r \in [1,+\infty]$ and $\bq, \bq' \in \mcl Q_2(\D)$ with $\bq-\bq' \in L^r$, we have
\begin{equation}
\label{e.continuity.der.psi}
\Ll|\dr_\bq \psi^\vec_\nu(\bq) -\dr_\bq \psi^\vec_\nu(\bq')\Rr|_{L^r}\leq 16\Ll|\bq-\bq'\Rr|_{L^r}.
\end{equation}
In particular, the mapping $\bq \mapsto \dr_\bq \psi^\vec_\nu(\bq)$ can be extended to $\mcl Q_1(\D)$ by continuity, and the properties in \eqref{e.bound.der.psi} and \eqref{e.continuity.der.psi} remain valid with $\bq, \bq' \in \mcl Q_1(\D)$.  
\end{lemma}

\begin{lemma}[Multi-species initial condition]\label{l.reg_initial}
For any $\lambda\in \R_+^\sS$, $\psi_\lambda$ given in~\eqref{e.psi=} can be extended to $ Q^\sS_1(\kappa)$ and satisfies
\begin{align*}
    \Ll|\psi_\lambda(q)-\psi_\lambda(q')\Rr|\leq \sum_{s\in\sS}\lambda_{s} |q_s-q'_s|_{L^1},\quad\forall q,q'\in Q^\sS_1(\kappa).
\end{align*}
The restriction $\psi_\lambda:\mcl Q^\sS_2(\kappa)\to \R$ is F\'echet (and Gateaux) differentiable everywhere. At every $q\in \mcl Q^\sS_2(\kappa)$, its derivative is given by
\begin{align*}
    \partial_q \psi_\lambda(q) = \Ll(\lambda_s \partial_{q_s} \psi^\vec_{\mu_s}(q_s)\Rr)_{s\in\sS} \in \mcl Q^\sS_{\infty,\leq\lambda}(\kappa).
\end{align*}
\end{lemma}

\begin{proof}
The extendability and Lipschitzness follows from~\eqref{e.psi=} and Lemma~\ref{l.psi^vec_smooth}. The expression for the derivative follows from~\eqref{e.psi=} and the definition of derivatives. The range $\mcl Q^\sS_{\infty,\leq\lambda}(\kappa)$ is clear from~\eqref{e.bound.der.psi}.
\end{proof}

\begin{lemma}[Initial condition]\label{l.initial_condition}
For every $N\in\N$, $q\in \mcl Q_1^\sS(\kappa)$, and $\lambda_N\in\blacktriangle_N$, we have
\begin{align*}
    \bar F_{N,\lambda_N} (0,q)=\psi_{\lambda_N}(q).
\end{align*}
\end{lemma}

\begin{proof}
We only need to prove identity at $q\in\mcl Q^\sS_\infty(\kappa)$ and the general case follows by extension.
For each $s\in\sS$ and $n\in I_{N,s}$, we define
\begin{align*}
    X_{s,n}(\alpha) = \log \int \exp \Ll(\sqrt{2}w^{q_s}_n(\alpha)\cdot \sigma_{\bullet n}-q_s(1)\cdot \sigma_{\bullet  n}\sigma_{\bullet n}^\intercal\Rr) \d \mu_s(\sigma_{\bullet n})
\end{align*}
where $w^{q_s}_n$ is introduced in~\eqref{e.Ew^q_s_iw^q_s_i=}. 
Comparing this with $\psi^\vec_\nu$ in~\eqref{e.psi_mu_s(q_s)=-Elogint...}, we have
\begin{align}\label{e.psi_mu_s(q_s)=-Elogint...}
    \psi^\vec_{\mu_s}(q_s) = - \E \log \int \exp \Ll(X_{s,n}(\alpha)\Rr)\d\fR(\alpha).
\end{align}
On the other hand, using the expression in~\eqref{e.F_N(t,q)=}, we can rewrite
\begin{align*}
    -N \bar F_{N,\lambda_N}(0,q) = \E \log \int \exp\bigg(\sum_{s\in\sS}\sum_{n\in I_{N,s}} X_{s,n}(\alpha)\bigg) \d \fR(\alpha).
\end{align*}
By a basic property of cascades stated for example in~\cite[Corollary~5.26]{HJbook}, the right-hand side in the above display is equal to
\begin{align*}
    \sum_{s\in\sS}\sum_{n\in I_{N,s}}\E \log \int \exp \Ll(X_{s,n}(\alpha)\Rr)\d\fR(\alpha) \stackrel{\eqref{e.psi_mu_s(q_s)=-Elogint...}}{=} -\sum_{s\in\sS}\Ll|I_{N,s}\Rr|\psi^\vec_{\mu_s}(q_s)
\end{align*}
which together with the definition of $\lambda_N$ in~\eqref{e.lambda_N,s=} implies the announced result.
\end{proof}

\section{Cavity computation and proofs in the general case}\label{s.cavity_and_proofs}

In Section~\ref{s.rel_to_vector_spin}, we have shown that, if the limit of species proportions $(\lambda_{N,s})_{s\in\sS}$ are all rational, the multi-species model is equivalent to the vector spin glass model (see Corollary~\ref{c.equiv_rational}). Hence, we can directly apply cavity computation results in~\cite[Sections~6 and~7]{HJ_critical_pts} stated for vector spin glasses. However, in the case where some entries in the limit of $\lambda_N$ are irrational, such argument no longer works. To handle this, we need an additional approximation argument. 

In Section~\ref{s.cavity_computation}, we give definitions of various objects appearing in the cavity computation and then state the cavity computation results, Lemmas~\ref{l.cavity1} and~\ref{l.cavity2}, assuming that the species proportions are rational. These two lemmas are straightforward adaptions of results in~\cite{HJ_critical_pts}.
In Section~\ref{s.approx}, we consider the (general) irrational case and use approximation to extend the two results to Lemmas~\ref{l.approx1} and~\ref{l.approx2}. In Section~\ref{s.crti_pt_and_rel_results}, we apply these two lemmas to prove the results corresponding to those in~\cite[Section~7]{HJ_critical_pts}. In particular, we prove Theorems~\ref{t.crit_pt_rep},~\ref{t.crit_pt_id}, and~\ref{t.main3} here.

\subsection{Cavity computation}\label{s.cavity_computation}

We start with introducing definitions necessary for the cavit computation.

\subsubsection{Parisi functional}

Recall the definition of $\psi_{\lambda}$ from~\eqref{e.psi=}.
For $\xi$ in~\eqref{e.EH_N(sigma)H_N(sigma')=}, we define $\theta: \prod_{s\in\sS}\R^{\kappa_s\times\kappa_s}\to \R$ by
\begin{align}\label{e.theta=}
    \theta(a)= a\cdot \nabla\xi (a) -\xi(a),\quad\forall a \in \prod_{s\in\sS}\R^{\kappa_s\times\kappa_s}.
\end{align}
Here, $\nabla\xi :\prod_{s\in\sS}\R^{\kappa_s\times\kappa_s}\to \prod_{s\in\sS}\R^{\kappa_s\times\kappa_s} $ is the gradient of $\xi$ defined with respect to the entry-wise inner product structure on the linear space $\prod_{s\in\sS}\R^{\kappa_s\times\kappa_s}$.

For $t\in\R_+$, $q\in \mcl Q^\sS_\infty(\kappa)$, and $\lambda \in \R_+^\sS$, we set
\begin{align}\label{e.sP_lambda,t,q}
    \sP_{\lambda,t,q}(p) = \psi_\lambda(q+t \nabla\xi(p))-t\int_0^1\theta(p(r)) \d r.
\end{align}
Comparing this with~\eqref{e.mcJ=}, we have
\begin{align}\label{e.rel_parisi_mcJ}
    \sP_{\lambda_\infty, t,q}(p)= \mcl J_{\lambda_\infty, t,q}(q+t\nabla \xi(p), p).
\end{align}

\subsubsection{Hamiltonian from the cavity computation}

Then, we describe the Gibbs measure appearing in the cavity computation and the corresponding free energy. Fix any $M\in\N$ interpreted as the number of cavity spins. The Hamiltonian arising from the cavity computation is a centered Gaussian process $\big(\tilde H_N(\sigma)\big)_{\sigma \in \Sigma^N}$ with covariance
\begin{align}\label{e.E[tildeH_N...]var=}
    \E \Ll[\tilde H_N(\sigma)\tilde H_N(\sigma')\Rr]= (N+M)\xi \Big(\frac{N}{N+M}R_{N,\lambda_N}(\sigma,\sigma')\Big)
\end{align}
which should be compared with $H_N(\sigma)$ given in~\eqref{e.EH_N(sigma)H_N(sigma')=}. Notice that we have kept $M$ implicit from the notation. We assume that $\tilde H_N(\sigma)$ is independent from other randomness. Let $\tilde W^q_N(\sigma,\alpha)$ be an independent copy of $W^q_N(\sigma,\alpha)$ given in~\eqref{e.W^q_N(sigma,alpha)=}. 
Analogous to $H^{t,q}_N(\sigma,\alpha)$ in~\eqref{e.H^t,q_N=}, we define
\begin{align}\label{e.tildeH^t,q_N(sigma,lapha)}
\begin{split}
    \tilde H^{t,q}_N(\sigma,\alpha) = \sqrt{2t}\tilde H_N(\sigma) - t(N+M)\xi \Ll(\frac{N}{N+M}R_{N,\lambda_N}(\sigma,\sigma)\Rr) 
    \\
    + \tilde W^q_N(\sigma,\alpha) -q(1)\cdot R_{N,\lambda_N}(\sigma,\sigma).
\end{split}
\end{align}
We prefer to omit $\M$ in the notation.

\subsubsection{Perturbations}
To ensure Ghirlanda--Guerra identities, we need to introduce perturbation terms to the Hamiltonian. Since we want to apply Lemma~\ref{l.match_mp_with_vector} and results from~\cite{HJ_critical_pts} for vector spin glasses, the perturbation to be introduced below is not the most suitable choice. Indeed, in the multi-species setting, we only need Ghirlanda--Guerra identities for the overlap array $\big(R_{\M N,\lambda_{\M N}}(\sigma^l,\sigma^{l'})\big)_{l,l'\in\N}$ (see~\eqref{e.R_N,s=}) which contains less entries than $\big(N^{-1}\bsigma^l \bsigma^{l'}\big)_{l,l'\in\N}$ in the matching vector spin glass model (see~\eqref{e.R_N,s=(matrix_overlap)}). However, to use results from~\cite{HJ_critical_pts}, we need to employ perturbation in the style of vector spin glasses. In Remark~\ref{r.perturbation}, we describe the most suitable perturbation in the multi-species setting and corresponding results without proofs.

Henceforth, we fix any $\M\in\N$ and consider the multi-species model with size $\M N$ for $N\in\N$. We assume that $\lambda_{\M N}$ satisfies~\eqref{e.lambda^MN_s=|M_s|/M} for some fixed weak partition $(\sfM_s)_{s\in\sS}$ of $\{1,\dots,\M\}$.
As in~\eqref{e.Delta=}, let $\Delta = \Delta(\M,\lambda_{\M N})$
be the dimension for the matching vector spin model (which is independent of $N$). As in~\eqref{e.sigma_to_bsigma_bijection}, fix any such bijection (for each $N$) and let $\bsigma$ be the image of $\sigma$ after this mapping.

Let $(r_n)$ be an enumeration of $[0,1]\cap \Q$ and $(a_n)_{n\in\N}$ be an enumeration of elements in $\S^\Delta_+\cap \Q^{\Delta\times\Delta}$. Conditioned on $\fR$, for every $N\in\N$ and every $h=(h_i)_{1\leq i\leq 4}\in \N^4$, let $\Ll(H^h_{\M N}(\sigma,\alpha)\Rr)_{\sigma\in\Sigma^N,\, \alpha\in\supp\fR}$ be an independent centered Gaussian process with covariance
\begin{align}\label{e.E[H^hH^h]=}
    \E \Ll[H^h_{\M N}(\sigma,\alpha)H^h_{\M N}(\sigma',\alpha')\Rr] = N \Ll(a_{h_1}\cdot \Ll(N^{-1}\bsigma\bsigma'^\intercal)\Rr)^{\odot h_2}+\lambda_{h_3} \alpha\wedge\alpha'\Rr)^{h_4}.
\end{align}
Here, $\odot$ denotes the Schur product, namely, $a\odot b= (a_{ij}b_{ij})_{i,j}$ for two matrices $a$ and $b$ of the same dimension.
The existence of such a process is explained in~\cite[Section~6.1.1]{HJ_critical_pts}. For each $h\in\N^4$, we write $|h|_1 =\sum_{i=1}^4 h_i$ and let $c_h>0$ be a constant such that
\begin{align}\label{e.c_h<}
    c_h \sqrt{N^{-1} \E\Ll[H^h_{\M N}(\sigma,\alpha)H^h_{\M N}(\sigma,\alpha)\Rr]} \leq 2^{-|h|_1},
\end{align}
uniformly over $\sigma \in\Sigma^N$, $\alpha\in\supp\fR$, and $N\in\N$. 
Fix an orthonormal basis $\mathsf{bas}$ of $\S^\Delta$. 
We define the space of perturbation parameters:
\begin{align}\label{e.pert=}
    \pert(\M,\lambda_{\M N}) = [0,3]^{\N^4\times \mathsf{bas}}
\end{align}
where the dependence on $\lambda_{\M N}$ and $\M$ is through $\Delta$ as in~\eqref{e.Delta=}.
For every
\begin{align}\label{e.x_perturb_parameter}
    x= \Ll((x_h)_{h\in\N^4},\ (x_e)_{e\in\mathsf{bas}}\Rr)\in \pert(\M,\lambda_{\M N}),
\end{align}
we set
\begin{align}\label{e.H^x_N(sigma,alpha)}
    H^x_{\M N}(\sigma,\alpha) = \sum_{h\in\N^4}x_hc_hH^h_{\M N}(\sigma,\alpha) + \frac{1}{|\mathsf{bas}|} \sum_{e\in \mathsf{bas}}x_e e\cdot \bsigma\bsigma^\intercal.
\end{align}
Compared with the standard perturbation as in~\cite{pan.multi,pan.potts,pan.vec}, the additional second sum ensures the concentration of the self-overlap $N^{-1}\bsigma\bsigma'^\intercal$ (which implies the concentration of $R_{\M N,\lambda_{\M N}}(\sigma,\sigma)$).
Recall the original Hamiltonian $H^{t,q}_{\M N}(\sigma,\alpha)$ from~\eqref{e.H^t,q_N=} (with $\M N$ substituted for $N$ therein).
For each $N\in\N$, we define the free energy with perturbation
\begin{align}\label{e.F^x_N=}
\begin{split}
    &\bar F^{\M,x}_{\M N,\lambda_{\M N}}(t,q) 
    \\
    &=-\frac{1}{{\M N}}\E \log \iint\exp\Ll(H^{t,q}_{\M N}(\sigma,\alpha) +  N^{-\frac{1}{16}}H^x_{\M N}(\sigma,\alpha)\Rr)\d P_{\M N,\lambda_{\M N}}(\sigma) \d \fR(\alpha).
\end{split}
\end{align}
The choice of $\frac{1}{16}$ is inconsequential and can be replaced by any number in $(\frac{1}{32},\frac{1}{8})$. This factor is needed to ensure that the perturbation is weak enough not to change the limit of free energy and strong enough to ensure the validity of the Ghirlanda--Guerra identities. For each $N\in\N$, we denote the associated Gibbs measure by
\begin{align}\label{e.<>_N,lambda,x}
    \la \cdot\ra^{\M,x}_{\M N,\lambda_{\M N},x}\propto \exp\Ll(H^{t,q}_{\M N}(\sigma,\alpha) + N^{-\frac{1}{16}}H^x_{\M N}(\sigma,\alpha)\Rr)\d P_{\M N,\lambda_{\M N}}(\sigma) \d \fR(\alpha)
\end{align}
where the value for $(t,q)$ will be clear from the context.

Let $\tilde H^{t,q}_{\M N}(\sigma,\alpha)$ be given as in~\eqref{e.tildeH^t,q_N(sigma,lapha)} with $\M N$ substituted for $N$ therein. For each $N\in\N$, we denote the associated perturbed free energy and Gibbs measure by
\begin{gather}
    \begin{split}
        &\tilde F^{\M,x}_{\M N,\lambda_{\M N}}(t,q) \\&= -\frac{1}{{\M N}}\E \log \iint\exp\Ll(\tilde H^{t,q}_{\M N}(\sigma,\alpha) + N^{-\frac{1}{16}}H^x_{\M N}(\sigma,\alpha)\Rr)\d P_{\M N,\lambda_{\M N}}(\sigma) \d \fR(\alpha);
    \end{split}
     \label{e.tilde_F^x_N=}
    \\
    \la \cdot\ra^{\circ,\M,x}_{{\M N},\lambda_{\M N}}\propto \exp\Ll(\tilde H^{t,q}_{\M N}(\sigma,\alpha) + N^{-\frac{1}{16}}H^x_{\M N}(\sigma,\alpha)\Rr)\d P_{\M N,\lambda_{\M N}}(\sigma) \d \fR(\alpha). \label{e.<>^circ=}
\end{gather}
The symbol $\circ$ in the superscript signifies ``cavity''.

In the notation for the free energy and the Gibbs measure in~\eqref{e.F^x_N=}, \eqref{e.<>_N,lambda,x}, \eqref{e.tilde_F^x_N=}, and~\eqref{e.<>^circ=}, we have emphasized the dependence on $\M$ (the cavity dimension) in the superscript.

We denote by $(\sigma,\alpha)$ the canonical random variable under $\la\cdot\ra^{\M,x}_{{\M N},\lambda_{\M N}}$ and $\la\cdot\ra^{\circ,\M,x}_{{\M N},\lambda_{\M N}}$. We write $\Ll(\sigma^l,\alpha^l\Rr)_{l\in\N}$ to denote independent copies of $(\sigma,\alpha)$.

\begin{lemma}\label{l.F-F^x}
Let $M\in\N$ and let $\lambda_{MN}$ satisfy~\eqref{e.lambda^MN_s=|M_s|/M}.
There is a constant $C>0$ depending only $\kappa$, $\mu$, and $\xi$ such that for every $N$, $t\in\R_+$, and $q\in \mcl Q^\sS_\infty(\kappa)$, we have
\begin{gather}
    \sup_{x\in\pert(\M,\lambda_{\M N})} \Ll|\bar F_{\M N,\lambda_{\M N}}(t,q)- \bar F^{\M,x}_{\M N,\lambda_{\M N}}(t,q)\Rr| \leq C N^{-\frac{1}{16}}, \label{e.|F-F^x|<}
    \\
    \sup_{x\in\pert(\M,\lambda_{\M N})} \Ll|\bar F^{\M,x}_{\M N,\lambda_{\M N}}(t,q)- \tilde F^{\M,x}_{\M N,\lambda_{\M N}}(t,q)\Rr| \leq C|t|N^{-1}. \label{e.|F^x-tF^x|<}
\end{gather}
\end{lemma}
\begin{proof}
Comparing $\bar F_{\M N,\lambda_{\M N}}$ given in~\eqref{e.F_N(t,q)=} and $\bar F^{\M, x}_{\M N,\lambda_{\M N}}$ in~\eqref{e.F^x_N=}, we see that the additional term is $N^{-\frac{1}{16}}H^x_{\M N}(\sigma,\alpha)$.
The definition of $\Delta$ in~\eqref{e.Delta=} implies that $|\Delta|\leq \M |\kappa|_\infty$ for $|\kappa|_\infty = \max_{s\in\sS} \kappa_s$.
Since every entry in $\bsigma \in \R^{\Delta\times N}$ lies in $[-1,+1]$, we have $|\bsigma\bsigma'|\leq \M N |\kappa|_\infty$.
Using this, the choice of $c_h$ in~\eqref{e.c_h<}, and the presence of $|\mathsf{bas}|$ in~\eqref{e.H^x_N(sigma,alpha)}, we can apply the interpolation argument in Lemma~\ref{l.Gaussian interpolation} to get~\eqref{e.|F-F^x|<}. 

Comparing $\bar F^{\M, x}_{\M N,\lambda_{\M N}}$ in~\eqref{e.F^x_N=} and $\tilde F^{\M, x}_{\M N,\lambda_{\M N}}$ in~\eqref{e.tilde_F^x_N=}, the difference lies in the terms associated with $\sqrt{2t}H_{\M N}(\sigma)$ and $\sqrt{2t}\tilde H_{\M N}(\sigma)$. The variance of the former is given in~\eqref{e.EH_N(sigma)H_N(sigma')=} and latter in~\eqref{e.E[tildeH_N...]var=}. There is a constant $C_\xi$ depending on $\xi$ such that
\begin{align*}
    \bigg|(\M N+\M)\xi \Big(\frac{\M N}{\M N+\M}R_{\M N,\lambda_{\M N}}(\sigma,\sigma')\Big)- \M N\xi \Big(R_{\M N,\lambda_{\M N}}(\sigma,\sigma')\Big)\bigg|\leq C_\xi \M.
\end{align*}
Using this and the interpolation argument in Lemma~\ref{l.Gaussian interpolation}, we can get~\eqref{e.|F^x-tF^x|<}.
\end{proof}

\begin{remark}
We clarify that the law of $\Ll(\sigma^l,\alpha^l\Rr)_{l\in\N}$ under $\la\cdot\ra^{\M,x}_{\M N,\lambda_{\M N}}$ or $\la\cdot\ra^{\circ,\M,x}_{\M N,\lambda_{\M N}}$ depends on the partition $(I_{\M N,s})_{s\in\sS}$ of $\{1,\dots, N\}$. But, the law of overlaps
\begin{align*}
    \big(R_{\M N,\lambda_{\M N}}(\sigma^l,\sigma^{l'}),\,\alpha^l\wedge\alpha^{l'}\big)_{l,l'\in\N}
\end{align*}
(see~\eqref{e.R_N,s=}) under these Gibbs measures only depend on the proportions $\lambda_N$. Since we are more interested in the law of overlaps, we prefer to display in the notation only the dependence on $\lambda_{\M N}$. \qed
\end{remark}

Since we have introduced several Gibbs measures, it is a good place state the following important property (e.g.\ see~\cite[Proposition~4.8]{HJ_critical_pts}). Recall the notation $\la\cdot\ra_\fR$ from Section~\ref{s.cascade}.

\begin{lemma}[Invariance of cascades]\label{l.invar}
Let $\la\cdot\ra$ be the one of the following Gibbs measures: $\la\cdot\ra_{N,\lambda_N}$ in~\eqref{e.<>_N=}, $\la\cdot\ra^{\M,x}_{\M N,\lambda_{\M N}}$ in~\eqref{e.<>_N,lambda,x}, $\la\cdot\ra^{\circ,\M,x}_{\M N,\lambda_{\M N}}$ in~\eqref{e.<>^circ=}, or any interpolation of these appearing in Lemma~\ref{l.Gaussian interpolation}. Then, the law of $\big(\alpha^l\wedge\alpha^{l'}\big)_{l,l'\in\N}$ under $\E \la\cdot\ra$ is equal to that under $\E\la\cdot\ra_\fR$.
In particular, $\alpha\wedge\alpha'$ distributes uniformly over $[0,1]$ under $\E\la\cdot\ra$.
\end{lemma}

Later, we need the next result for the convergence of overlap arrays.

\begin{lemma}[Criterion for convergence of overlap array]\label{l.cvg_overlap}
Let $(\alpha^l)_{l\in\N}$ be i.i.d.\ samples from $\la\cdot\ra_\fR$. Let $E$ be a fixed compact Euclidean set and let $(\Omega_n,\nu_n)_{n\in\N}$ be a sequence of probability spaces such that, for each $n$, the following holds:
\begin{itemize}
    \item $(\alpha^l)_{l\in\N}$ are random variables on $\Omega_n$ and the law of $\big(\alpha^l\wedge\alpha^{l'}\big)_{l,l'\in\N}$ under $\nu_n$ is equal to that under $\E\la\cdot\ra_\fR$;
    \item there is a random array $R_n=\big(R^{l,l'}_n\big)_{l,l'}$ of $E$-valued random variables on $\Omega_n$ and the law of $R_n$ under $\nu_n$ is invariant under permutation of labels.
\end{itemize}
Let $a\in E$ and $p:[0,1]\to E$ be bounded and measurable. Then, we have
\begin{align}\label{e.l.cvg_overlap}
    \lim_{n\to\infty} \nu_n \Ll(\Ll|R^{1,2}_n - p\Ll(\alpha^1\wedge\alpha^2\Rr)\Rr|\Rr) =0,\qquad \lim_{n\to\infty} \nu_n \Ll(\Ll|R^{1,1}_n - a\Rr|\Rr)=0,
\end{align}
if and only if $\big(R_n^{l,l'},\, \alpha^l\wedge\alpha^{l'}\big)_{l,l'\in\N}$ under $\nu_n$ converges in law to
\begin{align*}
    \Big(\Ll(p \big(\alpha^l\wedge\alpha^{l'}\big),\ \alpha^l\wedge\alpha^{l'}\Rr)\one_{l\neq l'} + (a,1)\one_{l=l'}\Big)_{l,l'\in\N}
\end{align*}
under $\E\la\cdot\ra_\fR$, as $n$ tends to infinity.
\end{lemma}

\begin{proof}
We first show ``$\Longrightarrow$''.
Since the convergence in law is defined to be convergence over finitely many entries, by considering test functions of the form of a product of functions of one entry, it suffices to prove convergence for each entry.
The convergence of any diagonal entry ($l=l'$) is obvious. For non-diagonal entry, by symmetry, we only need to consider the one with index $(1,2)$. Let $g:\R^d\times \R\to\R$ be any Lipschitz function. Then, writing $R=R^{1,2}_n$ and $Q=\alpha^1\wedge\alpha^2$, we have
\begin{align*}
    \Ll|\nu_n\Ll(g(R,Q)\Rr)- \E \la g(p(Q),Q)\ra_\fR\Rr|= \Ll|\nu_n\Ll(g(R,Q)\Rr)- \nu_n\Ll(g(p(Q),Q)\Rr)\Rr| 
    \\
    \leq \|g\|_\mathrm{Lip}\nu_n \Ll(|R-p(Q)|\Rr)
\end{align*}
which converges to zero by~\eqref{e.l.cvg_overlap}. This completes the proof of ``$\Longrightarrow$''. To see ``$\Longleftarrow$'', we first notice that it is sufficient to prove~\eqref{e.l.cvg_overlap} with $|\cdot|$ replaced by $|\cdot|^2$ (random variables are bounded because $E$ is compact). Then, we can expand the square and use the convergence in law to get the desired result.
\end{proof}

\subsubsection{Two results from the cavity computation}
Recall the definition of the discrete simplex $\blacktriangle_N$ in~\eqref{e.simplex}.
Denote by $\overline{\mathrm{conv}}$ the operator taking the closed convex hull of some set in a finite-dimensional linear space. Define
\begin{align}\label{e.mcl_K=}
    \mcl K = \prod_{s\in\sS}\overline{\mathrm{conv}}\Ll\{\tau\tau^\intercal:\:\tau \in \supp \mu_s\Rr\}.
\end{align}
Recall the definition of $\mcl Q^\sS_{\infty,\leq \lambda}(\kappa)$ from~\eqref{e.Q^s_infty,<lambda(kappa)=} and the space of perturbation parameters in~\eqref{e.pert=}.

\begin{lemma}\label{l.cavity1}
Let $\M\in\N$ and $\lambda^\star \in \blacktriangle_\M$. Set $\lambda_{\M N}=\lambda^\star$ for every $N$.
For every $(t,q)\in\R_+\times \mcl Q^\sS_\infty(\kappa)$, there are sequences $\Ll(N^\pm_k\Rr)_{k\in\N}$ of strictly increasing integers, $\Ll(x^\pm_k\Rr)_{k\in\N}$ of parameters in $\pert(\M,\lambda^\star)$, $p_\pm \in \mcl Q^\sS_{\infty,\leq \lambda^\star}(\kappa)$, and $a_\pm\in\mcl K$ satisfying $a_\pm\geq p_\pm$ such that
\begin{enumerate}
    \item \label{i.l.cavity1.1} $\big(R_{\M N^\pm_k,\lambda^\star}\big(\sigma^l,\sigma^{l'}\big),\, \alpha^l\wedge\alpha^{l'}\big)_{l,l'\in\N}$ under $\E \la \cdot\ra^{\circ,\M,x^\pm_k}_{\M N^\pm_k,\lambda^\star}$ converges in law to
    \begin{align*}
        \Big(\Ll(p_\pm \big(\alpha^l\wedge\alpha^{l'}\big),\ \alpha^l\wedge\alpha^{l'}\Rr)\one_{l\neq l'} + (a_\pm,1)\one_{l=l'}\Big)_{l,l'\in\N}
    \end{align*}
    under $\E\la\cdot\ra_\fR$, as $k$ tends to infinity;

    \item\label{i.l.cavity1.2} we have 
    \begin{align*}
        \sP_{\lambda^\star,t,q}(p_-)\leq \liminf_{N\to\infty} \bar F_{\M N,\lambda^\star}(t,q) \leq \limsup_{N\to\infty} \bar F_{\M N,\lambda^\star}(t,q) \leq \sP_{\lambda^\star,t,q}(p_+).
    \end{align*}
\end{enumerate}
\end{lemma}

The two Gibbs measures appearing in Part~\eqref{i.l.cavity1.1} are given in~\eqref{e.<>^circ=} and Section~\ref{s.cascade}, respectively.
In the statement, $a_\pm\geq p_\pm$ and more generally $a\geq p$ means
\begin{align}\label{e.a>p}
    a_s-p_s(r) \in \S^{\kappa_s}_+,\quad\forall r\in[0,1),\ s\in\sS.
\end{align}

\begin{proof}
We fix any weak partition $(\sfM_s)_{s\in\sS}$ of $\{1,\dots,\M\}$ satisfying $\lambda^\star = (|\sfM_s|/\M)_{s\in\sS}$. Hence, the assumption $\lambda_{\M N}=\lambda^\star$ ensures~\eqref{e.lambda^MN_s=|M_s|/M}. 
Fix any $(t,q)$. Let $\bq$ be given as in~\eqref{e.bq=mp} and $\bar F^\vec_N$ be given by Lemma~\ref{l.match_mp_with_vector}. We directly define $\bar F^{\vec,x}_N(t,\bq) = M \bar F^{\M,x}_{\M N, \lambda_{\M N}}(t,q)$ and $\tilde F^{\vec,x}_N(t,\bq) = M \tilde F^{\M,x}_{\M N, \lambda_{\M N}}(t,q)$. Allowed by the bijection between $\sigma$ and $\bsigma$ in~\eqref{e.sigma_to_bsigma_bijection}, we can view $\la\cdot\ra^{\M,x}_{\M N,\lambda_{\M N}}$ (see~\eqref{e.<>_N,lambda,x}) and $\la\cdot\ra^{\circ,\M,x}_{\M N,\lambda_{\M N}}$ (see~\eqref{e.<>^circ=}) as the Gibbs measure associated with $\bar F^{\vec,x}_N(t,\bq)$ and $\tilde F^{\vec,x}_N(t,\bq)$, respectively.

Recall that in the proof of Lemma~\ref{e.H^t,q_N=(d)=}, we used various identities to derive the equivalence~\eqref{e.H^t,q_N=(d)=} between the multi-species Hamiltonian and the one in the vector spin glass model. As a consequence, we can match the corresponding free energies. Similar arguments together with the definition of $H^x_{\M N}(\sigma,\alpha)$ in~\eqref{e.H^x_N(sigma,alpha)} can be used to match $\bar F^{\vec,x}_N(t,\bq)$ and $\tilde F^{\vec,x}_N(t,\bq)$ exactly~\cite[(6.6) and (6.8)]{HJ_critical_pts}. Similarly, the Gibbs measures $\la\cdot\ra^{\M,x}_{\M N,\lambda_{\M N}}$ and $\la\cdot\ra^{\circ,\M,x}_{\M N,\lambda_{\M N}}$ match exactly~\cite[(6.7) and (6.9)]{HJ_critical_pts} (denoted by $\la\cdot\ra_{N,x}$ and $\la\cdot\ra^\circ_{N,x}$ therein).

Hence, applying~\cite[Corollary~6.11]{HJ_critical_pts} (with $M$ therein set to be $1$), we get that there are $(N^\pm_k)_{k\in\N}$, $(x^\pm_k)_{k\in\N}$, $\bp_\pm \in \mcl Q_{\infty,\leq 1}(\Delta)$, and $\ba_\pm \in \overline{\mathrm{conv}}\Ll\{\bsigma\bsigma^\intercal:\: \bsigma\in \supp P^\vec_1\Rr\}$ (see~\eqref{e.P^vec_1=mp}) satisfying $\ba_\pm\geq \bp_\pm$ such that
\begin{enumerate}[label=(\roman*)]
    \item \label{i.c6.11.1} $\big(N^{-1}\bsigma^l\big(\bsigma^{l'}\big)^\intercal,\, \alpha^l\wedge\alpha^{l'}\big)_{l,l'\in\N}$ under $\E \la\cdot\ra^{\circ,\M,x^\pm_k}_{\M N^\pm_k,\lambda^\star}$ converges in law to
    \begin{align*}
        \Big(\Ll(\bp_\pm \big(\alpha^l\wedge\alpha^{l'}\big),\ \alpha^l\wedge\alpha^{l'}\Rr)\one_{l\neq l'} + (\ba_\pm,1)\one_{l=l'}\Big)_{l,l'\in\N}
    \end{align*}
    under $\E\la\cdot\ra_\fR$, as $k$ tends to infinity;
    \item \label{i.c6.11.2} we have 
    \begin{align*}
        \sP^\vec_{t,q}(\bp_-)\leq \liminf_{N\to\infty} \bar F^\vec_N(t,\bq) \leq \limsup_{N\to\infty} \bar F^\vec_N(t,\bq) \leq \sP^\vec_{t,q}(\bp_+).
    \end{align*}
\end{enumerate}
Here, the functional (see~\cite[(6.15)]{HJ_critical_pts}) is given by, for every $\bp \in \mcl Q_{\infty}(\Delta)$,
\begin{align}\label{e.P^vec_t,q(bp)=}
    \sP^\vec_{t,q}(\bp)= \psi^\vec_{P^\vec_1}(\bq + t\nabla \bxi(\bp))-t\int_0^1\boldsymbol{\theta}(\bp(r))\d r
\end{align}
where $\psi^\vec_{P^\vec_1}$ is given as in~\eqref{e.psi^vec=} and $\boldsymbol\theta$ is given by $\boldsymbol{\theta} = \mathbf{b}\cdot \nabla \bxi(\mathbf{b}) - \bxi(\mathbf{b})$ for every $\mathbf{b}\in \R^{\Delta\times\Delta}$ ($\bxi$ as in~\eqref{e.bxi=mp}).

In the following, we use the above statement to complete the proof. 

Recall the reparametrization of $\{1,\dots,\Delta\}$ given in~\eqref{e.lexico} and recall notation in~\eqref{e.a(m,bullet)(m',bullet)=}.
We set
\begin{align}\label{e.p=(bp)}
    p_\pm = \Big(M^{-1}\sum_{m\in \sfM_s}(\bp_\pm)_{(m,\bullet)(m,\bullet)}\Big)_{s\in\sS},\qquad a_\pm =\Big(M^{-1}\sum_{m\in\sfM_s}(\ba_\pm)_{(m,\bullet)(m,\bullet)}\Big)_{s\in\sS}.
\end{align}
Using this and the relation between $R_{\M N,\lambda_{\M N}}(\sigma,\sigma')$ and $N^{-1}\bsigma\bsigma'^\intercal$ in~\eqref{e.R_N,s=(matrix_overlap)}, we can deduce Part~\eqref{i.l.cavity1.1} from Part~\ref{i.c6.11.1} in the above.

Next, we verify properties of $p_\pm$ and $a_\pm$. First of all, it is clear that $p_\pm \in \mcl Q_\infty^\sS(\kappa)$. By~\eqref{e.p=(bp)} and $\ba_\pm\geq \bp_\pm$, we can deduce $a_\pm\geq p_\pm$. From the definition of $P^\vec_1$ in~\eqref{e.P^vec_1=mp}, the fact that $\ba_\pm \in \overline{\mathrm{conv}}\{\cdots\}$ in the above, and~\eqref{e.p=(bp)}, we have $a_\pm \in \mcl K$ defined in~\eqref{e.mcl_K=}. Lastly, from~\eqref{e.p=(bp)}, we have $|p_{\pm,s}(r)|\leq M^{-1}|\sfM_s||\bp_\pm(r)|\leq \lambda^\star_s$ for every $r\in[0,1)$ and $s\in\sS$. Hence, we can conclude $p_\pm \in \mcl Q^\sS_{\infty,\leq\lambda^\star}(\kappa)$.

Lastly, we verify Part~\eqref{i.l.cavity1.2}.
We want to match the functional appearing in Part~\eqref{i.l.cavity1.2} with that in Part~\ref{i.c6.11.2}.
We start with the second term in the functionals.
By computing $\frac{\d}{\d\eps} \bxi(\mathbf{b}+\eps \mathbf{b})\big|_{\eps=0}$, we can get from the definition of $\bxi$ in~\eqref{e.bxi=mp} that, for any $\mathbf{b},\mathbf{b}'\in\R^{\Delta\times\Delta}$,
\begin{align}\label{e.b'nablabxi(b)=}
    \mathbf{b}'\cdot \nabla \bxi(\mathbf{b}) = \Big(\sum_{m\in \sfM_s}\mathbf{b}'_{(m,\bullet)(m,\bullet)}\Big)_{s\in\sS}\cdot \nabla\xi \bigg(\Big(M^{-1}\sum_{m\in \sfM_s}\mathbf{b}_{(m,\bullet)(m,\bullet)}\Big)_{s\in\sS}\bigg)
\end{align}
Comparing the above with~\eqref{e.p=(bp)}, we get $\bp_\pm\cdot \nabla\bxi(\bp_\pm) = M p_\pm\cdot \nabla\xi(p_\pm)$. By~\eqref{e.p=(bp)} and the relation between $\bxi$ and $\xi$ in~\eqref{e.bxi=mp}, we get $\bxi(\bp_\pm) =M\xi(p_\pm)$. Recall the definition of $\theta$ in~\eqref{e.theta=}. Then, we can conclude $\boldsymbol{\theta}(\bp_\pm) = \M\theta(p_\pm)$. It remains to identify the first term on the right on~\eqref{e.P^vec_t,q(bp)=}.

In the notation in~\eqref{e.lexico} and~\eqref{e.a(m,bullet)(m',bullet)=}, from~\eqref{e.p=(bp)} and~\eqref{e.b'nablabxi(b)=} we can see that
$\nabla \bxi(\bp_\pm)$ and $\nabla \xi(p_\pm)$ satisfy~\eqref{e.bq=mp} with $\bq$ and $q$ therein replaced by them respectively. Hence, we have
\begin{align}\label{e.bq+..and_q+satisfy}
    \text{$\bq + t\nabla\bxi(\bp_\pm)$ and $q+t\nabla\xi(p_\pm)$ satisfy ~\eqref{e.bq=mp}}
\end{align}
with $\bq$ and $q$ therein substituted with this pair.
This along with Lemma~\ref{l.match_mp_with_vector} (at $N=1$ and $t=0$) implies
\begin{align}\label{e.barF_M=M...}
    \bar F_{\M,\lambda_\M}(0,q+t\nabla\xi(p_\pm))= \M^{-1} \bar F^\vec_1(0,\bq + t\nabla\bxi(\bp_\pm)).
\end{align}
In view of~\eqref{e.F^vec(t,q)=} and~\eqref{e.psi^vec=}, we have $\bar F^\vec_1(0,\cdot) = \psi^\vec_{P^\vec_1}$. Therefore, we get
\begin{align*}
    \psi^\vec_{P^\vec_1}(\bq + t\nabla \bxi(\bp)) = \bar F^\vec_1(0,\bq + t\nabla \bxi(\bp))\stackrel{\text{\eqref{e.barF_M=M...}, L.\ref{l.initial_condition}}}{=}\M\psi_{\lambda^\star}(q+t\nabla\xi(p)).
\end{align*}
Inserting this and $\boldsymbol{\theta}(\bp_\pm) = \M\theta(p_\pm)$ into~\eqref{e.P^vec_t,q(bp)=} and comparing it with~\eqref{e.sP_lambda,t,q}, we thus obtain
\begin{align}\label{e.sP^vec_t,q=MsP}
    \sP^\vec_{t,q}(\bp_\pm) = \M \sP_{\lambda^\star,t,q}(p_\pm).
\end{align}
This along with Part~\ref{i.c6.11.2} and Lemma~\ref{l.match_mp_with_vector} yields Part~\eqref{i.l.cavity1.2}. The proof is now complete.
\end{proof}

To state the second result from the cavity computation, we need to introduce the notation for the overlap of cavity spins. From the definition of the overlap in~\eqref{e.R_N,s=} and the free energy in~\eqref{e.F_N(t,q)=}, it is clear that we can reorder the elements in the partition $(I_{N,s})_{s\in\sS}$ as long as the proportions $\lambda_{N,s}$ are preserved. Hence, for every $N\in\N$, we can assume
\begin{align}\label{e.I_M,s_subset_I_MN,s}
    I_{\M,s}\subset I_{\M N, s},\quad\forall s\in\sS
\end{align}
so that the indices for cavity spins are fixed. Then, for every $N\in\N$ and every $\sigma\in \Sigma^{\M N}$, we define
\begin{align}\label{e.sigma_cav,s}
    \sigma_{\circ,s} = (\sigma_{kn})_{1\leq k\leq \kappa_s,\, n\in I_{\M ,s}},\qquad\forall s\in\sS.
\end{align}
to play the role of cavity spins. Then, for every $N\in\N$ and $\sigma,\sigma'\in\Sigma^{\M N}$, we consider the overlap of cavity spins
\begin{gather}\label{e.cavity_overlap}
\begin{split}
    R^{\circ,\M}_{\M N,\lambda_{\M N},s}(\sigma,\sigma')= \frac{1}{\M}\sigma_{\circ,s}\Ll(\sigma_{\circ,s}'\Rr)^\intercal,\quad\forall s\in\sS;
    \\
    R^{\circ,\M}_{\M N,\lambda_{\M N}}(\sigma,\sigma') =\Ll(R^{\circ,\M}_{\M N,\lambda_{\M N},s}(\sigma,\sigma')\Rr)_{s\in\sS}.
\end{split}
\end{gather}
We state a simple observation to be used later.
\begin{lemma}\label{l.last_vector}
Let $\M ,N\in\N$, $(t,q)\in\R_+\times\mcl Q^\sS_\infty(\kappa)$, $\lambda^\star \in \blacktriangle_\M$, and $x\in \pert(\M,\lambda^\star)$. Then at $(t,q)$ we have, for every bounded measurable $\pi:[0,1]\to \prod_{s\in\sS}\S^{\kappa_s}$,
\begin{align}\label{e.l.last_vector}
    \E \la \pi(\alpha\wedge\alpha')\cdot R_{\M N,\lambda_{\M N}}(\sigma,\sigma')\ra^{\M,x}_{\M N,\lambda^\star} = \E \la \pi(\alpha\wedge\alpha')\cdot R^{\circ,\M}_{\M N,\lambda_{\M N}}(\sigma,\sigma')\ra^{\M,x}_{\M N,\lambda^\star}.
\end{align}
\end{lemma}
\begin{proof}
We consider the bijection $\sigma\mapsto \bsigma$ given in~\eqref{e.sigma_to_bsigma_bijection} and assume~\eqref{e.I_M,s_subset_I_MN,s}.
In addition, we can choose the bijection $\sigma\mapsto\bsigma$ in~\eqref{e.sigma_to_bsigma_bijection} to ensure that, for every $N\in\N$, the cavity spins in $(\sigma_{\circ,s})_{s\in\sS}$ (see~\eqref{e.sigma_cav,s}) with $\sigma\in \Sigma^{\M N}$ are mapped to $\btau=(\bsigma_{(m,k)N})_{m,k}\in \R^\Delta$, the last column vector of $\bsigma\in\R^{\Delta\times N}$ (here $\Delta$ is given in~\eqref{e.Delta=}).
Hence, analogous to~\eqref{e.R_N,s=(matrix_overlap)}, we have
\begin{align}\label{e.R^cav_MN=(matrix_overlap)}
    R^{\circ,\M}_{\M N,\lambda_{\M N}}(\sigma,\sigma') = \Big(\M^{-1} \sum_{m\in\sfM_s}\Ll(\btau\btau'^\intercal\Rr)_{(m,\bullet)(m,\bullet)}\Big)_{s\in\sS}
\end{align}
in the notation introduced in~\eqref{e.a(m,bullet)(m',bullet)=}.
In the following, we write $\la\cdot\ra= \la\cdot\ra^{\M,x}_{\M N,\lambda^\star}$. Since the Hamiltonian in~\eqref{e.<>_N,lambda,x} only depends on $\sigma$ only through the overlaps $(\bsigma^l{\bsigma^{l'}}^\intercal)_{l,l'\in\N}$, the law of $\bsigma$ under $\E\la\cdot\ra$ is invariant if we permute the indices of the column vectors. 
Using this and the expression in~\eqref{e.bsigmabisgma'=}, we have
\begin{align*}
    \E \la \pi_s(\alpha\wedge\alpha')\cdot \sum_{m\in\sfM_s}\Ll(\bsigma\bsigma'^\intercal\Rr)_{(m,\bullet)(m,\bullet)}\ra = N  \E \la \pi_s(\alpha\wedge\alpha')\cdot \sum_{m\in\sfM_s}\Ll(\btau\btau'^\intercal\Rr)_{(m,\bullet)(m,\bullet)}\ra
\end{align*}
where we have written $\pi=(\pi_s)_{s\in\sS}$. This together with~\eqref{e.R_N,s=(matrix_overlap)} and~\eqref{e.R^cav_MN=(matrix_overlap)} yields~\eqref{e.l.last_vector}.
\end{proof}

Lastly, for every $\lambda_\M \in \blacktriangle_\M$ and $q\in\mcl Q^\sS_\infty(\kappa)$, we define the Gibbs measure
\begin{align}\label{e.<>_M,lambda,R,q}
    \la \cdot\ra_{\M,\lambda_\M,\fR,q} \propto \exp\Ll(\sqrt{2}W^q_\M(\sigma,\alpha)- \M q(1)\cdot R_{\M,\lambda_\M}(\sigma,\sigma)\Rr) \d P_\M(\sigma)\d \fR(\alpha)
\end{align}
where $W^q_\M(\sigma,\alpha)$ and $P_M$ are given as in~\eqref{e.W^q_N(sigma,alpha)=} and~\eqref{e.P_N=}, respectively.
In view of~\eqref{e.F_N(t,q)=}, we can see that this Gibbs measure is the one associated with $\bar F_{\M,\lambda_M}(0,q)$.

\begin{lemma}\label{l.cavity2}
Let $\M\in\N$ and $\lambda^\star \in \blacktriangle_\M$. Set $\lambda_{\M N}=\lambda^\star$ for every $N$. For any $(t,q)\in\R_+\times \mcl Q^\sS_\infty(\kappa)$ and any sequence $\seq$ of increasing integers, there are a subsequence $(N_k)_{k\in\N}$ of $\seq$, $(x_k)_{k\in\N}$ of parameters in $\pert(\M,\lambda^\star)$, $p\in\mcl Q^\sS_{\infty,\leq \lambda^\star}$, and $a\in\mcl K$ (see~\eqref{e.mcl_K=}) satisfying $a\geq p$ (see~\eqref{e.a>p}) such that
\begin{enumerate}
     \item \label{i.l.cavity2.1} $\big(R_{\M N_k,\lambda^\star}\big(\sigma^l,\sigma^{l'}\big),\, \alpha^l\wedge\alpha^{l'}\big)_{l,l'\in\N}$ under $\E \la \cdot\ra^{\circ,\M,x_k}_{\M N_k,\lambda^\star}$ converges in law to
    \begin{align*}
        \Big(\Ll(p \big(\alpha^l\wedge\alpha^{l'}\big),\ \alpha^l\wedge\alpha^{l'}\Rr)\one_{l\neq l'} + (a,1)\one_{l=l'}\Big)_{l,l'\in\N}
    \end{align*}
    under $\E\la\cdot\ra_\fR$, as $k$ tends to infinity;
    \item  \label{i.l.cavity2.2}
    for every bounded continuous $g:\Ll(\prod_{s\in\sS}\R^{\kappa_s\times\kappa_s}\Rr)\times \R\to\R$, we have
    \begin{align*}
        \lim_{k\to\infty}\E \la g\Ll(R^{\circ,\M}_{\M (N_k+1),\lambda^\star}(\sigma,\sigma'),\ \alpha\wedge\alpha'\Rr) \ra^{\M,x_k}_{\M(N_k+1),\lambda^\star}
        \\
        =\E \la g\Ll(R_{\M,\lambda^\star}(\sigma,\sigma'),\ \alpha\wedge\alpha'\Rr) \ra_{\M,\lambda^\star,\fR,q+t\nabla\xi(p)}.
    \end{align*}
\end{enumerate}
\end{lemma}

The two Gibbs measures in Part~\eqref{i.l.cavity2.1} are given in~\eqref{e.<>^circ=} and Section~\ref{s.cascade}, respectively. The two Gibbs measure in Part~\eqref{i.l.cavity2.2} are given in~\eqref{e.<>_N,lambda,x} and~\eqref{e.<>_M,lambda,R,q}, respectively.

\begin{proof}
We match $\bar F_{\M N,\lambda_{\M N}}$, its perturbed versions, and the associated Gibbs measures with those of $\bar F^\vec_N$ as described in the first two paragraphs in the proof of Lemma~\ref{l.cavity1}. 

Then, applying~\cite[Corollary~6.12]{HJ_critical_pts} (with $M$ therein set to be $1$), we get that there are 
$(N_k)_{k\in\N}$, $(x_k)_{k\in\N}$, $\bp \in \mcl Q_{\infty,\leq 1}(\Delta)$, and $\ba \in \overline{\mathrm{conv}}\Ll\{\bsigma\bsigma^\intercal:\: \bsigma\in \supp P^\vec_1\Rr\}$ (see~\eqref{e.P^vec_1=mp}) satisfying $\ba\geq \bp$ such that
\begin{enumerate}[label=(\roman*)]
    \item \label{i.c6.12.1} $\big(N^{-1}\bsigma^l\big(\bsigma^{l'}\big)^\intercal,\, \alpha^l\wedge\alpha^{l'}\big)_{l,l'\in\N}$ under $\E \la\cdot\ra^{\circ,\M,x_k}_{\M N_k,\lambda^\star}$ converges in law to
    \begin{align*}
        \Big(\Ll(\bp \big(\alpha^l\wedge\alpha^{l'}\big), \alpha^l\wedge\alpha^{l'}\Rr)\one_{l\neq l'} + (\ba,1)\one_{l=l'}\Big)_{l,l'\in\N}
    \end{align*}
    under $\E\la\cdot\ra_\fR$ as $k$ tends to infinity;
    \item \label{i.c6.12.2}
    for every bounded continuous $\mathbf{g}:\R^{\Delta\times\Delta}\times \R\to\R$, we have 
    \begin{align*}
        \lim_{k\to\infty}\E \la \mathbf{g}\Ll(\btau\btau'^\intercal,\ \alpha\wedge\alpha'\Rr) \ra^{\M,x_k}_{\M(N+1),\lambda^\star}
        =\E \la \mathbf{g}\Ll(\btau\btau'^\intercal,\ \alpha\wedge\alpha'\Rr) \ra_{\fR,\bq+t\nabla\bxi(\bp)}^\vec.
    \end{align*}
\end{enumerate}
Here, $\btau=(\bsigma_{(m,k)N+1})_{m,k}\in \R^\Delta$ is the last column vector of $\bsigma\in\R^{\Delta\times (N+1)}$ and the Gibbs measure on the right-hand side is defined as in~\cite[(6.17)]{HJ_critical_pts}:
\begin{align*}
    \la\cdot\ra_{\fR,\bq+t\nabla\bxi(\bp)}^\vec\propto \exp\Ll(\sqrt{2}W^{\bq+t\nabla\bxi(\bp)}_1(\alpha)\cdot \btau- \Ll(\bq+t\nabla\bxi(\bp)\Rr)(1)\cdot\btau\btau^\intercal\Rr) \d P_1^\vec(\btau)\d \fR(\alpha)
\end{align*}
where $W^{\bq+t\nabla\bxi(\bp)}_1$ is given as in~\eqref{e.W^q_N(alpha)=}.

We set $p$ and $a$ analogously as in~\eqref{e.p=(bp)}. Using the same argument below~\eqref{e.p=(bp)}, we can get Part~\eqref{i.l.cavity2.1} from Part~\ref{i.c6.12.1}.

For the second part, as argued above~\eqref{e.R^cav_MN=(matrix_overlap)}, we can choose the bijection $\sigma\mapsto \bsigma$ in~\eqref{e.sigma_to_bsigma_bijection} to satisfy
\begin{align}\label{e.R^cav_M(N+1)=(matrix_overlap)}
    R^{\circ,\M}_{\M(N+1),\lambda^\star}(\sigma,\sigma') = \Big(\M^{-1} \sum_{m\in\sfM_s}\Ll(\btau\btau'^\intercal\Rr)_{(m,\bullet)(m,\bullet)}\Big)_{s\in\sS}.
\end{align}
Similar to~\eqref{e.bq+..and_q+satisfy}, we have that $\bq+t\nabla\bxi(\bp)$ and $q+t\nabla\xi(p)$ satisfy~\eqref{e.bq=mp}, which allows us to use the identity~\eqref{e.external_field_equal_distri} to see that, under the identification of $\bsigma$ and $\sigma$ in~\eqref{e.sigma_to_bsigma_bijection}, $(\sigma^l,\alpha^l)_{l\in\N}$ under $\E\la\cdot\ra_{\fR,\bq+t\nabla\bxi(\bp)}^\vec$ has the same law under $\E \la\cdot \ra_{\M,\lambda^\star,\fR,q+t\nabla\xi(p)}$. Therefore, Part~\eqref{i.l.cavity2.2} follows from Part~\ref{i.c6.12.2} and~\eqref{e.R^cav_M(N+1)=(matrix_overlap)}. Properties of $p$ and $a$ can be verified similarly as in the proof of Lemma~\ref{l.cavity1}.
\end{proof}

\begin{remark}[Perturbation specific to multi-species models]\label{r.perturbation}
In the above, the perturbation was introduced in the style of vector spins (presence of $\bsigma$ in~\eqref{e.E[H^hH^h]=} and~\eqref{e.H^x_N(sigma,alpha)}) because we want to directly apply Lemma~\ref{l.match_mp_with_vector} and results from~\cite{HJ_critical_pts} for vector spins. In fact, results in Lemmas~\ref{l.cavity1} and~\ref{l.cavity2} hold for the following perturbation specific to the multi-species models as used in~\cite{pan.multi,bates2022free}. Instead of only considering systems with size $\M N$ due to reliance on the assumption~\eqref{e.lambda^MN_s=|M_s|/M}, we are able to define the perturbation for each $N\in\N$ (but later we only apply to systems with sizes $\M N$). Let $(r_n)$ be an enumeration of $[0,1]\cap \Q$ and $(a_n)_{n\in\N}$ be an enumeration of elements in $\prod_{s\in\sS}\S^{\kappa_s}_+$ with rational entries. Conditioned on $\fR$, for every $h=(h_i)_{1\leq i\leq 4}\in \N^4$, let $\Ll(H^h_N(\sigma,\alpha)\Rr)_{\sigma\in\Sigma^N,\, \alpha\in\supp\fR}$ be an independent centered Gaussian process with covariance
\begin{align*}
    \E \Ll[H^h_N(\sigma,\alpha)H^h_N(\sigma',\alpha')\Rr] = N \Ll(a_{h_1}\cdot \Ll(R_{N,\lambda_N}(\sigma,\sigma')\Rr)^{\odot h_2}+\lambda_{h_3} \alpha\wedge\alpha'\Rr)^{h_4}.
\end{align*}
The existence of such a process is explained in~\cite[Section~6.1.1]{HJ_critical_pts}. 
Fix $c_h$ similarly as in~\eqref{e.c_h<} but this time for every $N$ instead of $\M N$.
Let $\mathsf{bas}$ be an orthonormal basis of $\prod_{s\in\sS}\S^{\kappa_s}$.
Then, we set
\begin{align*}
    H^x_N(\sigma,\alpha) = \sum_{h\in\N^4}x_hc_hH^h_N(\sigma,\alpha) + \sum_{e\in \mathsf{bas}}x_e e\cdot R_{N,\lambda_N}(\sigma,\sigma).
\end{align*}
Let $\M$ be the dimension of cavity as fixed in Lemmas~\ref{l.cavity1} and~\ref{l.cavity2}. Let $\tilde H^{t,q}_N(\sigma,\alpha)$ be given as in~\eqref{e.tildeH^t,q_N(sigma,lapha)} relative to this $\M$.
Then, for each $N\in\N$, we define
\begin{gather*}
    \la \cdot\ra_{N,\lambda_N,x}\propto \exp\Ll(H^{t,q}_N(\sigma,\alpha) + N^{-\frac{1}{16}}H^x_N(\sigma,\alpha)\Rr)\d P_{N,\lambda_N}(\sigma) \d \fR(\alpha),
    \\
    \la \cdot\ra_{N,\lambda_N,x}^\circ\propto \exp\Ll(\tilde H^{t,q}_N(\sigma,\alpha) + N^{-\frac{1}{16}}H^x_N(\sigma,\alpha)\Rr)\d P_{N,\lambda_N}(\sigma) \d \fR(\alpha). 
\end{gather*}
Then, Lemma~\ref{l.cavity1} and~\ref{l.cavity2} hold when the Gibbs measures therein are replaced by these two (at size $\M N$).
To see this, one needs to redo cavity computations in~\cite[Section~6]{HJ_critical_pts} using this perturbation. The key part is to prove that Ghirlanda--Guerra identities hold for the overlap $\Ll(R_{N,\lambda_N}(\sigma^l,\sigma^{L'})\Rr)_{l,l'\in\N}$ under this perturbation (see~\cite[Proposition~6.8]{HJ_critical_pts}). These modifications are straightforward but tedious.
\qed
\end{remark}

\subsection{Approximation}\label{s.approx}

Recall the definition of the discrete simplex $\blacktriangle_N$ from~\eqref{e.simplex}. Recall the continuous simplex $\blacktriangle_\infty$ from~\eqref{e.cts_simplex}.
In this section, we fix any $\lambda_\infty\in \blacktriangle_\infty$. We want to use approximation arguments to study the limit of $\bar F_{N,\lambda_N}$ and the limit law of overlaps, for $(\lambda_N)_{N\in\N}$ converging to $\lambda_\infty$. 

To use previous results, it is convenient to work with another sequence of proportions other than $(\lambda_N)_{N\in\N}$. We fix two sequences $(M_n)_{n\in\N}$ and $(\lambda^\star_n)_{n\in\N}$ satisfying
\begin{align}\label{e.lambda_star,n}
    \lambda^\star_n\in \blacktriangle_{M_n},\quad\forall n\in\N;\qquad \lim_{n\to\infty} \lambda^\star_n = \lambda_\infty.
\end{align}

To simplify our notation for approximations,
for $a,b\in\R$ and $\eps>0$, we write $a\lesssim_\eps b$ provided $a\leq b+\eps$;
also, we write $a\approx_\eps b$ provided $|a-b|\leq \eps$.

Recall the functional $\sP_{\lambda,t,q}$ in~\eqref{e.sP_lambda,t,q} and Gibbs measures introduced in Section~\ref{s.cascade}, \eqref{e.<>_N,lambda,x}, \eqref{e.<>^circ=}, and~\eqref{e.<>_M,lambda,R,q}. 
Also recall the space of perturbation parameters in~\eqref{e.pert=}.
We have two results and we now state the first.

\begin{lemma}\label{l.approx1}
Assume that $(\lambda_N)_{N\in\N}$ converges to some $\lambda_\infty$.
Let $(t,q)\in\R_+\times \mcl Q^\sS_\infty(\kappa)$.
Let $(M_n)_{n\in\N}$ and $(\lambda^\star_n)_{n\in\N}$ satisfy~\eqref{e.lambda_star,n}.
Then, there are
\begin{itemize}
    \item $(N^\pm_n)_{n\in\N}$ of strictly increasing positive integers satisfying $N^\pm_n\in M_n\N$ for each $n$,
    \item $(x^\pm_n)_{n\in\N}$ of perturbation parameters satisfying $x^\pm_n\in\pert(\M_n,\lambda^\star_n)$ for each $n$, 
    \item $p_\pm$ in $\mcl Q^\sS_{\infty,\leq \lambda_\infty}(\kappa)$ and $a_\pm$ in $ \prod_{s\in\sS}\S^{\kappa_s}_+$ satisfying $a_\pm\geq p_\pm$ (see~\eqref{e.a>p}),
\end{itemize}
such that the following holds at $(t,q)$:
\begin{enumerate}
    \item \label{i.l.approx1.1}
    $\big(R_{N^\pm_n,\lambda^\star_n}(\sigma^l,\sigma^{l'}),\ \alpha^l\wedge\alpha^{l'}\big)_{l,l'\in\N}$ under $\E \la \cdot\ra^{\circ,\M_n,x^\pm_n}_{N^\pm_n,\lambda^\star_n}$ converges in law to
    \begin{align*}
        \Big(\Ll(p_\pm \big(\alpha^l\wedge\alpha^{l'}\big),\ \alpha^l\wedge\alpha^{l'}\Rr)\one_{l\neq l'} + (a_\pm,1)\one_{l=l'}\Big)_{l,l'\in\N}
    \end{align*}
    under $\E\la\cdot\ra_\fR$, as $n$ tends to infinity;
    \item\label{i.l.approx1.2} we have 
    \begin{align*}
        \limsup_{n\to\infty}\bar F_{N^\pm_n,\lambda^\star_n}(t,q)\leq \sP_{\lambda_\infty,t,q}(p_+),\qquad \liminf_{n\to\infty}\bar F_{N^\pm_n,\lambda^\star_n}(t,q)\geq \sP_{\lambda_\infty,t,q}(p_-).
    \end{align*}
\end{enumerate}
Moreover, if $(\bar F_{N,\lambda_N})_{N\in\N}$ converges pointwise to some $f$, then both $\Ll(\bar F_{N^\pm_n,\lambda^\star_n}\Rr)_{n\in\N}$ and $\Ll(\tilde F^{\M_n,x^\pm_n}_{N^\pm_n,\lambda^\star_n}\Rr)_{n\in\N}$ converge pointwise to $f$.
\end{lemma}

\begin{proof}

We shall only consider the case with superscript $+$ and the other case can be treated similarly. Henceforth, we omit $+$ from the notation.
Fix any sequence $(\eps_n)_{n\in\N}$ of strictly positive real numbers satisfying $\lim_{n\to\infty}\eps_n =0$.

For each $n\in\N$,
let $(N^n_k)_{k\in\N}$, $(x^n_k)_{k\in\N}$, $p_n$, and $a_n$ be given by Lemma~\ref{l.cavity1} for the $+$ case with $\M_n$ and $\lambda^\star_n$ substituted for $\M$ and $\lambda^\star$ therein. By Part~\eqref{i.l.cavity1.1} of Lemma~\ref{l.cavity1} (and also Lemma~\ref{l.cvg_overlap}), there is $\mathbf{k}_1\in\N$ such that, for every $k\geq \mathbf{k}_1$, we have
\begin{gather}
    \E \la \Ll|R_{\M_n N^n_k,\lambda^\star_n}(\sigma,\sigma')- p_n(\alpha\wedge\alpha')\Rr|\ra^{\circ,\M_n,x^n_k}_{\M_n N^n_k,\lambda^\star_n}\approx_{\eps_n} 0,\label{e.l.approx1.2_pf}
    \\
    \E \la \Ll|R_{\M_n N^n_k,\lambda^\star_n}(\sigma,\sigma)-a_n\Rr|\ra^{\circ,\M_n,x^n_k}_{\M_n N^n_k,\lambda^\star_n}\approx_{\eps_n} 0.\label{e.l.approx1.3_pf}
\end{gather}
By Part~\eqref{i.l.cavity1.2} of Lemma~\ref{l.cavity1}, there is $\mathbf{k}_2\in\N$ such that, for every $k\geq \mathbf{k}_2$, we have
\begin{align}
    \bar F_{\M_n N^n_k,\lambda^\star_n}(t,q) \lesssim_{\eps_n} \limsup_{N\to\infty}\bar F_{M_n N,\lambda^\star_n}(t,q) \leq \sP_{\lambda^\star_n,t,q}(p_n).\label{e.l.approx1.4_pf}
\end{align}
Fix any $k_n$ satisfying $k\geq \max_{i\in\{1,2\}}\mathbf{k}_i$ and set $N_n = M_n N^n_{k_n}$ and $x_n = x^n_{k_n}$. Since there is no upper bound on $k_n$, we can choose larger $k_n$ to ensure that $(N_n)_{n\in\N}$ is strictly increasing as announced and
\begin{align}\label{e.(N/M)^-1<eps}
        \Ll(N_n / M_n\Rr)^{-1} \leq \eps_n
\end{align}
which will needed later.

By passing to a subsequence of $(N_n)_{n\in\N}$ and using the compactness result in Lemma~\ref{l.compact_embed_paths}, we may assume that there are $p\in\mcl Q^\sS_\infty(\kappa)$ such that $(p_n)_{n\in\N}$ converges to $p$ pointwise and in $L^1$. 
Also, we can assume that $(a_n)_{n\in\N}$ converges to some $a$. Recall from Lemma~\ref{l.cavity1} that we have $p_n \in \mcl Q^\sS_{\infty,\leq \lambda^\star_n}(\kappa)$. The pointwise convergence and~\eqref{e.lambda_star,n} implies $p  \in \mcl Q^\sS_{\infty,\leq \lambda_\infty}(\kappa)$. Clearly, we also have $a\geq p$. Hence, we have verified properties of $p$ and $a$.

We are ready to verify the two main parts of Lemma~\ref{l.approx1}. From the convergence of $(a_n)_{n\in\N}$ and~\eqref{e.l.approx1.3_pf}, we immediately get
\begin{align}
    \lim_{n\to\infty}\E \la \Ll|R_{N_n,\lambda_n}(\sigma,\sigma)-a\Rr|\ra^{\circ,\M_n,x_n}_{N_n,\lambda^\star_n}= 0.\label{e.l.approx1.3}
\end{align}
Using the invariance of cascades in Lemma~\ref{l.invar} and the convergence of $(p_n)_{n\in\N}$, we have
\begin{align*}
    \lim_{n\to\infty} \E \la \Ll| p_n(\alpha\wedge\alpha')- p(\alpha\wedge\alpha')\Rr|\ra^{\circ,\M_n,x_n}_{N_n,\lambda^\star_n} =0
\end{align*}
which together with~\eqref{e.l.approx1.2_pf} implies
\begin{align}
    \lim_{n\to\infty}\E \la \Ll|R_{N_n,\lambda_n}(\sigma,\sigma')- p(\alpha\wedge\alpha')\Rr|\ra^{\circ,\M_n,x_n}_{N_n,\lambda^\star_n}=0,\label{e.l.approx1.2}
\end{align}
Combining~\eqref{e.l.approx1.3} and~\eqref{e.l.approx1.2} with Lemma~\ref{l.cvg_overlap}, we can conclude Part~\eqref{i.l.approx1.1}.
From its definition in~\eqref{e.sP_lambda,t,q}, we can see that $\sP_{\lambda^\star_n,t,q}(p_n)$ is continuous jointly in $\lambda^\star_n$ and $p_n$ (due to Lemmas~\ref{l.psi^vec_smooth} and~\ref{l.reg_initial}). This along with~\eqref{e.l.approx1.4_pf} yields the first side in Part~\eqref{i.l.approx1.2}. As we explained previously, the other side can be treated by the same method.

Now, we turn to the last statement. Here, we display the superscript $\pm$.
For each $(t',q')$, let $C_{t',q'}$ be the constant appearing on the right-hand side in~\eqref{e.|F-F|<C|lambda-lambda|} of Lemma~\ref{l.lambda_continuity}, which only depends on $(\kappa_s)_{s\in\sS}$, $\xi$, $(\mu_s)_{s\in\sS}$, and $(t',q')$. We set $\eps^\pm_n = C_{t',q'}|\lambda^\star_n -\lambda_{N^\pm_n}|$. Since both $(\lambda^n_\star)_{n\to\infty}$ and $(\lambda_N)_{N\to\infty}$ converge to $\lambda_\infty$, we have $\lim_{n\to\infty}\eps^\pm_n=0$. Lemma~\ref{l.lambda_continuity} implies that, for every $(t',q')$,
\begin{align}\label{e.F_approx_eps^pm_n_F}
    \bar F_{N^\pm_n,\lambda^\star_n}(t',q') \approx_{\eps^\pm_n} \bar F_{N^\pm_n,\lambda_{N^\pm_n}}(t',q')
\end{align}
which gives the convergence of $\bar F_{N^\pm_n,\lambda^\star_n}$ to $f$.
By Lemma~\ref{l.F-F^x}, we have
\begin{align*}
    \Ll|\bar F_{N^\pm_n,\lambda^\star_n}(t',q')-\tilde F^{\M_n,x^\pm_n}_{N^\pm_n,\lambda^\star_n}(t',q')\Rr| \leq C(1+|t'|)(N^\pm_n / M_n)^{-\frac{1}{16}}.
\end{align*}
This along with~\eqref{e.(N/M)^-1<eps} and~\eqref{e.F_approx_eps^pm_n_F} implies that $\tilde F^{\M_n,x^\pm_n}_{N^\pm_n,\lambda^\star_n}(t',q')$ converges pointwise to $f$ as $n\to\infty$. 
\end{proof}

Recall the cavity overlap from~\eqref{e.cavity_overlap}. We state the second result.

\begin{lemma}\label{l.approx2}
Assume that $(\lambda_N)_{N\in\N}$ converges to some $\lambda_\infty$.
Let $(t,q)\in\R_+\times \mcl Q^\sS_\infty(\kappa)$.
Let $(M_n)_{n\in\N}$ and $(\lambda^\star_n)_{n\in\N}$ satisfy~\eqref{e.lambda_star,n}.
Let $\seq$ be any strictly increasing sequence in $\N$.
Then, there are
\begin{itemize}
    \item a subsequence $(N_n)_{n\in\N}$ of $\seq$ satisfying $N_n\in M_n\N$ for every $n$,
    \item $(x_n)_{n\in\N}$ of perturbation parameters satisfying $x_n\in\pert(\M_n,\lambda^\star_n)$ for each $n$, 
    \item $p$ in $\mcl Q^\sS_{\infty,\leq \lambda_\infty}(\kappa)$ and $a$ in $ \prod_{s\in\sS}\S^{\kappa_s}_+$ satisfying $a\geq p$,
\end{itemize}
such that the following holds at $(t,q)$:
\begin{enumerate}
    \item \label{i.l.approx2.1}
    $\big(R_{N_n,\lambda^\star_n}(\sigma^l,\sigma^{l'}),\ \alpha^l\wedge\alpha^{l'}\big)_{l,l'\in\N}$ under $\E \la \cdot\ra^{\circ,\M_n,x_n}_{N_n,\lambda^\star_n}$ converges in law to
    \begin{align*}
        \Big(\Ll(p_\pm \big(\alpha^l\wedge\alpha^{l'}\big),\ \alpha^l\wedge\alpha^{l'}\Rr)\one_{l\neq l'} + (a_\pm,1)\one_{l=l'}\Big)_{l,l'\in\N}
    \end{align*}
    under $\E\la\cdot\ra_\fR$, as $n$ tends to infinity;
    \item\label{i.l.approx2.2} 
    for any bounded continuous $\pi:[0,1]\to\Ll(\prod_{s\in\sS}\S^{\kappa_s}\Rr)$, we have
    \begin{align*}
        \lim_{n\to\infty} \E \la  \pi\Ll(\alpha\wedge\alpha'\Rr)\cdot R^{\circ,\M_n}_{N_n+M_n,\lambda^\star_n}\ra^{\M_n,x_n}_{N_n+M_n,\lambda^\star_n}  = \la \pi, \partial_q \psi_{\lambda_\infty} (q+t\nabla\xi(p))\ra_{L^2}.
    \end{align*}
\end{enumerate}
Moreover, if $(\bar F_{N,\lambda_N})_{N\in\seq}$ converges pointwise to some $f$, then both $\Ll(\bar F^{\M_n,x_n}_{N_n+\M_n,\lambda^\star_n}\Rr)_{n\in\N}$ and $\Ll(\tilde F^{\M_n,x_n}_{N_n,\lambda^\star_n}\Rr)_{n\in\N}$ converge pointwise to $f$.
\end{lemma}

\begin{proof}
Fix any sequence $(\eps_n)_{n\in\N}$ of strictly positive real numbers satisfying $\lim_{n\to\infty}\eps_n =0$.

For each $n$, we choose the corresponding parameters. Define $\seq(n) = \Ll(\lfloor N/\M_n\rfloor\Rr)_{N\in\seq}$.
Apply Lemma~\ref{l.cavity2} to $M_n$, $\lambda^\star_n$, and $\seq(n)$, we get the corresponding $(N^n_k)_{k\in\N}$, $(x^n_k)_{k\in\N}$, $p_n\in \mcl Q^\sS_{\infty,\leq \lambda^\star_n}(\kappa)$, and $a_n$.
By Part~\eqref{i.l.cavity2.1} of Lemma~\ref{l.cavity2} (together with Lemma~\ref{l.cvg_overlap}), there is $\mathbf{k}_1\in\N$ such that, for every $k\geq \mathbf{k}_1$, we have
\begin{gather}
    \E \la \Ll|R_{\M_n N^n_k,\lambda^\star_n}(\sigma,\sigma')- p_n(\alpha\wedge\alpha')\Rr|\ra^{\circ,\M_n,x^n_k}_{\M_n N^n_k,\lambda^\star_n}\approx_{\eps_n} 0,\label{e.l.approx2.2_pf}
    \\
    \E \la \Ll|R_{\M_n N^n_k,\lambda^\star_n}(\sigma,\sigma)-a_n\Rr|\ra^{\circ,\M_n,x^n_k}_{\M_n N^n_k,\lambda^\star_n}\approx_{\eps_n} 0.\label{e.l.approx2.3_pf}
\end{gather}
By Part~\eqref{i.l.cavity2.2} of Lemma~\ref{l.cavity2}, there is $\mathbf{k}_2\in\N$ such that, for every $k\geq \mathbf{k}_2$, we have
\begin{align}\label{e.l.approx2.4_pf}
    \begin{split}
        \E \la  \pi\Ll(\alpha\wedge\alpha'\Rr)\cdot R^{\circ,\M_n}_{\M_n(N^n_k+1),\lambda^\star_n}(\sigma,\sigma')\ra^{\M_n,x^n_k}_{\M_n(N^n_k+1),\lambda^\star_n} 
    \\
    \approx_{\eps_n} \E \la \pi\Ll( \alpha\wedge\alpha' \Rr)\cdot R_{\M_n,\lambda^\star_n}(\sigma,\sigma') \ra_{\M_n,\lambda^\star_n,\fR,q+t\nabla\xi(p_n)}
    \end{split}
\end{align}
These are preparations for Parts~\eqref{i.l.approx2.1} and~\eqref{i.l.approx2.2} of the lemma to prove. But before proceeding, we first prove the last statement in the lemma.
For each $k\in\N$, let $\tilde N^n_k$ from $\seq$ satisfy $\lfloor  \tilde N^n_k /M_n \rfloor = N^n_k$. Then, we can estimate, at every $(t',q')\in\R_+\times \mcl Q^\sS_\infty(\kappa)$,
\begin{align*}
    \Ll|\bar F_{\tilde N^n_k,\lambda_{\tilde N^n_k}} - \bar F^{\M_n,x^n_k}_{\M_n ( N^n_k+1),\lambda^\star_n}\Rr| 
    \leq \Ll|\bar F_{\tilde N^n_k,\lambda_{\tilde N^n_k}}- \bar F_{\tilde N^n_k,\lambda^\star_n}\Rr|+\Ll|\bar F_{\tilde N^n_k,\lambda^\star_n} - \bar F_{\M_n ( N^n_k+1),\lambda^\star_n}\Rr|
    \\
    + \Ll|\bar F_{\M_n ( N^n_k+1),\lambda^\star_n} -  \bar F^{\M_n,x^n_k}_{\M_n ( N^n_k+1),\lambda^\star_n}\Rr|
\end{align*}
Applying Lemmas~\ref{l.lambda_continuity}, \ref{l.NF_N-(N+M)F_(N+M)}, and~\ref{l.F-F^x} to the three terms respectively, we can see that the left-hand side is bounded by a constant $C_{t',q'}$ times
\begin{align*}
    \Ll|\lambda_{\tilde N^n_k}-\lambda^\star_n\Rr|+ \frac{M_n}{\tilde N^n_k} + \Ll(N^n_k\Rr)^{-1/16}.
\end{align*}
A similar bound holds for the difference between $\bar F_{\tilde N^n_k,\lambda_{\tilde N^n_k}}$ and $\tilde F^{\M_n,x^n_k}_{N_n,\lambda^\star_n}$. Notice that the last two terms in this bound vanish as $k$ tends to infinity (recall that $n$ is temporarily fixed) while the first term tends to $\Ll|\lambda_\infty -\lambda^\star_n\Rr|$. 
Recall that $(\eps_n)_{n\in\N}$ is an arbitrary vanishing sequence that we choose. Here, for convenience of notation, we can assume that we have chosen it to satisfy $\eps_n \geq 2|\lambda_\infty -\lambda^\star_n|$.
Then, we can find $\mathbf{k}_3\in\N$ (independent of $(t',q')$) such that the error term in the above display is bounded by $\eps_n$ for all $k\geq \mathbf{k}_3$. Hence, at every $(t',q')$ and for every $k\geq \mathbf{k}_3$, we have
\begin{align}\label{e.approx2.3}
    \bar F_{\tilde N^n_k,\lambda_{\tilde N^n_k}} \approx_{C_{t',q'}\eps_n}\bar F^{\M_n,x^n_k}_{\M_n ( N^n_k+1),\lambda^\star_n},\qquad \bar F_{\tilde N^n_k,\lambda_{\tilde N^n_k}}\approx_{C_{t',q'}\eps_n} \tilde F^{\M_n,x^n_k}_{N_n,\lambda^\star_n}.
\end{align}

Fix some $k_n$ satisfying $k\geq \max_{i\in\{1,2,3\}}\mathbf{k}_i$ and set 
\begin{align}\label{e.chooseN_nx_n}
    N_n = M_n N^n_{k_n},\qquad x_n = x^n_{k_n}.
\end{align}
This is done for each $n$.
We can choose larger $k_n$ to ensure that both $(N_n)_{n\in\N}$ and $(\tilde N^n_{k_n})_{n\in\N}$ are strictly increasing. To verify the last statement, recall that our choice of $\tilde N^n_k$ ensures that $\bar F_{\tilde N^n_k,\lambda_{\tilde N^n_k}}$ in~\eqref{e.approx2.3} belongs to the sequence $(\bar F_{N,\lambda_N})_{N\in\seq}$. Now inserting~\eqref{e.chooseN_nx_n} into~\eqref{e.approx2.3}, we can see that the last statement holds.

Now we turn to Parts~\eqref{i.l.approx2.1} and~\eqref{i.l.approx2.2}.
By passing to a subsequence of $(N_n)_{n\in\N}$ and using the compactness result in Lemma~\ref{l.compact_embed_paths}, we can find some $p\in\mcl Q^\sS_\infty(\kappa)$ such that $(p_n)_{n\in\N}$ converges to $p$ pointwise and in $L^1$. By the same argument as in the proof of Lemma~\ref{l.approx1}, we can verify that $p\in \mcl Q^\sS_{\infty,\leq\lambda_\infty}(\kappa)$ and $a\geq p$.

Notice that~\eqref{e.l.approx2.2_pf} and~\eqref{e.l.approx2.3_pf} have the same form as in~\eqref{e.l.approx1.2_pf} and~\eqref{e.l.approx1.3_pf}
Using the same argument in Lemma~\ref{l.approx1}, we can verify Part~\eqref{i.l.approx2.1} of Lemma~\ref{l.approx2}. 

For Part~\eqref{i.l.approx2.2}, recall that the Gibbs measure on the right of~\eqref{e.l.approx2.4_pf} is given in~\eqref{e.<>_M,lambda,R,q}. 
Comparing it with~\eqref{e.F_N(t,q)=} and~\eqref{e.<>_N=}, we can see that it is exactly the Gibbs measure associated with $\bar F_{\M_n,\lambda^\star_n}(0,q+t\nabla\xi(p_n))$.
Using this and the computation in~\eqref{e.def.der.FN}, we can see that the right-hand side in~\eqref{e.l.approx2.4_pf} is equal to
\begin{align*}
    \la \pi, \partial_q \bar F_{\M_n,\lambda^\star_n}(0,q+t\nabla\xi(p_n))\ra_{L^2}\stackrel{\text{L.\ref{l.initial_condition}}}{=}\la \pi, \partial_q \psi_{\lambda^\star_n}(q+t\nabla\xi(p_n))\ra_{L^2}.
\end{align*}
By Lemmas~\ref{l.psi^vec_smooth} and~\ref{l.reg_initial}, $\partial_q \psi_{\lambda^\star_n}(q+t\nabla\xi(p_n))$ is continuous in $\lambda^\star_n$ and $p_n$.
Hence, the above display along with~\eqref{e.l.approx2.4_pf} gives Part~\eqref{i.l.approx2.2} after $n$ is sent to infinity.
\end{proof}

\subsection{Results concerning critical points}\label{s.crti_pt_and_rel_results}

We use Lemma~\ref{l.approx1} to prove the following proposition, which is the counterpart to~\cite[Proposition~7.1]{HJ_critical_pts}.
Recall the definition of $\mcl Q_{\infty,\uparrow}(\D)$ from~\eqref{e.Q_uparrow(D)=} and that of $\mcl Q^\sS_{\infty,\uparrow}(\kappa)$ as in~\eqref{e.Q^S(kappa)=}.

\begin{proposition}[Critical point identification, I]\label{p.crit_pt_1}
Suppose that $(\lambda_N)_{N\in\N}$ converges to some $\lambda_\infty \in \blacktriangle_\infty$ and that $\Ll(\bar F_{N,\lambda_N}\Rr)_{N\in\N}$ converges pointwise to some $f$. Let $t\in\R_+$ and $q\in \mcl Q^\sS_{\infty,\uparrow}(\kappa)$. If $f(t,\cdot)$ is Gateaux differentiable at $q$, then $f(t,q) = \sP_{\lambda_\infty,t,q}\Ll(\partial_q f(t,q)\Rr)$.
\end{proposition}

\begin{proof}
Let the sequences and parameters be given by Lemma~\ref{l.approx1}. 
Since $\bar F_{N^\pm_n,\lambda^\star_n}$ also converges to $f$ due to Lemma~\ref{l.approx1}, we get from Part~\eqref{i.l.approx1.2} of this lemma that
\begin{align}\label{e.P<f<P}
     \sP_{\lambda_\infty,t,q}(p_-)\leq f(t,q)\leq\sP_{\lambda_\infty,t,q}(p_+).
\end{align}
It remains to identify $p_\pm$.
For brevity, we write $\tilde F^\pm_n = \tilde F^{\M_n,x^\pm_n}_{N^\pm_n,\lambda^\star_n}$ and $\la\cdot\ra^\circ_{\pm,n}=\la \cdot\ra^{\circ,\M_n,x^\pm_n}_{N^\pm_n,\lambda^\star_n}$.
Lemma~\ref{l.approx1} gives that $\bar F^\pm_n$ converges to $f$.
As a consequence of Proposition~\ref{p.cvg_der_semi_concave} and Remark~\ref{r.cvg_der_semi_concave}, we have the convergence of $\partial_q \tilde F^\pm_n(t,q)$ to $\partial_q f(t,q)$.
By a similar computation for~\eqref{e.def.der.FN}, we have that, for every continuous $\pi \in L^\infty([0,1]; \prod_{s\in\sS}S^{\kappa_s})$,
\begin{align*}
    &\la \pi,\,\partial_q \tilde F^\pm_n(t,q)\ra_{L^2} = \E \la \pi(\alpha\wedge\alpha')\cdot R_{N^\pm_n,\lambda^\star_n}(\sigma,\sigma')\ra^\circ_{\pm,n}
\end{align*}
We send $n$ to infinity and apply Lemma~\ref{l.approx1}~\eqref{i.l.approx1.1} to the right-hand side to get
\begin{align*}
    \la \pi,\,\partial_q f(t,q)\ra_{L^2} = \E \la \pi(\alpha\wedge\alpha')\cdot p_\pm(\alpha\wedge\alpha')\ra_\fR \stackrel{\text{L.\ref{l.invar}}}{=}\la \pi , p_\pm\ra_{L^2}
\end{align*}
which implies that $p_\pm = \partial_q f(t,q)$. Inserting this to~\eqref{e.P<f<P}, we get the desired result.
\end{proof}

The following corresponds to~\cite[Proposition~7.2]{HJ_critical_pts} and we apply Lemma~\ref{l.approx2}.

\begin{proposition}[Critical point identification, II]\label{p.crit_pt_2}
Suppose that $(\lambda_N)_{N\in\N}$ converges to some $\lambda_\infty \in \blacktriangle_\infty$ and that $\Ll(\bar F_{N,\lambda_N}\Rr)_{N\in\N}$ converges pointwise to some $f$ along a subsequence $(N_k)_{k\in\N}$. Let $t\geq 0$. If $f(t,\cdot)$ is Gateaux differentiable at some $q\in \mcl Q^\sS_{\infty,\uparrow}(\kappa)$, then
\begin{align}\label{e.d_qf=d_qpsi(...)}
    \partial_q f(t,q) = \partial_q \psi_{\lambda_\infty}(q+t\nabla\xi(\partial_q f(t,q))).
\end{align}
If, in addition, $t>0$, and $f(\cdot,q)$ is differentiable at $t$, then
\begin{align}\label{e.d_tf=intxi(d_qf)}
    \partial_t f(t,q) -\int_0^1\xi\Ll(\partial_qf(t,q)\Rr)=0.
\end{align}
\end{proposition}

\begin{proof}
Denote by $\seq$ the sequence $(N_k)_{k\in\N}$. Let $(N_n)_{n\in\N}$, $(x_n)_{n\in\N}$, $p$, and $a$ be given by Lemma~\ref{l.approx2} corresponding to $\seq$ and $(t,q)$ in the statement of this proposition. 

We first show~\eqref{e.d_tf=intxi(d_qf)}. We write $\tilde F_n = \tilde F^{\M_n,x_n}_{N_n,\lambda^\star_n}$ and $\la\cdot\ra^\circ_n=\la\cdot\ra^{\circ,\M_n,x_n}_{N_n,\lambda^\star_n}$ (appearing in Lemma~\ref{l.approx2}~\eqref{i.l.approx2.1}). The last statement of Lemma~\ref{l.approx2} ensures that $\tilde F_n$ converges to $f$ pointwise. Since $\la\cdot\ra^\circ_n$ is the Gibbs measure associated with $\tilde F_n$ (see~\eqref{e.tilde_F^x_N=} and~\eqref{e.<>^circ=}), a similar computation for~\eqref{e.def.der.FN} gives that, for every bounded continuous $\pi:[0,1]\to \prod_{s\in\sS}\S^{\kappa_s}$,
\begin{align}\label{e.<pi,d_q_tF_n>=...}
    \la \pi,\partial_q\tilde F_n(t,q)\ra_{L^2} = \E \la \pi(\alpha\wedge\alpha')\cdot R\ra^\circ_n,\qquad \partial_t \tilde F_n(t,q) = \E \la \xi(R)\ra^\circ_n
\end{align}
where we used the short hand $R = R_{N_n,\lambda^\star_n}(\sigma,\sigma')$. By Proposition~\ref{p.cvg_der_semi_concave} and Remark~\ref{r.cvg_der_semi_concave}, $\partial_q \tilde F_n(t,q)$ converges to $\partial_q f(t,q)$. Using this and Lemma~\ref{l.approx2}~\eqref{i.l.approx2.1}, we can send $n$ to infinity in the first relation of the above display to get
\begin{align*}
    \la \pi,\partial_qf(t,q)\ra_{L^2} =\E \la \pi(\alpha\wedge\alpha')\cdot p(\alpha\wedge\alpha')\ra_\fR \stackrel{\text{L.\ref{l.invar}}}{=} \la \pi,p\ra_{L^2}.
\end{align*}
Varying $\pi$, we get
\begin{align}\label{e.d_pf(t,q)=p}
    \partial_q f(t,q) = p.
\end{align}
Under the additional assumption of differentiability at $t$, we can use Proposition~\ref{p.cvg_der_semi_concave} and Remark~\ref{r.cvg_der_semi_concave} to get the convergence of $\partial_t \tilde F_n(t,q)$ to $\partial_t f(t,q)$. Sending $n$ to infinity in the second relation in~\eqref{e.<pi,d_q_tF_n>=...}, we get
\begin{align*}
    \partial_t f(t,q) = \E \la \xi(p(\alpha\wedge\alpha')\ra_\fR \stackrel{\text{L.\ref{l.invar}}}{=} \int_0^1 \xi(p(r))\d r
\end{align*}
which along with~\eqref{e.d_pf(t,q)=p} implies~\eqref{e.d_tf=intxi(d_qf)}.

Now, we write $\bar F_n= \bar F^{\M_n,x_n}_{N_n+M_n,\lambda^\star_n}$, which by Lemma~\ref{l.approx2} converges pointwise to $f$. Let $\la\cdot\ra_n=\la\cdot\ra^{\M_n,x_n}_{N_n+\M_n,\lambda^\star_n}$ be the Gibbs measure associated with $\bar F_n$ which appears in Lemma~\ref{l.approx2}~\eqref{i.l.approx2.2}.
Similar to~\eqref{e.def.der.FN}, we can compute
\begin{align*}
    \la \pi,\partial_q\bar F_n(t,q)\ra_{L^2} = \E \la \pi(\alpha\wedge\alpha')\cdot R_{N_n+\M_n,\lambda^\star_n}(\sigma,\sigma')\ra_n
    \\
    \stackrel{\text{L.\ref{l.last_vector}}}{=} \E \la \pi(\alpha\wedge\alpha')\cdot R^{\circ, \M_n}_{N_n+\M_n,\lambda^\star_n}(\sigma,\sigma')\ra_n
\end{align*}
for any bounded continuous $\pi :[0,1]\to \prod_{s\in\sS}\S^{\kappa_s}$, where Lemma~\ref{l.last_vector} is applied with $N_n+\M_n$ and $\M_n$ substituted for $N$ and $\M$ therein.
Again Proposition~\ref{p.cvg_der_semi_concave} and Remark~\ref{r.cvg_der_semi_concave} together give the convergence of $\partial_q \bar F_n(t,q)$ to $\partial_t f(t,q)$. Sending $n$ to infinity and applying Lemma~\ref{l.approx2}~\eqref{i.l.approx2.2}, we get
\begin{align*}
    \la \pi,\partial_q f(t,q)\ra_{L^2} = \la \pi, \partial_q\psi_{\lambda_\infty}(q+t\nabla\xi(p))\ra_{L^2}.
\end{align*}
Varying $\pi$ and inserting~\eqref{e.d_pf(t,q)=p}, we obtain~\eqref{e.d_qf=d_qpsi(...)}.
\end{proof}

\begin{proof}[Proof of Theorem~\ref{t.crit_pt_id}]
Let $p$ and $q'$ be given in the statement. The relation~\eqref{e.d_qf=d_qpsi(...)} in Proposition~\ref{p.crit_pt_2} implies that $(q',p)$ is a critical point defined in~\eqref{e.critical_rel}. The convergence in~\eqref{e.t.crit_pt_rep} follows from Proposition~\ref{p.crit_pt_1} and~\eqref{e.rel_parisi_mcJ}.
\end{proof}

The following corresponds to~\cite[Proposition~7.3]{HJ_critical_pts}.
\begin{proposition}[Critical point representation]\label{p.crit_pt_rep}
Suppose that $(\lambda_N)_{N\in\N}$ converges to some $\lambda_\infty \in \blacktriangle_\infty$ and that $\Ll(\bar F_N\Rr)_{N\in\N}$ converges pointwise to some $f$. For every $(t,q)\in\R_+\times \mcl Q^\sS_2(\kappa)$, there is $p\in \mcl Q^\sS_{\infty,\leq \lambda_\infty}(\kappa)$ such that
\begin{align}\label{e.p.crit_pt_rep}
    f(t,q)=\sP_{\lambda_\infty,t,q}\Ll(p\Rr),\qquad p = \partial_q \psi_{\lambda_\infty}(q+t\nabla\xi(p)).
\end{align}
\end{proposition}

The proof is based on Propositions~\ref{p.crit_pt_1} and~\ref{p.crit_pt_2} together with approximation arguments. We prefer to omit the detail since the proof is exactly the same as that for~\cite[Proposition~7.3]{HJ_critical_pts}. We only mention that Propositions~5.3, 5.4, 7.1, and 7.2 in~\cite{HJ_critical_pts} used in that proof correspond to Propositions~\ref{p.reg_limit}, \ref{p.cvg_der_semi_concave}, \ref{p.crit_pt_1}, and~\ref{p.crit_pt_2} here; and Corollary~5.2 there corresponds to Lemma~\ref{l.approx2} here.

\begin{proof}[Proof of Theorem~\ref{t.crit_pt_rep}]
Let $p$ be given in Proposition~\ref{p.crit_pt_rep} and set $q'= q+t\nabla\xi (p)$. The second relation in~\eqref{e.p.crit_pt_rep} ensures that $(q',p)$ is a critical point defined in~\eqref{e.critical_rel}. The first relation in~\eqref{e.p.crit_pt_rep} and~\eqref{e.rel_parisi_mcJ} yield~\eqref{e.t.crit_pt_rep}.
\end{proof}

Recall the Gibbs measure $\la\cdot\ra_{N,\lambda_N}$ from~\eqref{e.<>_N=} associated with the original free energy.
We consider the array of conditional overlaps:
\begin{align}\label{e.cond_overlap}
    R^{l,l'}_{N,\lambda_N,\sigma|\alpha} = \E \la R_{N,\lambda_N}(\sigma^l,\sigma^{l'})\,\Big|\, \alpha^l\wedge\alpha^{l'}\ra_{N,\lambda_N},\quad\forall l,l'\in \N
\end{align}
where the conditional expectation is taken with respect to $\E\la\cdot\ra_{N,\lambda_N}$ (not $\la\cdot\ra_{N,\lambda_N}$).
Also recall $\la\cdot\ra_\fR$ from Section~\ref{s.cascade}. Also, $\sigma|\alpha$ in the subscript of $R^{l,l'}_{N,\lambda_N,\sigma|\alpha}$ is purely symbolic to indicate the conditioning.

The next result is the version of~\cite[Proposition~7.4]{HJ_critical_pts} in the multi-species setting. Recall the Gibbs measure $\la\cdot\ra_{N,\lambda_N}$ as in~\eqref{e.<>_N=}.

\begin{proposition}[Convergence of conditional overlap]\label{p.cvg_cond_overlap}
Suppose that $(\bar F_{N,\lambda_N})_{N\in\N}$ converges pointwise to some $f$ along a subsequence $(N_k)_{k\in\N}$ and let $t\in\R_+$ (here, convergence of $(\lambda_N)_{N\in\N}$ is not assumed).
If $f(t,\cdot)$ is Gateaux differentiable at some $q\in \mcl Q^\sS_{\infty,\uparrow}(\kappa)$, then $\big(R^{l,l'}_{N,\lambda_N,\sigma|\alpha}\big)_{l,l'\in\N:l\neq l'}$ under $\E \la\cdot\ra_{N_k,\lambda_{N_k}}$ (at $(t,q)$) converges in law to $\big(p(\alpha^l\wedge\alpha^{l'})\big)_{l,l'\in\N:l\neq l'}$ under $\E\la\cdot\ra_\fR$, as $k$ tends to infinity, where $p=\partial_q f(t,q)$.
\end{proposition}

This result does not reply on Propositions~\ref{p.crit_pt_1} or~\ref{p.crit_pt_2}. The same proof for~\cite[Proposition~7.4]{HJ_critical_pts} works here. There, Propositions~4.8, 5.4 and display~(5.7) correspond to Lemma~\ref{l.invar}, Proposition~\ref{p.cvg_der_semi_concave}, and~\eqref{e.def.der.FN} here.

Before preceding, we extract a useful representation for $R^{l,l'}_{N,\lambda_N,\sigma|\alpha}$ (defined in~\eqref{e.cond_overlap}) from the proof of~\cite[Proposition~7.4]{HJ_critical_pts}. The following rephrases~\cite[(7.10)]{HJ_critical_pts}.

\begin{lemma}[Representation of conditional overlap]\label{l.rep_cond_overlap}
For every $N\in\N$, $\lambda_N\in \blacktriangle_N$, and $(t,q)\in\R_+\times \mcl Q^\sS_2(\kappa)$, we write $p_N=\partial_q \bar F_{N,\lambda_N}(t,q)$ as an element in $\mcl Q^\sS_\infty(\kappa)$. Then, we have
\begin{align*}
    R^{l,l'}_{N,\lambda_N,\sigma|\alpha} = p_N(\alpha^l\wedge\alpha^{l'})
\end{align*}
for every $l,l'\in\N$ with $l\neq l'$, a.s.\ under $\E\la\cdot\ra_{N,\lambda_N}$.
\end{lemma}

Next, we describe the convergence of the overlap when there is a small perturbation.

For $N\in\N$ and $\lambda_N\in \blacktriangle_N$, let $\hat H_N(\sigma)$ be the Hamiltonian with quadratic interaction that was introduced in~\eqref{e.hatH_N}.
Recall the definition of the Hamiltonian $H^{t,q}_N(\sigma)$ in~\eqref{e.H^t,q_N=}. For every $(t,\hat t,q)\in \R_+\times \R_+\times \mcl Q^\sS_\infty(\kappa)$, we consider
\begin{align}\label{e.def.FN.t.hat}
\begin{split}
    &\hat F_{N,\lambda_N}\Ll(t,\hat t, q\Rr)
    \\
    &= - \frac{1}{N}\log \iint  \exp\Ll( H^{t,q}_N(\sigma,\alpha)+ \sqrt{2\hat t}\hat H_N(\sigma) - \hat t N \Ll|R_{N,\lambda_N}(\sigma,\sigma)\Rr|^2\Rr)\d P_{N,\lambda_N}(\sigma)\d \fR(\alpha).
\end{split}
\end{align}
We denote the associated Gibbs measure by $\la\cdot\ra_{N,\lambda_N,\hat t}$, where we omit the dependence on $(t,q)$. 
Let $\lambda_\infty\in \blacktriangle_\infty$ and recall the functional $\hat{\mcl J}_{\lambda_\infty, t,\hat t, q}(q',p)$ defined for $q'\in \mcl Q^\sS_2(\kappa)$ and $p\in L^2([0,1],\prod_{s\in\sS}\S^{\kappa_s})$, which was introduced previously in~\eqref{e.hatJ=}.
The following result is an adaption of~\cite[Proposition~7.5]{HJ_critical_pts}. The most interesting is the third part.

\begin{proposition}[Convergence of overlap under perturbation]\label{p.cvg_overlap_pert}
Suppose that $(\lambda_N)_{N\in\N}$ converges to some $\lambda_\infty \in \blacktriangle_\infty$ and that $\Ll(\hat F_N\Rr)_{N\in\N}$ converges pointwise to some $f$ along a subsequence $(N_k)_{k\in\N}$. Then, for each $t\geq0$, the function $f(t,\cdot,\cdot) :\R_+\times \mcl Q^\sS_2(\kappa)\to\R$ is Gateaux differentiable (jointly in the two variables) on a subset of $\R_+\times \mcl Q^\sS_{\infty,\uparrow}(\kappa)$ that is dense in $\R_+\times \mcl Q^\sS_2(\kappa)$. Moreover, for every $\hat t\geq 0$ and every $q\in\mcl Q^\sS_{\infty,\uparrow}(\kappa)$ of Gateaux differentiability of $f(t,\hat t,\cdot)$, the following holds for $p=\partial _q f(t,\hat t,q)$ and $q'= q + t\nabla\xi(p)+2\hat t p$:
\begin{enumerate}
    \item $p=\partial_q \psi(q')$;
    \item if $(N_k)_{k\in\N}$ is the full sequence $(N)_{N\in\N}$, then
    \begin{align*}
        \lim_{N\to\infty}\hat F_{N,\lambda_N}\Ll(t,\hat t,q\Rr)= \hat{\mcl J}_{\lambda_\infty, t,\hat t, q}(q',p);
    \end{align*}
    \item if $\hat t>0$ and $f(t,\cdot,q)$ is differentiable at $\hat t$, then $\Ll(R_{N_k,\lambda_{N_k}}(\sigma^l,\sigma^{l'})\Rr)_{l,l'\in\N:l\neq l'}$ under $\E \la \cdot\ra_{N_k,\lambda_{N_k},\hat t}$ converges in law to $\Ll(p(\alpha^l\wedge\alpha^{l'})\Rr)_{l,l'\in\N:l\neq l'}$ under $\E\la\cdot\ra_\fR$, as $k$ tends to infinity.
\end{enumerate}
\end{proposition}

The proof is a straightforward adaption of that for~\cite[Proposition~7.5]{HJ_critical_pts}.
In that proof, Propositions~5.3, 7.1, 7.2, and 7.4 correspond to Propositions~\ref{p.reg_limit}, \ref{p.crit_pt_1}, \ref{p.crit_pt_2}, and~\ref{p.cvg_cond_overlap} here.

\begin{proof}[Proof of Theorem~\ref{t.main3}]
The theorem follows from Proposition~\ref{p.cvg_overlap_pert}.    
\end{proof}

The next result adapts~\cite[Proposition~1.5]{HJ_critical_pts} to the multi-species setting.
\begin{proposition}[Uniqueness of critical point at high temperature]\label{p.uniq_crit_pit_high_temp}
There exists $t_\mathrm{c}>0$ such that for every $t\in[0,t_\mathrm{c})$ and $q\in \mcl Q^\sS_2(\kappa)$, the function $\mcl J_{t,q}$ has a unique critical point in $\mcl Q^\sS_2(\kappa)\times \mcl Q^\sS_2(\kappa)$.
\end{proposition}

The proof is the same as that for~\cite[Proposition~1.5]{HJ_critical_pts}. The estimates in (5.23) and (5.21) used therein correspond to those in Lemma~\ref{l.psi^vec_smooth}, which should be used together with Lemma~\ref{l.reg_initial} here. 

The next result adapts~\cite[Proposition~1.6]{HJ_critical_pts}.

\begin{proposition}[Relevant critical points must be stable]\label{p.crit_pt_stable}
Assume that $(\lambda_N)_{N\in\N}$ converges to $\lambda_\infty$.
For each $n\in\N$, let $(t_n,q_n)\in\R_+\times \mcl Q^\sS_2(\kappa)$ and let $(q'_n,p_n)\in \mcl Q^\sS_2(\kappa)\times \mcl Q^\sS_2(\kappa)$ be a critical point of $\mcl J_{\lambda_\infty,t_n,q_n}$ such that
\begin{align*}
    \lim_{N\to\infty} \bar F_{N,\lambda_N}(t_n,q_n)= \mcl J_{\lambda_\infty, t_n,q_n}(q'_n,p_n).
\end{align*}
Suppose that $(t_n,q_n)$ converges towards $(t,q)\in\R_+\times \mcl Q^\sS_2(\kappa)$. Then, $(q'_n,p_n)_{n\in\N}$ is precompact in $Q^\sS_2(\kappa)\times \mcl Q^\sS_2(\kappa)$. Moreover, any subsequential limit $(q',p)\in Q^\sS_2(\kappa)\times \mcl Q^\sS_2(\kappa)$ is a critical point $\mcl J_{\lambda_\infty,t,q}$ and is such that
\begin{align*}
     \lim_{N\to\infty} \bar F_{N,\lambda_N}(t,q)= \mcl J_{\lambda_\infty, t,q}(q',p).
\end{align*}
\end{proposition}

Again, the proof is the same. Proposition~3.1, Lemma~3.4, and Corollary~5.2 used there correspond to Proposition~\ref{p.F_N_smooth} (Lipschitzness), Lemma~\ref{l.compact_embed_paths}, and Lemma~\ref{e.psi^vec=} (together with Lemma~\ref{l.reg_initial}).

\section{Results for convex models}
\label{s.convex}

We apply results in Section~\ref{s.crti_pt_and_rel_results} to the case where $\xi$ is convex. The results proved here correspond to those in~\cite[Section~8]{HJ_critical_pts}. In particular, we prove Theorems~\ref{t.convex} and~\ref{t.convex_2}.

Recall $\mcl K$ from~\eqref{e.mcl_K=} and $\sP_{\lambda,t,q}$ from~\eqref{e.sP_lambda,t,q}.
The following is the version of~\cite[Proposition~8.1]{HJ_critical_pts} for the multi-species setting.

\begin{proposition}[Parisi formula for enriched model]\label{p.mp-parisi}
Let $(\lambda_N)_{N\in\N}$ converge to some $\lambda_\infty \in\blacktriangle_\infty$. 
If $\xi$ is convex on $\prod_{s\in\sS}\S^{\kappa_s\times\kappa_s}_+$, then we have, for every $(t,q)\in\R_+\times \mcl Q^\sS_2(\kappa)$,
\begin{align}\label{e.mp_parisi}
    \lim_{N\to\infty} \bar F_{N,\lambda_N}(t,q) =\sup_{p\in \mcl Q^\sS(\kappa):\:\exists a \in \mcl K,\, a\geq p}\sP_{\lambda_\infty,t,q}(p) =\sup_{p\in \mcl Q^\sS_\infty(\kappa)}\sP_{\lambda_\infty,t,q}(p).
\end{align}
\end{proposition}

The relation $a\geq p$ under the supremum is understood as in~\eqref{e.a>p}.

\begin{proof}
We first prove this for rational $\lambda_\infty$. Suppose that there is $M\in\N$ such that $\lambda_\infty \in \blacktriangle_\M$. Then, we can use Corollary~\ref{c.equiv_rational} to match $\bar F_{N,\lambda_N}$ with $\bar F^\vec_N$ as described therein. Due to its definition in~\eqref{e.bxi=mp}, we can deduce that $\bxi$ is convex on $\S^\Delta_+$ from the convexity of $\xi$. By~\cite[Proposition~8.1]{HJ_critical_pts}, we have the Parisi formula: for every $(t,\bq)\in\R_+\times \mcl Q_2(\Delta)$,
\begin{align}\label{e.vector_parisi}
    \lim_{N\to\infty} \bar F^\vec_N(t,\bq) = \sup_{\bp\in \mcl Q(\Delta):\:\exists \ba \in \mcl K^\vec,\, \ba\geq\bp}\sP^\vec_{t,\bq}(\bp) =\sup_{\bp\in \mcl Q_\infty(\Delta)}\sP^\vec_{t,\bq}(\bp)
\end{align}
where $\sP^\vec_{t,\bq}(\bp)$ is given as in~\eqref{e.P^vec_t,q(bp)=} and $\mcl K^\vec = \overline{\mathrm{conv}}\Ll\{\bsigma\bsigma^\intercal:\: \bsigma\in \supp P^\vec_1\Rr\}$.
Now, fix any $(t,q)$ and let $\bq$ be given as in~\eqref{e.bq=mp}.
For every $\bq$ and $\ba$ appearing in~\eqref{e.vector_parisi}, let $p$ and $a$ be given as in~\eqref{e.p=(bp)} (without $\pm$).
Hence, the setup is the same as in the proof of Lemma~\ref{l.cavity1} (with $\lambda^\star$ therein substituted with $\lambda_\infty$.
By the argument in the paragraph below~\eqref{e.p=(bp)}, we have $a\geq p$ and $a\in \mcl K$. Also, we have $\sP^\vec_{t,q}(\bp) = \M \sP_{\lambda_\infty,t,q}(p)$ as verified in~\eqref{e.sP^vec_t,q=MsP}. This correspondence between $(a,p)$ and $(\ba,\bp)$ is bijective. Hence, by \eqref{e.vector_parisi} and Corollary~\ref{c.equiv_rational}, we get~\eqref{e.mp_parisi} when all entries of $\lambda_\infty$ are rational.

The general case follows from the rational case and a continuity argument.
Fix a sequence $(\lambda^n_\infty)_{n\in\N}$ in $\blacktriangle_\infty$ with rational entries such that this sequence converges to $\lambda_\infty$. For each $n$, fix any $(\lambda^n_N)_{N\in\N}$ converging to $\lambda^n_\infty$. 
For each $n$, Proposition~\ref{p.crit_pt_rep} gives $p_n \in \mcl Q^\sS_{\infty,\leq \lambda^n_\infty}(\kappa)$ such that
\begin{align}\label{e.F_n=P_n}
    \lim_{N\to\infty} \bar F_{N,\lambda^n_N}(t,q) = \sP_{\lambda^n_\infty,t,q}(p_n).
\end{align}
Notice that the sequence $(p_n)_{n\in\N}$ is bounded uniformly in $Q^\sS_\infty(\kappa)$. 
This along with the Lipschitzness of $\psi^\vec_\nu$ in Lemma~\ref{l.psi^vec_smooth} implies that $\psi^\vec_{\mu_s}((q+t\nabla\xi(p_n))_s)$ is bounded uniformly in $s$ and $n$. Hence, the definition of $\psi_\lambda$ in~\eqref{e.psi=} and that of $\sP_{\lambda,t,q}$ in~\eqref{e.sP_lambda,t,q} imply
\begin{align}\label{e.|P_n-P|<}
    \Ll|\sP_{\lambda^n_\infty,t,q}(p_n) - \sP_{\lambda_\infty,t,q}(p_n)\Rr|\leq C\Ll|\lambda^n_\infty -\lambda_\infty\Rr|,\quad\forall n\in\N,
\end{align}
for some constant $C$.
Also, Lemma~\ref{l.lambda_continuity} yields
\begin{align}\label{e.limsup|F_n-F|<}
    \limsup_{N\to\infty}\Ll|\bar F_{N,\lambda^n_N}(t,q)-\bar F_{N,\lambda^n_N}(t,q)\Rr| \leq C'\Ll|\lambda^n_\infty -\lambda_\infty\Rr|,\quad\forall n\in\N,
\end{align}
for some constant $C'$.
For $r,r'\geq \R$ and $\eps>0$, let us write $r\lesssim_\eps r'$ for $r\leq r'+\eps$. Then, for every $\eps$, we can find $p_\eps\in\mcl Q^\sS_\infty(\kappa)$ and $n\in\N$ sufficiently large such that
\begin{align*}
    \sup_{p\in \mcl Q^\sS_\infty(\kappa)}\sP_{\lambda_\infty,t,q}(p)\lesssim_\eps \sP_{\lambda_\infty,t,q}(p_\eps)\stackrel{\eqref{e.psi=},\eqref{e.sP_lambda,t,q}}{\lesssim_{\eps}}\sP_{\lambda^n_\infty,t,q}(p_\eps)
    \\
    \stackrel{\eqref{e.mp_parisi}}{\leq} \lim_{N\to\infty}\bar F_{N,\lambda^n_N}(t,q) \stackrel{\eqref{e.limsup|F_n-F|<}}{\lesssim_\eps} \liminf_{N\to\infty}\bar F_{N,\lambda_N}(t,q) 
\end{align*}
On the other hand, for every $\eps$, there is $n\in\N$ such that
\begin{align*}
    \limsup_{N\to\infty}\bar F_{N,\lambda_N}(t,q)\stackrel{\eqref{e.limsup|F_n-F|<}}{\lesssim_\eps} \lim_{N\to\infty}\bar F_{N,\lambda^n_N}(t,q) \stackrel{\eqref{e.F_n=P_n}}{=} \sP_{\lambda^n_\infty,t,q}(p_n)
    \\
    \stackrel{\eqref{e.|P_n-P|<}}{\lesssim_\eps}\sP_{\lambda_\infty,t,q}(p_n) \leq \sup_{p\in \mcl Q^\sS_\infty(\kappa)}\sP_{\lambda_\infty,t,q}(p).
\end{align*}
The above two displays yield one identity in~\eqref{e.mp_parisi} in the general case. The other identity in~\eqref{e.mp_parisi} can be deduced similarly.
\end{proof}

\begin{proof}[Proof of Theorem~\ref{t.convex}]
We can obtain~\eqref{e.t.convex} by using Proposition~\ref{p.mp-parisi} and the relation in~\eqref{e.rel_parisi_mcJ} between functionals.
\end{proof}

\begin{proof}[Proof of Theorem~\ref{t.convex_2}]

Recall the definition of $\theta$ in~\eqref{e.theta=}. 
Under the assumption that $\xi$ is convex on $\prod_{s\in\sS}\S^{\kappa_s}_+$, we can follow the same argument for (8.5) in the proof of~\cite[Proposition~8.1]{HJ_critical_pts} to get
\begin{gather}\label{e.inttheta=}
    \int_0^1\theta(p(r))\d r = \sup_{p'\in \mcl Q^\sS_\infty(\kappa)}\Ll\{\la \nabla \xi(p), p'\ra_{L^2}- \int_0^1\xi\Ll(p'(r)\Rr)\d r\Rr\},\quad\forall p \in \mcl Q^\sS_\infty(\kappa). 
\end{gather}
Recall the definition of $\sP_{\lambda_\infty,t,q}$ in~\eqref{e.sP_lambda,t,q}. Inserting~\eqref{e.inttheta=} to the right-hand side of~\eqref{e.mp_parisi}, we get
\begin{align*}
    &\limsup_{N\to\infty} \bar F_{N,\lambda_N}(t,q) \\
    &= \sup_{p\in \mcl Q^\sS_\infty(\kappa)} \inf_{p'\in \mcl Q^\sS_\infty(\kappa)}\Ll\{\psi_{\lambda_\infty}(q+t \nabla\xi(p))-\la t\nabla \xi(p), p'\ra_{L^2} + t\int_0^1\xi\Ll(p'(r)\Rr)\d r\Rr\}
    \\
    &\stackrel{\eqref{e.mcJ=}}{\leq} \sup_{q'\in q+ \mcl Q^\sS_\infty(\kappa)}\inf_{p\in \mcl Q^\sS_\infty(\kappa)}\mcl J_{\lambda_\infty,t,q}\Ll(q',p\Rr).
\end{align*}
On the other hand, \eqref{e.liminfF_N>} and~\eqref{e.f=hopf-lax} in Claim~\ref{c.free_energy_upper_bd} together yield
\begin{align}\label{e.liminfF>supinf}
    \liminf_{N\to\infty} \bar F_{N,\lambda_N}(t,q)\geq \sup_{q'\in q+ \mcl Q^\sS_\infty(\kappa)}\inf_{p\in \mcl Q^\sS_\infty(\kappa)}\mcl J_{\lambda_\infty,t,q}\Ll(q',p\Rr).
\end{align}
The two above displays then give~\eqref{e.t.convex_2}.
\end{proof}

Recall that in the statement of Theorem~\ref{t.convex_2}, we assumed Claim~\ref{c.free_energy_upper_bd}. The same claim also allows us to derive a different version of~\eqref{e.t.convex_2}, which adapts~\cite[Corollary~8.2]{HJ_critical_pts} to the current setting.
For every $a\in \prod_{s\in\sS} \R^{\kappa_s\times\kappa_s}$, we define
\begin{align}\label{e.xi^*=}
    \xi^*(a) =\sup_{b\in \prod_{s\in\sS}\S^{\kappa_s}_+}\Ll\{a\cdot b -\xi(b)\Rr\}.
\end{align}

\begin{corollary}[Alternative form of Hopf--Lax formula]
Assume that Claim~\ref{c.free_energy_upper_bd} is valid.
If $\xi$ is a convex function on $\prod_{s\in\sS}\S^{\kappa_s}_+$ and $(\lambda_N)_{N\in\N}$ converges to some $\lambda_\infty$, then for every $(t,q)\in \R_+\times \mcl Q^\sS_2(\kappa)$, we have
\begin{align}\label{e.hopf-lax_2}
    \lim_{N\to\infty} \bar F_{N,\lambda_N}(t,q) = \sup_{q'\in\mcl Q^\sS_\infty(\kappa)}\Ll\{\psi_{\lambda_\infty}(q+q')-t \int_0^1 \xi^*\Ll(t^{-1}q'\Rr)\Rr\}.
\end{align}
\end{corollary}

\begin{proof}
We fix any $(t,q)$ and denote the two sides in~\eqref{e.hopf-lax_2} by $\mathrm{LHS}$ and $\mathrm{RHS}$.
By the convexity of $\xi$ and the same argument as that for~\cite[(8.4)]{HJ_critical_pts}, it is easy to see $\theta(a)= \xi^*(\nabla(a))$ for every $a \in \prod_{s\in\sS}\S^{\kappa_s}_+$. Inserting this to the right-hand side in~\eqref{e.mp_parisi} and recalling the definition of $\sP_{\lambda_\infty,t,q}$ in~\eqref{e.sP_lambda,t,q}, we can get $\mathrm{LHS}\leq \mathrm{RHS}$.
On the other hand, by~\eqref{e.xi^*=}, we get
\begin{align*}
    \sup_{p\in\mcl Q^\sS_\infty(\kappa)}\Ll\{t^{-1}\la q',p\ra_{L^2}-\int_0^1\xi (p(r))\d r\Rr\}\leq \int_0^1 \xi^*\Ll(t^{-1} q'(r)\Rr)\d r.
\end{align*}
Inserting this to the right-hand side of~\eqref{e.liminfF>supinf} (which requires Claim~\ref{c.free_energy_upper_bd}) and using the definition of $\mcl J_{\lambda_\infty,t,q}$ in~\eqref{e.mcJ=} (together with changing $q'$ to $q+q'$), we get $\mathrm{LHS}\geq \mathrm{RHS}$.
\end{proof}

We can extract from Proposition~\ref{p.mp-parisi} a more familiar form of the Parisi formula.
For every $s\in\sS$, $\bq \in \mcl Q_\infty(\kappa_s)$, and $a\in \S^{\kappa_s}$, we define
\begin{align*}
    X_{s}(\bq,a) = \E \log \iint \exp\Ll(\bw^{\bq}(\alpha)\cdot \btau-\frac{1}{2}\bq(1)\cdot \btau\btau^\intercal + a\cdot \btau\btau^\intercal\Rr) \d \mu_s(\btau) \d \fR(\alpha)
\end{align*}
where $\bw^{\bq}$ is the $\R^{\kappa_s}$-valued process given as in~\eqref{e.E[bwbw]=}.
For every $\pi \in \mcl Q^\sS_\infty(\kappa)$ and $x\in \prod_{s\in\sS}\S^{\kappa_s}$, we define
\begin{align*}
    \sP_{\lambda_\infty}(\pi, x)= \sum_{s\in\sS}\lambda_{\infty,s} X_{s}((\nabla \xi(\pi))_s,x_s) + \frac{1}{2}\int_0^1\theta(\pi(r))\d r
\end{align*}
where $(\nabla \xi(\pi))_s$ is the $\S^{\kappa_s}_+$-valued path in the $s$-coordinate of $\nabla \xi(\pi)\in \mcl Q^\sS_\infty(\kappa)$.
Recall the Hamiltonian $H_N(\sigma)$ from~\eqref{e.EH_N(sigma)H_N(sigma')=}.
\begin{corollary}[Parisi formula for free energy with correction]

Let $(\lambda_N)_{N\to\infty}$ converge to some $\lambda_\infty\in\blacktriangle_\infty$.
If $\xi$ is convex on $\prod_{s\in\sS}\S^{\kappa_s}_+$, then,
\begin{align*}
    \lim_{N\to\infty}\frac{1}{N}\E \iint \exp\Ll(H_N(\sigma) - \frac{N}{2}\xi \Ll(R_{N,\lambda_N}(\sigma,\sigma)\Rr)\Rr) \d P_{N,\lambda_N}(\sigma)\d\fR(\alpha)
    =\inf_{\pi \in \mcl Q^\sS_\infty(\kappa)}\sP_{\lambda_\infty}(\pi,0).
\end{align*}
\end{corollary}
\begin{proof}
Due to~\eqref{e.E[bwbw]=}, we have $\bw^\bq \stackrel{\d}{=}\sqrt{2}\bw^{\bq/2}$.
By~\eqref{e.psi^vec=}, we can see $X_s(\bq,0) = -\psi^\vec_{\mu_s}(\bq/2)$, which by~\eqref{e.psi=} gives $\sum_{s\in\sS}\lambda_{\infty,s} X_{s}((\nabla \xi(\pi))_s,x_s) = - \psi_{\lambda_\infty}(\frac{1}{2}(\nabla \xi(\pi))$. Hence, we have $\sP_{\lambda_\infty}(\pi,0) = -\sP_{\lambda_\infty, \frac{1}{2},0}(\pi)$ given as in~\eqref{e.sP_lambda,t,q}. On the other hand, notice that expression after the limit on the left-hand side of the above display is equal to $-\bar F_{N,\lambda_N}(\frac{1}{2},0)$ (see~\eqref{e.F_N(t,q)=}). Therefore, this result follows from Proposition~\ref{p.mp-parisi} with $(\frac{1}{2},0)$ substituted for $(t,q)$ therein.
\end{proof}

This corollary corresponds to~\cite[Corollary~8.3]{HJ_critical_pts} and the next result to~\cite[Proposition~8.4]{HJ_critical_pts}. We can get the Parisi formula without the correction term $- \frac{N}{2}\xi \Ll(R_{N,\lambda_N}(\sigma,\sigma)\Rr)$.

\begin{proposition}[Parisi formula]
Let $(\lambda_N)_{N\to\infty}$ converge to some $\lambda_\infty\in\blacktriangle_\infty$.
If $\xi$ is convex on $\prod_{s\in\sS}\S^{\kappa_s}_+$, then,
\begin{align*}
    \lim_{N\to\infty}\frac{1}{N}\E \iint \exp\Ll(H_N(\sigma)\Rr) \d P_{N,\lambda_N}(\sigma)\d\fR(\alpha)
     =\sup_{y\in\S_+}\inf_{\pi \in \mcl Q^\sS_\infty(\kappa)}\Ll\{\sP_{\lambda_\infty}(\pi,y)-\frac{1}{2}\xi^*(2y)\Rr\}
    \\
     = \sup_{z\in\S_+} \inf_{\substack{y\in\S_+\\\pi \in \mcl Q^\sS_\infty(\kappa)}}\Ll\{\sP_{\lambda_\infty}(\pi,y)-y\cdot z+\frac{1}{2}\xi(z)\Rr\}
\end{align*}
where we used the shorthand $\S_+= \prod_{s\in\sS}\S^{\kappa_s}_+$.
\end{proposition}

\begin{proof}
The argument involves some Hamilton--Jacobi equation as in~\cite[Section~5]{mourrat2020extending} (also see~\cite[Section~5]{chen2023self}). One can follow the same steps as in the proof of~\cite[Proposition~8.4]{HJ_critical_pts}. The difference is that there the PDE is considered on $\R_+\times \S^\D_+$ for some $\D\in\N$ but here we need to adapt it to $\R_+\times \Ll(\prod_{s\in\sS}\S^{\kappa_s}_+\Rr)$. 
The only new ingredient is to show that the Hopf--Lax formula and the Hopf formula still hold on this domain. In the following, we explain how to prove this.

By~\cite[Propositions~6.3 and 6.4]{chen2022hamiltonCones}, it is sufficient to verify that the convex cone $\prod_{s\in\sS}\S^{\kappa_s}_+$ satisfies the \textit{Fenchel--Moreau property} as described in~\cite[Definition~6.1]{chen2022hamiltonCones}. Using a straightforward modification of the argument for $\S^\D_+$ in~\cite[Proposition~5.1]{chen2020fenchel}, we can verify that $\prod_{s\in\sS}\S^{\kappa_s}_+$ is also a \textit{perfect cone} described in \cite[Definition~2.1]{chen2020fenchel}. By~\cite[Corollary~2.3]{chen2020fenchel}, every perfect cone satisfies the Fenchel--Moreau property. Hence, the Hopf--Lax formula and the Hopf formula are valid in our setting.
\end{proof}

The next result adapts~\cite[Proposition~8.6]{HJ_critical_pts}.

\begin{proposition}[Differentiability of Parisi formula]\label{p.diff_parisi}
Let $(\lambda_N)_{N\to\infty}$ converge to some $\lambda_\infty\in\blacktriangle_\infty$, let $\xi$ be convex on $\prod_{s\in\sS} \S^{\kappa_s}_+$, and let $f$ be the pointwise limit of $(\bar F_{N,\lambda_N})_{N\in\N}$ given by Proposition~\ref{p.mp-parisi}.
\begin{itemize}
    \item For each $t\in\R_+$, the function $f(t,\cdot)$ is Gateaux differentiable everywhere on $\mcl Q^\sS_{\infty,\uparrow}(\kappa)$.
    \item The function $f$ is Gateaux differentiable everywhere on $(0,\infty)\times\mcl Q^\sS_{\infty,\uparrow}(\kappa)$.
\end{itemize}
\end{proposition}

A straightforward modification of the proof of~\cite[Proposition~8.6]{HJ_critical_pts} works here. Let us mention how to substitute results here for those in that proof. Lemma~6.4 corresponds to Lemma~\ref{l.F-F^x} here; Propositions~5.3, 5.4, 7.1, and 8.1 correspond to Propositions~\ref{p.reg_limit}, \ref{p.cvg_der_semi_concave}, \ref{p.crit_pt_1}, and~\ref{p.mp-parisi}; Corollary~5.2 and 6.11 correspond to Lemma~\ref{l.psi^vec_smooth} (to be used together with Lemma~\ref{l.reg_initial}) and Lemma~\ref{l.approx1}. Only Proposition~2.7 does not have a restatement here (which states that a Lipschitz function is differentiable ``almost everywhere'' in infinite dimensions), but it easily adapts to the setting here.

As in~\cite[Corollary~8.7]{HJ_critical_pts}, we can summarize the results in the convex case as follows.

\begin{corollary}\label{c.convex}
Let $(\lambda_N)_{N\to\infty}$ converge to some $\lambda_\infty\in\blacktriangle_\infty$ and let $\xi$ be convex on $\prod_{s\in\sS} \S^{\kappa_s}_+$.
Then, the sequence $\Ll(\bar F_{N,\lambda_N}\Rr)_{N\in\N}$ converges pointwise to some limit $f$ on $\R_+\times \mcl Q^\sS_2(\kappa)$.
At every $(t,q) \in (0,\infty)\times \mcl Q^\sS_{\infty,\uparrow}(\kappa)$, the function $f$ is Gateaux differentiable (jointly in its two variables) and satisfies
\begin{align}\label{e.pde_convex}
    \partial_t f(t,q) - \int_0^1 \xi\Ll(\dr_q f(t,q)\Rr) =0.
\end{align}
For every $t\in\R_+$, $f(t,\cdot)$ is Gateaux differentiable at every $q\in \mcl Q^\sS_{\infty,\uparrow}(\kappa)$ and the following holds for $p=\dr_q f(t,q)$ and $p_N= \dr_q \bar F_{N,\lambda_N}(t,q)$:
\begin{enumerate}
    \item \label{i.c.convex_1} $p_\infty\in \mcl Q^\sS_{\infty,\leq \lambda_\infty}(\kappa)$, $p_N\in \mcl Q^\sS_{\infty,\leq \lambda_N}(\kappa)$ for every $N\in\N$, and $(p_N)_{N\in\N}$ converges to $p$ in $L^r$ for every $r\in[1,\infty)$ as $N$ tends to infinity;
    \item \label{i.c.convex_2} $f(t,q) = \sP_{\lambda_\infty,t,q}(p)$ and $p=\dr_q \psi_{\lambda_\infty} (q+t\nabla \xi(p))$;
    \item \label{i.c.convex_3} $p_N(\alpha\wedge\alpha') = \E \la R_{N,\lambda_N}(\sigma,\sigma')\, |\, \alpha\wedge\alpha' \ra_{N,\lambda_N}$ almost surely under $\E \la\cdot\ra_{N,\lambda_N}$ for every $N$, and the overlap array $\big(p_N(\alpha^\ell\wedge\alpha^{\ell'})\big)_{\ell,\ell'\in\N:\:\ell\neq \ell'}$ under $\E \la\cdot\ra_{N,\lambda_N}$ converges in law to $\big(p(\alpha^\ell\wedge\alpha^{\ell'})\big)_{\ell,\ell'\in\N:\:\ell\neq \ell'}$ under $\E \la\cdot\ra_\fR$ as $N$ tends to infinity.
\end{enumerate}
\end{corollary}

\begin{proof}
The existence of $f$ is given by Proposition~\ref{p.mp-parisi}. The differentiability of $f$ follows from~Proposition~\ref{p.diff_parisi}. Proposition~\ref{p.crit_pt_2} yields~\eqref{e.pde_convex}. In Part~\eqref{i.c.convex_1}, the range for $p_N$ is due to~\eqref{e.bounds.der.FN} in Proposition~\ref{p.F_N_smooth}; the convergence of $(p_N)_{N\in\N}$ follows from Proposition~\ref{p.cvg_der_semi_concave}; the range for $p$ is a consequence of these two results (because we can extract a subsequence converging a.e.\ on $[0,1]$). Part~\eqref{i.c.convex_2} follows from Propositions~\ref{p.crit_pt_1} and~\ref{p.crit_pt_2}. Part~\eqref{i.c.convex_3} is due to Lemma~\ref{l.rep_cond_overlap} and Proposition~\ref{p.cvg_cond_overlap}.
\end{proof}

We can strengthen Part~\eqref{i.c.convex_3} to the convergence of the unconditioned overlap under the additional assumption that $\xi$ is strictly convex. The next result adapts~\cite[Proposition~8.8]{HJ_critical_pts} and the original proof can be modified easily (Propositions~4.8 and 5.4 used therein correspond to Lemma~\ref{l.invar} and Proposition~\ref{p.cvg_der_semi_concave}).

\begin{proposition}
Let $(\lambda_N)_{N\to\infty}$ converge to some $\lambda_\infty\in\blacktriangle_\infty$, let $\xi$ be strictly convex over $\prod_{s\in\sS}\R^{\kappa_s\times\kappa_s}$, let $t \in (0,\infty)$ and $q\in \mcl Q^\sS_{\infty,\uparrow}(\kappa)$, and let $p$ be as in Corollary~\ref{c.convex}. 
Then, the off-diagonal overlap array $\big(R_{N,\lambda_N}(\sigma^l,\sigma^{l'})\big)_{l,l'\in\N:\:l\neq l'}$ under $\E \la\cdot\ra_{N,\lambda_N}$ converges in law to $\big(p(\alpha^l\wedge\alpha^{l'})\big)_{l,l'\in\N:\:l\neq l'}$ under $\E\la\cdot\ra_\fR$ as $N$ tends to infinity.
\end{proposition}

\noindent
\textbf{Acknowledgments.}
The author thanks Jean-Christophe Mourrat for helpful discussions.

\noindent
\textbf{Funding.}
The author is funded by the Simons Foundation.

\noindent
\textbf{Data availability.}
No datasets were generated during this work.

\noindent
\textbf{Conflict of interests.}
The author has no conflicts of interest to declare.

\noindent
\textbf{Competing interests.}
The author has no competing interests to declare.

\small
\bibliographystyle{abbrv}

\begin{thebibliography}{10}

\bibitem{abarra}
E.~Agliari, A.~Barra, R.~Burioni, and A.~Di~Biasio.
\newblock Notes on the p-spin glass studied via {H}amilton-{J}acobi and
  smooth-cavity techniques.
\newblock {\em J. Math. Phys.}, 53(6):063304, 29, 2012.

\bibitem{aizenman1998stability}
M.~Aizenman and P.~Contucci.
\newblock On the stability of the quenched state in mean-field spin-glass
  models.
\newblock {\em J. Statist. Phys.}, 92(5-6):765--783, 1998.

\bibitem{alberici2020annealing}
D.~Alberici, A.~Barra, P.~Contucci, and E.~Mingione.
\newblock Annealing and replica-symmetry in deep {B}oltzmann machines.
\newblock {\em J. Stat. Phys.}, 180(1-6):665--677, 2020.

\bibitem{alberici2021deep}
D.~Alberici, P.~Contucci, and E.~Mingione.
\newblock Deep {B}oltzmann machines: rigorous results at arbitrary depth.
\newblock {\em Ann. Henri Poincar\'e}, 22(8):2619--2642, 2021.

\bibitem{auffinger2013complexity}
A.~Auffinger and G.~Ben~Arous.
\newblock Complexity of random smooth functions on the high-dimensional sphere.
\newblock {\em Ann. Probab.}, 41(6):4214--4247, 2013.

\bibitem{auffinger2013random}
A.~Auffinger, G.~Ben~Arous, and J.~\v{C}ern\'{y}.
\newblock Random matrices and complexity of spin glasses.
\newblock {\em Comm. Pure Appl. Math.}, 66(2):165--201, 2013.

\bibitem{baik2020free}
J.~Baik and J.~O. Lee.
\newblock Free energy of bipartite spherical {S}herrington-{K}irkpatrick model.
\newblock {\em Ann. Inst. Henri Poincar\'{e} Probab. Stat.}, 56(4):2897--2934,
  2020.

\bibitem{barbier2016}
J.~Barbier, M.~Dia, N.~Macris, F.~Krzakala, T.~Lesieur, and L.~Zdeborov\'{a}.
\newblock Mutual information for symmetric rank-one matrix estimation: A proof
  of the replica formula.
\newblock In {\em Advances in Neural Information Processing Systems (NIPS)},
  volume~29, pages 424--432, 2016.

\bibitem{barbier2019adaptive}
J.~Barbier and N.~Macris.
\newblock The adaptive interpolation method: a simple scheme to prove replica
  formulas in bayesian inference.
\newblock {\em Probability Theory and Related Fields}, 174(3-4):1133--1185,
  2019.

\bibitem{barbier2017layered}
J.~Barbier, N.~Macris, and L.~Miolane.
\newblock The layered structure of tensor estimation and its mutual
  information.
\newblock In {\em 2017 55th Annual Allerton Conference on Communication,
  Control, and Computing (Allerton)}, pages 1056--1063. IEEE, 2017.

\bibitem{barcon}
A.~Barra, P.~Contucci, E.~Mingione, and D.~Tantari.
\newblock Multi-species mean field spin glasses. {R}igorous results.
\newblock {\em Ann. Henri Poincar\'{e}}, 16(3):691--708, 2015.

\bibitem{barra2}
A.~Barra, G.~Del~Ferraro, and D.~Tantari.
\newblock Mean field spin glasses treated with {PDE} techniques.
\newblock {\em Eur. Phys. J. B}, 86(7):Art. 332, 10, 2013.

\bibitem{barra1}
A.~Barra, A.~Di~Biasio, and F.~Guerra.
\newblock {Replica symmetry breaking in mean-field spin glasses through the
  Hamilton--Jacobi technique}.
\newblock {\em Journal of Statistical Mechanics: Theory and Experiment},
  2010(09):P09006, 2010.

\bibitem{barra2014quantum}
A.~Barra, A.~Di~Lorenzo, F.~Guerra, and A.~Moro.
\newblock On quantum and relativistic mechanical analogues in mean-field spin
  models.
\newblock {\em Proceedings of the Royal Society A: Mathematical, Physical and
  Engineering Sciences}, 470(2172):20140589, 2014.

\bibitem{barramulti}
A.~Barra, G.~Genovese, and F.~Guerra.
\newblock Equilibrium statistical mechanics of bipartite spin systems.
\newblock {\em J. Phys. A}, 44(24):245002, 22, 2011.

\bibitem{bates2022crisanti}
E.~Bates and Y.~Sohn.
\newblock Crisanti-{S}ommers formula and simultaneous symmetry breaking in
  multi-species spherical spin glasses.
\newblock {\em Comm. Math. Phys.}, 394(3):1101--1152, 2022.

\bibitem{bates2022free}
E.~Bates and Y.~Sohn.
\newblock Free energy in multi-species mixed p-spin spherical models.
\newblock {\em Electronic Journal of Probability}, 27:1--75, 2022.

\bibitem{bates2023parisi}
E.~Bates and Y.~Sohn.
\newblock {P}arisi formula for balanced {P}otts spin glass.
\newblock {\em {Preprint, arXiv:2310.06745}}, 2023.

\bibitem{bauerschmidt2023stochastic}
R.~Bauerschmidt, T.~Bodineau, and B.~Dagallier.
\newblock Stochastic dynamics and the {P}olchinski equation: an introduction.
\newblock {\em {Preprint, arXiv:2307.07619}}, 2023.

\bibitem{benarous2022exponential}
G.~Ben~Arous, P.~Bourgade, and B.~McKenna.
\newblock Exponential growth of random determinants beyond invariance.
\newblock {\em Probab. Math. Phys.}, 3(4):731--789, 2022.

\bibitem{bra83}
J.~G. Brankov and V.~A. Zagrebnov.
\newblock On the description of the phase transition in the
  {H}usimi-{T}emperley model.
\newblock {\em J. Phys. A}, 16(10):2217--2224, 1983.

\bibitem{chen2022statistical}
H.~Chen, J.-C. Mourrat, and J.~Xia.
\newblock Statistical inference of finite-rank tensors.
\newblock {\em Ann. H. Lebesgue}, 5:1161--1189, 2022.

\bibitem{HB1}
H.-B. Chen.
\newblock Hamilton-{J}acobi equations for nonsymmetric matrix inference.
\newblock {\em Ann. Appl. Probab.}, 32(4):2540--2567, 2022.

\bibitem{chen2023parisi}
H.-B. Chen.
\newblock On {P}arisi measures of {P}otts spin glasses with correction.
\newblock {\em {Preprint, arXiv:2311.11699}}, 2023.

\bibitem{chen2023self}
H.-B. Chen.
\newblock Self-overlap correction simplifies the {P}arisi formula for vector
  spins.
\newblock {\em Electron. J. Probab.}, 28:Paper No. 170, 20, 2023.

\bibitem{HJ_critical_pts}
H.-B. Chen and J.-C. Mourrat.
\newblock On the free energy of vector spin glasses with non-convex
  interactions.
\newblock {\em arXiv preprint arXiv:2311.08980}, 2023.

\bibitem{simul_rsb}
H.-B. Chen and J.-C. Mourrat.
\newblock Simultaneous replica-symmetry breaking for vector spin glasses.
\newblock {\em in preparation}, 2024.

\bibitem{chen2020fenchel}
H.-B. Chen and J.~Xia.
\newblock {Fenchel--Moreau} identities on self-dual cones.
\newblock {\em arXiv preprint arXiv:2011.06979}, 2020.

\bibitem{chen2022hamiltonCones}
H.-B. Chen and J.~Xia.
\newblock {Hamilton--Jacobi equations with monotone nonlinearities on convex
  cones}.
\newblock {\em arXiv preprint arXiv:2206.12537}, 2022.

\bibitem{HBJ}
H.-B. Chen and J.~Xia.
\newblock Hamilton-{J}acobi equations for inference of matrix tensor products.
\newblock {\em Ann. Inst. Henri Poincar\'{e} Probab. Stat.}, 58(2):755--793,
  2022.

\bibitem{chen2022hamilton}
H.-B. Chen and J.~Xia.
\newblock {Hamilton-Jacobi equations from mean-field spin glasses}.
\newblock {\em {Preprint, arXiv:2201.12732}}, 2022.

\bibitem{chen2013aizenman}
W.-K. Chen.
\newblock The {A}izenman-{S}ims-{S}tarr scheme and {P}arisi formula for mixed
  {$p$}-spin spherical models.
\newblock {\em Electron. J. Probab.}, 18:no. 94, 14, 2013.

\bibitem{dey2021fluctuation}
P.~S. Dey and Q.~Wu.
\newblock Fluctuation results for multi-species {S}herrington-{K}irkpatrick
  model in the replica symmetric regime.
\newblock {\em J. Stat. Phys.}, 185(3):Paper No. 22, 40, 2021.

\bibitem{dominguez2022infinite}
T.~Dominguez and J.-C. Mourrat.
\newblock Infinite-dimensional {H}amilton--{J}acobi equations for statistical
  inference on sparse graphs.
\newblock {\em SIAM J. Math. Anal.}, 56(4):4530--4593, 2024.

\bibitem{dominguez2022mutual}
T.~Dominguez and J.-C. Mourrat.
\newblock Mutual information for the sparse stochastic block model.
\newblock {\em Ann. Probab.}, 52(2):434--501, 2024.

\bibitem{HJbook}
T.~Dominguez and J.-C. Mourrat.
\newblock {\em Statistical mechanics of mean-field disordered systems: a
  {H}amilton-{J}acobi approach}.
\newblock Zurich Lectures in Advanced Mathematics. EMS Press, Berlin, 2024.

\bibitem{fyodorov2004complexity}
Y.~V. Fyodorov.
\newblock Complexity of random energy landscapes, glass transition, and
  absolute value of the spectral determinant of random matrices.
\newblock {\em Phys. Rev. Lett.}, 92(24):240601, 2004.

\bibitem{fyo1}
Y.~V. Fyodorov, I.~Y. Korenblit, and E.~Shender.
\newblock Antiferromagnetic ising spin glass.
\newblock {\em J. Phys. C}, 20(12):1835, 1987.

\bibitem{fyo2}
Y.~V. Fyodorov, I.~Y. Korenblit, and E.~Shender.
\newblock Phase transitions in frustrated metamagnets.
\newblock {\em EPL (Europhysics Letters)}, 4(7):827, 1987.

\bibitem{genovese2023minimax}
G.~Genovese.
\newblock Minimax formula for the replica symmetric free energy of deep
  restricted {B}oltzmann machines.
\newblock {\em Ann. Appl. Probab.}, 33(3):2324--2341, 2023.

\bibitem{ghirlanda1998general}
S.~Ghirlanda and F.~Guerra.
\newblock General properties of overlap probability distributions in disordered
  spin systems. {T}owards {P}arisi ultrametricity.
\newblock {\em J. Phys. A}, 31(46):9149--9155, 1998.

\bibitem{guerra1996overlap}
F.~Guerra.
\newblock About the overlap distribution in mean field spin glass models.
\newblock {\em International Journal of Modern Physics B},
  10(13n14):1675--1684, 1996.

\bibitem{gue03}
F.~Guerra.
\newblock Broken replica symmetry bounds in the mean field spin glass model.
\newblock {\em Comm. Math. Phys.}, 233(1):1--12, 2003.

\bibitem{hartnett2018replica}
G.~S. Hartnett, E.~Parker, and E.~Geist.
\newblock Replica symmetry breaking in bipartite spin glasses and neural
  networks.
\newblock {\em Phys. Rev. E}, 98(2):022116, 2018.

\bibitem{huang2023strong}
B.~Huang and M.~Sellke.
\newblock Strong topological trivialization of multi-species spherical spin
  glasses.
\newblock {\em {Preprint, arXiv:2308.09677}}, 2023.

\bibitem{issa2024existence}
V.~Issa.
\newblock Existence and uniqueness of permutation-invariant optimizers for
  {P}arisi formula.
\newblock {\em {Preprint, arXiv:2407.13846}}, 2024.

\bibitem{kadmon2018statistical}
J.~Kadmon and S.~Ganguli.
\newblock Statistical mechanics of low-rank tensor decomposition.
\newblock In {\em Advances in Neural Information Processing Systems}, pages
  8201--8212, 2018.

\bibitem{kireeva2023breakdown}
A.~Kireeva and J.-C. Mourrat.
\newblock Breakdown of a concavity property of mutual information for
  non-{G}aussian channels.
\newblock {\em Inf. Inference}, 13(2):iaae008, 2024.

\bibitem{kivimae2023ground}
P.~Kivimae.
\newblock The ground state energy and concentration of complexity in spherical
  bipartite models.
\newblock {\em Comm. Math. Phys.}, 403(1):37--81, 2023.

\bibitem{ko2020free}
J.~Ko.
\newblock Free energy of multiple systems of spherical spin glasses with
  constrained overlaps.
\newblock {\em Electron. J. Probab.}, 25:Paper No. 28, 34, 2020.

\bibitem{korenblit1985spin}
I.~Y. Korenblit and E.~Shender.
\newblock Spin glass in an lsing two-sublattice magnet.
\newblock {\em Zh. Eksp. Teor. Fiz}, 89:1785--1795, 1985.

\bibitem{lelarge2019fundamental}
M.~Lelarge and L.~Miolane.
\newblock Fundamental limits of symmetric low-rank matrix estimation.
\newblock {\em Probability Theory and Related Fields}, 173(3-4):859--929, 2019.

\bibitem{lesieur2017statistical}
T.~Lesieur, L.~Miolane, M.~Lelarge, F.~Krzakala, and L.~Zdeborov{\'a}.
\newblock Statistical and computational phase transitions in spiked tensor
  estimation.
\newblock In {\em 2017 IEEE International Symposium on Information Theory
  (ISIT)}, pages 511--515. IEEE, 2017.

\bibitem{luneau2020high}
C.~Luneau, N.~Macris, and J.~Barbier.
\newblock High-dimensional rank-one nonsymmetric matrix decomposition: the
  spherical case.
\newblock In {\em 2020 IEEE International Symposium on Information Theory
  (ISIT)}, pages 2646--2651. IEEE, 2020.

\bibitem{mayya2019mutualIEEE}
V.~Mayya and G.~Reeves.
\newblock Mutual information in community detection with covariate information
  and correlated networks.
\newblock In {\em 2019 57th Annual Allerton Conference on Communication,
  Control, and Computing (Allerton)}, pages 602--607. IEEE, 2019.

\bibitem{mckenna2021complexity}
B.~McKenna.
\newblock Complexity of bipartite spherical spin glasses.
\newblock {\em Ann. Inst. Henri Poincar\'e{} Probab. Stat.}, 60(1):636--657,
  2024.

\bibitem{MPV}
M.~M{\'e}zard, G.~Parisi, and M.~Virasoro.
\newblock {\em Spin glass theory and beyond: an introduction to the replica
  method and its applications}, volume~9.
\newblock World Scientific Publishing Company, 1987.

\bibitem{miolane2017fundamental}
L.~Miolane.
\newblock Fundamental limits of low-rank matrix estimation: the non-symmetric
  case.
\newblock {\em \noop{2017}{Preprint, arXiv:1702.00473}}, 2017.

\bibitem{mourrat2020hamilton}
J.-C. Mourrat.
\newblock Hamilton-{J}acobi equations for finite-rank matrix inference.
\newblock {\em Ann. Appl. Probab.}, 30(5):2234--2260, 2020.

\bibitem{mourrat2021hamilton}
J.-C. Mourrat.
\newblock Hamilton-{J}acobi equations for mean-field disordered systems.
\newblock {\em Ann. H. Lebesgue}, 4:453--484, 2021.

\bibitem{mourrat2021nonconvex}
J.-C. Mourrat.
\newblock Nonconvex interactions in mean-field spin glasses.
\newblock {\em Probab. Math. Phys.}, 2(2):281--339, 2021.

\bibitem{mourrat2022parisi}
J.-C. Mourrat.
\newblock The {P}arisi formula is a {H}amilton-{J}acobi equation in
  {W}asserstein space.
\newblock {\em Canad. J. Math.}, 74(3):607--629, 2022.

\bibitem{mourrat2023free}
J.-C. Mourrat.
\newblock Free energy upper bound for mean-field vector spin glasses.
\newblock {\em Ann. Inst. Henri Poincar\'{e} Probab. Stat.}, 59(3):1143--1182,
  2023.

\bibitem{mourrat2020extending}
J.-C. Mourrat and D.~Panchenko.
\newblock Extending the {P}arisi formula along a {H}amilton-{J}acobi equation.
\newblock {\em Electronic Journal of Probability}, 25:Paper No. 23, 17, 2020.

\bibitem{new86}
C.~Newman.
\newblock Percolation theory: A selective survey of rigorous results.
\newblock In {\em Advances in multiphase flow and related problems}. SIAM,
  1986.

\bibitem{pan.aom}
D.~Panchenko.
\newblock The {P}arisi ultrametricity conjecture.
\newblock {\em Ann. of Math. (2)}, 177(1):383--393, 2013.

\bibitem{pan}
D.~Panchenko.
\newblock {\em The {S}herrington-{K}irkpatrick model}.
\newblock Springer Monographs in Mathematics. Springer, New York, 2013.

\bibitem{pan14}
D.~Panchenko.
\newblock The {P}arisi formula for mixed {$p$}-spin models.
\newblock {\em Ann. Probab.}, 42(3):946--958, 2014.

\bibitem{pan.multi}
D.~Panchenko.
\newblock The free energy in a multi-species {S}herrington-{K}irkpatrick model.
\newblock {\em Ann. Probab.}, 43(6):3494--3513, 2015.

\bibitem{pan.potts}
D.~Panchenko.
\newblock Free energy in the {P}otts spin glass.
\newblock {\em Ann. Probab.}, 46(2):829--864, 2018.

\bibitem{pan.vec}
D.~Panchenko.
\newblock Free energy in the mixed {$p$}-spin models with vector spins.
\newblock {\em Ann. Probab.}, 46(2):865--896, \noop{2019}2018.

\bibitem{panchenko2007overlap}
D.~Panchenko and M.~Talagrand.
\newblock On the overlap in the multiple spherical {SK} models.
\newblock {\em Ann. Probab.}, 35(6):2321--2355, 2007.

\bibitem{parisi1979infinite}
G.~Parisi.
\newblock Infinite number of order parameters for spin-glasses.
\newblock {\em Physical Review Letters}, 43(23):1754, 1979.

\bibitem{parisi1980order}
G.~Parisi.
\newblock The order parameter for spin glasses: a function on the interval 0-1.
\newblock {\em Journal of Physics A: Mathematical and General}, 13(3):1101,
  1980.

\bibitem{parisi1980sequence}
G.~Parisi.
\newblock A sequence of approximated solutions to the sk model for spin
  glasses.
\newblock {\em Journal of Physics A: Mathematical and General}, 13(4):L115,
  1980.

\bibitem{parisi1983order}
G.~Parisi.
\newblock Order parameter for spin-glasses.
\newblock {\em Physical Review Letters}, 50(24):1946, 1983.

\bibitem{reeves2020information}
G.~Reeves.
\newblock Information-theoretic limits for the matrix tensor product.
\newblock {\em IEEE Journal on Selected Areas in Information Theory},
  1(3):777--798, 2020.

\bibitem{reeves2019geometryIEEE}
G.~Reeves, V.~Mayya, and A.~Volfovsky.
\newblock The geometry of community detection via the mmse matrix.
\newblock In {\em 2019 IEEE International Symposium on Information Theory
  (ISIT)}, pages 400--404. IEEE, 2019.

\bibitem{sherrington1975solvable}
D.~Sherrington and S.~Kirkpatrick.
\newblock Solvable model of a spin-glass.
\newblock {\em Physical Review Letters}, 35(26):1792, 1975.

\bibitem{subag2017complexity}
E.~Subag.
\newblock The complexity of spherical {$p$}-spin models---a second moment
  approach.
\newblock {\em Ann. Probab.}, 45(5):3385--3450, 2017.

\bibitem{subag2021multi2}
E.~Subag.
\newblock {TAP approach for multi-species spherical spin glasses I: general
  theory}.
\newblock {\em Preprint, arXiv:2111.07132}, 2021.

\bibitem{subag2021free}
E.~Subag.
\newblock The free energy of spherical pure {$p$}-spin models: computation from
  the {TAP} approach.
\newblock {\em Probab. Theory Related Fields}, 186(3-4):715--734, 2023.

\bibitem{subag2021multi1}
E.~Subag.
\newblock On the second moment method and {RS} phase of multi-species spherical
  spin glasses.
\newblock {\em Electron. J. Probab.}, 28:Paper No. 50, 21, 2023.

\bibitem{subag2021multi3}
E.~Subag.
\newblock T{AP} approach for multispecies spherical spin glasses {II}: the free
  energy of the pure models.
\newblock {\em Ann. Probab.}, 51(3):1004--1024, 2023.

\bibitem{talagrand2006free}
M.~Talagrand.
\newblock Free energy of the spherical mean field model.
\newblock {\em Probab. Theory Related Fields}, 134(3):339--382, 2006.

\bibitem{Tpaper}
M.~Talagrand.
\newblock The {P}arisi formula.
\newblock {\em Ann. of Math. (2)}, 163(1):221--263, 2006.

\bibitem{Tbook1}
M.~Talagrand.
\newblock {\em Mean field models for spin glasses. {V}olume {I}}, volume~54 of
  {\em Ergebnisse der Mathematik und ihrer Grenzgebiete}.
\newblock Springer-Verlag, Berlin, 2011.

\bibitem{Tbook2}
M.~Talagrand.
\newblock {\em Mean field models for spin glasses. {V}olume {II}}, volume~55 of
  {\em Ergebnisse der Mathematik und ihrer Grenzgebiete}.
\newblock Springer, Heidelberg, 2011.

\end{thebibliography}
\newcommand{\noop}[1]{} \def\cprime{$'$}

\end{document}